\documentclass[aps,prx,twocolumn,nofootinbib,superscriptaddress,notitlepage]{revtex4-2}
\usepackage{braket}
\usepackage{graphicx}
\usepackage{amssymb}
\usepackage{amsmath,amsthm}
\usepackage{bm}

\usepackage{color}
\usepackage{physics}
\usepackage{mathtools}
\usepackage{bbm}
\usepackage{epstopdf}
\usepackage{hyperref}
\usepackage[capitalize,noabbrev]{cleveref}
\usepackage{qcircuit}
\usepackage{subfigure}
\usepackage{float}
\makeatletter
\let\newfloat\newfloat@ltx
\makeatother

\usepackage{algorithm}
\usepackage{algorithmic}
\usepackage{enumerate}
\usepackage{kantlipsum}
\usepackage{refcount}
\crefrangeformat{enumi}{#3#1#4--#5#2#6}

\usepackage{multirow}
\usepackage[english]{babel}

\usepackage{algorithm}
\usepackage{algorithmic}
\usepackage{xspace}
\usepackage{xcolor}

\usepackage{newfloat}
\usepackage{etoolbox}
\usepackage{dsfont}
\usepackage{booktabs}

\newtheorem{theorem}{Theorem}
\newtheorem{definition}[theorem]{Definition}

\newtheorem{lemma}[theorem]{Lemma}
\newtheorem{corollary}[theorem]{Corollary}

\newtheorem{problem}[theorem]{Problem}

\newcommand{\bP}{\mathbb{P}}

\newcommand{\bI}{\mathbb{I}}

\newcommand{\rd}{\mathrm{d}}

\newcommand{\bvec}[1]{\mathbf{#1}}
\newcommand{\mymathbb}[1]{\bm{#1}}

\newcommand{\vc}{\bvec{c}}

\newcommand{\vk}{\bvec{k}}

\newcommand{\vp}{\bvec{p}}
\newcommand{\vq}{\bvec{q}}

\newcommand{\diag}{\mathrm{diag}~}
\renewcommand{\Re}{\mathrm{Re}}
\renewcommand{\Im}{\mathrm{Im}}
\newcommand{\I}{\mathrm{i}}

\newcommand{\mc}[1]{\mathcal{#1}}
\newcommand{\mf}[1]{\mathfrak{#1}}
\newcommand{\wt}[1]{\widetilde{#1}}

\newcommand{\Or}{\mathcal{O}}

\newcommand{\NN}{\mathbb{N}}
\newcommand{\RR}{\mathbb{R}}
\newcommand{\CC}{\mathbb{C}}

\newcommand{\bPP}{\mathbb{P}}

\newcommand{\expt}[1]{\mathbb{E}\left( #1 \right)}
\newcommand{\expl}{\mathrm{exp}}

\renewcommand{\Tr}[1]{\text{tr}\left(#1\right)}

\newcommand{\myargmin}{\mathop{\mathrm{argmin}}}

\newcommand{\ceil}[1]{\left\lceil#1\right\rceil}
\newcommand{\spans}{\mathrm{span}}
\newcommand{\matrepwrt}[2]{\left[#1\right]_{#2}}

\newcommand{\fsim}{\text{FsimGate}\xspace}
\newcommand{\scp}[1]{{(#1)}}

\usepackage[normalem]{ulem}

\newcommand{\REV}[1]{{#1}}

\makeatletter
\def\@hangfrom@section#1#2#3{\@hangfrom{#1#2}#3}
\def\@hangfroms@section#1#2{#1#2}
\makeatother

\usepackage{array}

\usepackage{booktabs}
\newcolumntype{L}[1]{>{\raggedright\arraybackslash}p{#1}}

\newcommand{\SP}{{Supplementary Information}}
\newcommand{\SPN}{{Supplementary Note}\xspace}

\usepackage{titletoc}
  \titlecontents{subsection}
 [1.5em] %
 {\smallskip}
 {\thecontentslabel\hspace{1.02em}}
{\hspace*{2.32em}}
 {\,\,\titlerule*[0.77pc]{.}\contentspage}

\begin{document}
	\title{Optimal Low-Depth Quantum Signal-Processing Phase Estimation   }
	 
        \author{Yulong Dong}
        \email[Electronic address: ]{dongyl@berkeley.edu}
        \affiliation{Google Quantum AI, Venice, California 90291, USA.}
        \affiliation{Department of Mathematics, University of California, Berkeley, California 94720, USA.} 

        	\author{Jonathan A. Gross}
	\affiliation{Google Quantum AI, Venice, California 90291, USA.}
 
\author{Murphy Yuezhen Niu}
	\email[Electronic address: ]{murphyniu@ucsb.edu}
	\affiliation{Google Quantum AI, Venice, California 90291, USA.}
  \affiliation{Department of Computer Science, University of California, Santa Barbara, California, 93106, USA.}

	\begin{abstract}

Quantum effects like entanglement and coherent amplification can be used to drastically enhance the accuracy of quantum parameter estimation beyond classical limits. However, challenges such as decoherence and time-dependent errors hinder Heisenberg-limited amplification. We introduce Quantum Signal-Processing Phase Estimation algorithms that are robust against these challenges and achieve optimal performance as dictated by the Cram\'{e}r-Rao bound. These algorithms use quantum signal transformation to decouple interdependent phase parameters into largely orthogonal ones, ensuring that time-dependent errors in one do not compromise the accuracy of learning the other. Combining provably optimal classical estimation with near-optimal quantum circuit design, our approach achieves a standard deviation accuracy of $10^{-4}$ radians for estimating unwanted swap angles in superconducting two-qubit experiments, using low-depth ($<10$) circuits. This represents up to two orders of magnitude improvement over existing methods. Theoretically and numerically, we demonstrate the optimality of our algorithm against time-dependent phase errors, observing that the variance of the time-sensitive parameter $\varphi$ scales faster than the asymptotic Heisenberg scaling in the small-depth regime. Our results are rigorously validated against the quantum Fisher information, confirming our protocol's ability to achieve unmatched precision for two-qubit gate learning.

\end{abstract}

	\maketitle


\renewcommand{\figurename}{Fig.}
\renewcommand{\tablename}{Tab.}
\setcounter{algorithm}{1}
 
Quantum metrology's efficiency is fundamentally influenced by two critical factors: the Heisenberg limit, which defines how accuracy scales with quantum resources, and the coefficients of this scaling. While a variety of quantum metrology strategies~\cite{kimmel2015robust,neill_accurately_2021,arute_observation_2020} successfully adhere to the Heisenberg scaling, the real challenge lies in achieving or even addressing optimality in the scaling coefficients with realistic constraints. This aspect is particularly vital for applications in quantum error correction, where achieving fault-tolerant thresholds demands exceptionally high accuracy in quantum gate characterization. The necessity for deep circuitry, a significant hurdle in practical applications, stems directly from the lack of optimality in these scaling coefficients. This inefficiency is compounded by the challenges of finite coherence times and the amplification of drift errors from low-frequency noise or control fluctuations.  Therefore, current quantum metrology protocols, limited to accuracy levels between $10^{-2}$ and $10^{-3}$ radians~\cite{neill_accurately_2021,arute_observation_2020} for estimating gate angles, often fall short of the accuracy~($\sim 10^{-4}$) needed to verify the crossing of fault-tolerant error threshold for quantum error correction and other near-term quantum applications.
 
In quantum metrology for gate calibration, two primary approaches are used: robust phase estimation (RPE) and randomized benchmarking. RPE, along with its extensions like Floquet calibration, can achieve the Heisenberg limit under ideal conditions and are robust against state preparation and measurement (SPAM) errors across both single and multi-qubit gates~\cite{kimmel2015robust, neill_accurately_2021, arute_observation_2020}. However, its practical implementation is limited by the need for deep circuits and resource-intensive, iterative black-box optimizations to ensure accurate calibration. Moreover, its precision drops to between $10^{-1}$ and $10^{-2}$ radians when dealing with the time-dependent drifts common in superconducting qubit systems. Though recent progress \cite{NiLiYing2023,DingLin2023} refine the multiplicative overhead of  RPE cost, they focus on the asymptotic regime rather than physically short-depth and noise-robust implementation. Meanwhile, the randomized benchmarking approach, although general, forgoes Heisenberg scaling. It also requires extensive circuit depth to accurately estimate parameters and lacks sensitivity to coherent rotation errors~\cite{PhysRevLett.106.180504,PhysRevA.77.012307,PhysRevA.85.042311,GoogleQuantumSupremacy2019}. As a result, these prevalent quantum metrology techniques have not yet achieved optimal performance in practice for learning two-qubit gates. 

The effectiveness of a quantum metrology scheme can be assessed by the fundamental limits set by both classical and quantum Cram\'{e}r-Rao bounds~\cite{rao2008cramer,BraunsteinCaves1994}. To meet the classical Cram\'{e}r-Rao bounds, the inference subroutines that process measurements to estimate quantum gate parameters must be optimal. Similarly, to achieve the quantum Cram\'{e}r-Rao bounds, 
 the quantum measurement schemes, characterized by any Positive-Operator-Valued Measurements, must also be optimal~\cite{BraunsteinCaves1994}. Realizing optimality in both aspects requires refining classical post-processing techniques and the quantum circuit designs used in quantum gate calibration. In this work, we show that the RPE-based multi-parameter phase estimation method requires an additional phase-matching condition: the diagonal elements of the gates must share the same phase. If this condition is violated, the RPE-based method will fail to achieve both Heisenberg scaling and the classical Cram\'{e}r-Rao bounds when there's more than one phase to learn, even in the absence of quantum noise.

 In this work, we propose a metrology protocol that is, by design, robust against realistic time-dependent errors and only requires shallow~($<10$) circuits to achieve up to two orders of magnitudes of improvement over existing methods in the precision of gate-parameter estimates.

 \section*{Results}
We harness the analytical structure of a class of quantum-metrology circuits using a theoretical toolbox from classical signal processing \cite{Kay1989,ShenLiu2019}, Quantum Signal Processing~(QSP) \cite{LowChuang2017,GilyenSuLowEtAl2019,WangDongLin2021} and polynomial analysis \cite{Markov1890}.
QSP allows us to treat the inherent quantum dynamics as input quantum signals and perform universal transformations on the input to realize targeted quantum dynamics as output. Classical signal processing provides methods \REV{for} analyzing these transformed signals to produce robust estimations.
We propose a general gate model, which we term $U$-gate model, that encapsulates two-level invariant subspace structure in the native gate sets of superconducting, neutral atoms, and ion trap quantum computers. We parameterize the subspace of interest in our model $U$-gates with a set of angle parameters, and provide a metrology algorithm \REV{for} learning the swap angle $\theta$ and the phase difference $\varphi$.

Our metrology algorithm, which we term Quantum Signal-Processing Phase Estimation (QSPE), separates the estimation of the parameter-free from time-dependent errors~($\theta$) from that which is affected by time-dependent drift~($\varphi$).
Interestingly, the parameter $\varphi$ variance shrinks faster than Heisenberg scaling concerning circuit depth in the pre-asymptotic low-depth regime of experimental interest.
We analyze the stability of our protocol in the presence of realistic experimental noise and sampling errors. 
We prove that our method achieves the Cram\'{e}r-Rao lower bound in the presence of sampling errors and achieves up to $10^{-4}$ STD accuracy in learning swap angle $\theta$ in both simulation and experimental deployments on superconducting qubits.
We provide the evaluation of metrology protocol's quantum Fisher information (QFI) and show that our approach is a factor of two above the quantum CRLB (QCRLB). 
Furthermore, we demonstrate an interesting transition of the optimal metrology variance scaling as a function of circuit depth $d$ from the pre-asymptotic regime $d \ll 1/\theta$ to the Heisenberg limit $d \to \infty$.
	
We summarize the main results of our metrology algorithm and start by defining the metrology problem, the learning of a general $U$-gate, followed by an analysis of the QSP circuit with $U$-gates used in our algorithm. Building upon these closed-form results, we propose a phase estimation method combining Fourier analysis with QSP to separate the two gate parameters of interest in their functional forms. 
Our estimation algorithm enables fast and deterministic data post-processing using only direct linear algebra operations rather than iterative black-box optimizations used in multi-parameter robust phase estimation~\cite{kimmel2015robust,neill_accurately_2021}. Moreover, separating the inference of $\theta$ and $\varphi$ enhances the robustness of the phase estimation method against time-dependent errors that predominantly affect the gate parameter $\varphi$, which arise from the time-dependent qubit frequency noise~\cite{martinis2003decoherence,foxen2020,wudarski2023characterizing}. The analysis and modeling of Monte Carlo sampling error also indicate that our phase estimation method achieves the fundamental quantum metrology optimality in a practical regime against realistic errors for near-term devices. We also provide a comprehensive mathematical analysis of methods based on robust phase estimation \cite{arute_observation_2020,neill_accurately_2021}, and prove that the vulnerability of phase angle $\varphi$ to time-dependent errors ultimately renders the estimation accuracy of the swap angle $\theta$ exponentially worse than Heisenberg limit.
In addition, we proposed a noise-robust QSPE protocol that enables the estimation of gate parameters even when the gate angles fall outside the confidence regime for phase estimation.
Furthermore, we demonstrated QSPE experimentally on 34 superconducting qubits using the Google Quantum AI team's hardware, achieving $10^{-4}$ two-qubit phase estimation accuracy in practice, which improves by two orders of magnitude over standard previous methods.
Lastly, we include an empirical noise-robust QSPE protocol that enables the estimation of gate parameters even when the gate angles fall outside of the approximation regime for phase estimation.

\subsection*{General gate model with two-level system invariant subspace}
Our QSPE technique applies to any gate that contains a two-level invariant subspace $\mathcal{B}$, such that states within $\mathcal{B}$ remain within $\mathcal{B}$ when acted upon.
Here, we define a general two-level unitary model, which we term the $U$-gate model, around which we base our framework.
We parameterize this model gate when restricted to the subspace $\mathcal{B}$  as:

\begin{align}
    [U\left(\theta, \varphi, \chi, \psi\right)]_{\mathcal{B}} = \left(\begin{array}{cc}
       e^{-i\varphi - i\psi}\cos\theta  & -ie^{i\chi - i\psi}\sin\theta \\
       -ie^{-i\chi - i\psi}\sin\theta  & e^{i\varphi - i\psi}\cos\theta
    \end{array}\right)
\end{align}

We refer to $\theta$ as the swap angle, $\varphi$ as the phase difference, $\chi$ as the phase accumulation during the swap and $\psi$ as the global phase present in the entire gate. 
Note that $\chi$ cannot be amplified on the same basis as $\theta$ and $\varphi$. 

We emphasize that many quantum gates~(including single- and multi-qubit gates) can be   reduced to the $U$-gate model and thus be characterized by QSPE. For example, in our experimental deployment on the Google Quantum AI superconducting qubits, we study Fermionic Simulation Gates (\fsim{}s), which are native to superconducting qubit computers. 
We also remark that not all parameters in our model are necessary for every gate.

The problem of calibrating $U$-gate is to estimate $\theta$, $\varphi$, and $\chi$ for some invariant subspace of $U$ against realistic noise given access to the $U$-gate and basic quantum operations, which we now formalize:

\begin{problem}[Calibrating $U$-gate]\label{prob:qspc}
    Given access to an unknown $U$-gate, basic quantum gates and projective measurements, how to estimate gate parameters with bounded error and finite measurement samples?
\end{problem}

Previous metrology methods \cite{kimmel2015robust,neill_accurately_2021,arute_observation_2020} based on optimal measurements~\cite{BraunsteinCaves1994} for achieving the Heisenberg limit fall short of providing sufficient accuracy in $\theta$ when  $\theta \ll 1$.
Two significant factors limit these traditionally-regarded ``optimal'' metrology schemes.
First, the accuracy in $\theta$ depends on the amplification factor, i.e., maximum circuit depth.
The relatively low qubit coherence times of superconducting qubits render randomization-based quantum gate learning techniques \cite{PhysRevA.77.012307, PhysRevLett.106.180504,GoogleQuantumSupremacy2019} impractical due to their inefficient circuit depths. The finite low qubit coherence times~\cite{GoogleQuantumSupremacy2019} of superconducting qubits render randomization-based quantum gate learning techniques \cite{PhysRevA.77.012307, PhysRevLett.106.180504,GoogleQuantumSupremacy2019} impractical due to their inefficient circuit depths needed to achieve the desired accuracy close to surface code threshold~\cite{acharya2024quantum}.
Techniques based on robust phase estimation can require prohibitive depths to achieve a \REV{meaningful} full signal-to-noise ratio for small $\theta$ and   require iterative black-box optimizations for their estimators~\cite{neill_accurately_2021} instead of fast, deterministic post-processing for single-parameter phase estimation~\cite{kimmel2015robust}.
Second, time-dependent  unitary error in $\varphi$ is prevalent in architectures like Google's superconducting quantum computers~\cite{google2020hartree}, which invalidates basic assumptions in traditionally optimal and Heisenberg-limit-achieving metrology schemes.

 \subsection*{Quantum signal-processing phase estimation (QSPE)}

 In this work, we provide a low-depth phase estimation method for estimating the angles in some invariant subspace of an unknown $U$-gate when the swap angle is small, of order below $10^{-3}$, while facing  realistic time-dependent phase errors in $\varphi$. The phase estimation method leverages the structure of periodic circuits analyzed by classical and quantum signal processing and provides a framework to engineer quantum metrology from the perspective of universal quantum signal transformation. We call this type of metrology method Quantum Signal-Processing Phase Estimation (QSPE). Let $\omega$ be a tunable phase parameter and $\{\ket{0_\ell}, \ket{1_\ell}\}$ be the logical basis of the two-level space of interest. Then, QSPE measures the transition probabilities of the quantum circuits corresponding to logical Bell states $\ket{+_\ell} := \frac{1}{\sqrt{2}}\left( \ket{0_\ell} + \ket{1_\ell} \right)$ and $\ket{\I_\ell} := \frac{1}{\sqrt{2}}\left( \ket{0_\ell} + \I \ket{1_\ell} \right)$. The transition is measured with respect to the logical basis state $\ket{0_\ell}$. We depict the quantum circuit for QSPE in \cref{fig:qsp-pc-circuit} with an exemplified two-qubit $U$-gate for simplicity. In the example, the two-level subspace is set to the single-excitation subspace with basis $\ket{0_\ell} = \ket{01}$ and $\ket{1_\ell} = \ket{10}$. Then, the logical Bell state coincides with the conventional Bell state. We remark that the quantum circuit for QSPE can be generalized to multi-qubit cases following the recipe outlined in this paragraph. Details and the quantum circuit for QSPE in a general setup are provided in \SPN 2.
 The transition probability corresponding to the logical Bell state $\ket{+_\ell}$ is denoted as $p_X(\omega; \theta, \varphi, \chi)$, and that corresponding to the logical Bell state $\ket{\I_\ell}$ is denoted as $p_Y(\omega; \theta, \varphi, \chi)$.

      The measurement probabilities can be viewed as the expectation values of the logical Pauli operators:

	\begin{align} 
	    & \langle X_\ell \rangle(\omega; \theta, \varphi, \chi) = 2 p_X(\omega; \theta, \varphi, \chi) - 1,\\
            & \langle Y_\ell \rangle(\omega; \theta, \varphi, \chi) = 2 p_Y(\omega; \theta, \varphi, \chi) - 1,\\ 
	   & \mf{h}(\omega; \theta, \varphi, \chi) = \langle \frac{1}{2} \big( X_\ell + \I Y_\ell\big) \rangle(\omega; \theta, \varphi, \chi) = \langle a_\ell \rangle (\omega; \theta, \varphi, \chi).
	\end{align}
The physical meaning of the reconstructed function $\mf{h}(\omega; \theta, \varphi, \chi)$ coincides with the expected value of the logical annihilation operator which gauges the magnitude of the coherent rotation error in the single-excitation subspace. This observation qualitatively justifies the potential of the candidate function in the proposed phase estimation method.

As outlined in \SPN 3,
the reconstructed function derived from measurement probabilities admits an approximated expansion $\mf{h}(\omega; \theta, \varphi, \chi) = \sum_{-d+1}^{d-1} c_k e^{2\I k \omega}$ and when $d \theta \le 1$, the coefficients are
\begin{equation}
    c_k \approx \I \theta e^{-\I\chi} e^{-\I(2k+1)\varphi} \text{ with } k = 0, \cdots, d- 1.
\end{equation}

Due to Fourier expansion, sampling the reconstructed function on $(2d-1)$ distinct $\omega$ points is sufficient to characterize its information completely. For efficient processing with the Fast Fourier Transformation (FFT), we choose a set of $\omega$ points that are equally spaced. This choice of equally-spaced sampling points not only ensures numerical stability, as demonstrated in textbooks on numerical analysis \cite{trefethen2019approximation}, but also simplifies error analysis, as described in \SPN 4.
The second important consequence of this result is that the dependencies on $\theta$ and $\varphi$ are completely separated into the amplitude and the phase of the Fourier coefficients, respectively. The estimation problems of $\theta$ and $\varphi$ are then reduced to two independent linear regression problems. As $\chi$ is not considered, we apply a sequential phase difference to distill the angle $\varphi$:
\begin{equation}
    \boldsymbol{\Delta} = (\Delta_0, \cdots, \Delta_{d - 2})^\top, \Delta_k := \mathsf{phase}(c_k \overline{c_{k + 1}}) = 2 \varphi.
\end{equation}
Considering the Monte Carlo sampling error due to the finite number of measurements, we derive in \SPN 4
the linear-regression-based estimators of the relevant angles:
		\begin{equation}\label{eqn:main-qspcf-estimators}
		\hat{\theta} = \frac{1}{d} \sum_{k=0}^{d-1} \abs{c_k} \quad \text{ and } \quad  \hat{\varphi} = \frac{1}{2} \frac{\boldsymbol{\mymathbb{1}}^\top \mf{D}^{-1} \boldsymbol{\Delta}}{\boldsymbol{\mymathbb{1}}^\top \mf{D}^{-1} \boldsymbol{\mymathbb{1}}}.
		\end{equation}
  Here, $\boldsymbol{\mymathbb{1}}$ is an all-one vector and $\mf{D}$ is a $(d-1)$-by-$(d-1)$ constant tridiagonal matrix which coincides with the discrete Laplacian of a central finite difference form (see Definition 15 in \SPN 4 for more details). 
  The structure of $\mf{D}$ comes from differentiating experimental noises when applying sequential phase difference. To obtain $\hat{\chi}$, we defer the task to the metrology circuit in~\cite[Fig. S5]{arute_observation_2020} and do not use QSPE for the task. The main workflow of the QSPE is depicted in \cref{fig:main-qspc-main} and the Algorithm displayed in Box \ref{tab:alg:main-qsp-pc}.

 \subsection*{Classical and quantum optimality analysis}
	\REV{The performance of the statistical estimators is measured by their biases and variances.} In \SPN 4 B,
    we derive the performance of QSPE estimators with the following theorem by treating QSPE as linear statistical models. Furthermore, in \SPN 6,
    we show that QSPE estimators in \cref{eqn:main-qspcf-estimators} are optimal by saturating the Cram\'{e}r-Rao lower bound (CRLB) of the estimation problem.
    \begin{theorem}\label{prop:variance-qsp-pc-fourier}
    	In the regime $d \ll 1/\theta$, QSPE estimators in \cref{eqn:main-qspcf-estimators} are unbiased and with variances:
    	\begin{equation}
    		\begin{split}
    		& \mathrm{Var}(\hat{\theta}) \approx \frac{1}{8 d^2 M} \quad \text{ and } \quad \mathrm{Var}(\hat{\varphi}) \approx \frac{3}{8 d^4\theta^2 M}
    		\end{split}
    	\end{equation}
     where $M$ is the number of measurement shots in each experiment.
    \end{theorem}
    We note that the unbiasedness of these estimators holds up to a high-order bias, which is negligible in the target regime. For further details, please refer to \SPN 3.
    
  \paragraph{Comparison with Heisenberg limit.} According to the framework developed in Ref. \cite{Lloyd2006}, the variance of any quantum metrology is lower bound by the Heisenberg limit. It indicates that in our experimental setup when $d$ is large enough,  optimal variance scales as $1/(d^3 M)$. This seemingly contradicts \cref{prop:variance-qsp-pc-fourier}, where the variance of QSPE $\varphi$-estimator can achieve $1/(d^4 M)$. This counterintuitive conclusion is due to the pre-asymptotic regime $d \ll 1/\theta$. In \SPN 6,
  we analyze the CRLB of QSPE. The optimal variance given by CRLB is exactly solvable in the pre-asymptotic regime $d \ll 1/\theta$.

  The key reason behind such faster than Heisenberg limit scaling in pre-asymptotic regime depends on the unique structure of the QSPE circuit: the measurement outcome (Supplementary Equation (7))  concentrates around a constant value regardless of the gate parameter values. Yet when $d$ is large enough to pass to the asymptotic regime, measurement probabilities will take arbitrary values. Furthermore, the analysis of the CRLB suggests that the optimal asymptotic variance agrees with the Heisenberg limit when $d$ is \REV{large} enough. This nontrivial transition of optimal variance is theoretically analyzed and numerically justified in \SPN 6.
  We summarize this nontrivial transition of the optimal variance scaling of QSPE as a phase diagram in \cref{fig:crlb-qspe} a. To numerically demonstrate the transition, we compute the exact CRLB of QSPE when $\theta = 1\times 10^{-2}$ and $\theta = 1\times 10^{-3}$. In \cref{fig:crlb-qspe} b, the slope of the curve in log-log scale exhibits a clear transition before and after $d = 1/\theta$, which supports the phase diagram in \cref{fig:crlb-qspe} a. 
  Detailed theoretical and numerical discussions of the transition are carried out in \SPN 6.

\paragraph{Optimality analysis using Cram\'{e}r-Rao bounds.} Analyzing Cram\'{e}r-Rao bounds suggests the optimality of a quantum metrology protocol or the suboptimality leading to further improvement. Given an initialization and measurement, the optimality lies in the analysis of the classical CRLB, which investigates the most information one can retrieve from measurement probabilities. As outlined in \SPN 6 A,
the CRLBs are solvable when $d \ll 1 / \theta$, which exactly agree with the variance of our estimators derived in \cref{prop:variance-qsp-pc-fourier}. The optimality of our estimator is also validated from numerical simulation depicted in \cref{fig:crlb-qspe} b. Such agreement reveals the optimality of our data post-processing. Although linear-regression-based estimators are used, this linearization does not sacrifice the information retrieval in the experimental data. Furthermore, in contrast to other iterative inference methods, our estimators directly estimate angles using basic linear algebra operations, to which stability and fast processing are credited.

Despite the informative indication by analyzing CRLB, it cannot provide direct suggestions on improving initialization and measurement. Such generalization demands the switch to quantum Cram\'{e}r-Rao lower bound (QCRLB) which requires upper bounding the quantum Fisher information (QFI). As a quantum analog of the classical Fisher information, QFI lies in the center of quantum metrology by providing a fundamental lower bound on the accuracy one can infer from the system of a given resource limit. According to the analysis in \cite{KullGuerinVerstraete2020}, the QFI is an upper bound on the Fisher information over all possible measurements. For brevity, we only consider the inference of $\theta$ and hold all other unknown parameters constant in the analysis. However, our analysis can be generalized to the multiple parameter inference by adopting the multi-variable QFI in Ref. \cite{KullGuerinVerstraete2020}. In \SPN 6 D,
we derive that the average QFI is upper bounded by an integral:
\begin{equation}
    \mf{F}_\theta \le \frac{4}{\pi} \int_0^\pi \frac{\sin^2(d(\omega - \varphi))}{\sin^2((\omega - \varphi))} \rd \omega = 4 d.
\end{equation}
Here, the integrand gauges the information contained in an experiment with angle $\omega$. The integral stands for the use of equally spaced $\omega$ samples due to the absence of accurate information of $\varphi$. It is worth noting that the integrand is sharply peaked at $4 d^2$ when $\omega = \varphi$ which is also referred to as the phase-matching condition. The missing information of $\varphi$ downgrades the average QFI from $4d^2$ to $4 d$. However, as discussed in the following subsection, the lack of information about $\varphi$ in existing methods can potentially significantly degrade the Fisher information to $\Or(\log(d))$, thereby severely impeding the achievement of Heisenberg-limit scaling in estimation accuracy. Furthermore, this also suggests a finer estimation of $\theta$ when some rough information of $\varphi$ is provided either as a priori or from some preliminary estimation. This improvement is discussed in \SPN 4 C.
Consequently, the QCRLB of the QSPE formalism is
\begin{equation}
    \mathrm{Var}(\hat{\theta}) \ge \mathrm{QCRLB} \ge \frac{1}{16 d^2 M}.
\end{equation}
Compared with \cref{prop:variance-qsp-pc-fourier}, we see differentiation in a constant suboptimal factor of $2$, which is explainable. Note that we use two logical Bell states to perform experiments. The advantage is the experimental probabilities of these two experiments form a conjugate pair to reconstruct a complex function \REV{for the ease of analysis}. This complex function and its properties (see Theorems 7 and 9 in \SPN 3) 
eventually lead to a simple robust statistical estimator requiring only light computation. In contrast, the data generated from the initialization of one Bell state still contains full information on the parameters to be estimated. However, the highly nonlinear and oscillatory dependency renders the practical inference challenging. Hence, the factor of $2$ is due to the use of a pair of Bell states. Although the QFI indicates that inference variance can be lower by removing such redundancy in the initialization, the nature of ignoring practical ease makes it hard to achieve.

\subsection*{Advantage of QSPE over prior arts}

The key behind the success of QSPE is the isolation of $\theta$ and $\varphi$ estimations in Fourier space. This enables the robustness of individual angle estimation against the error and noise in another angle. A prior art that is widely used in the gate calibration in Google's superconducting platform is periodic calibration or Floquet calibration \cite{arute_observation_2020,neill_accurately_2021}. Periodic calibration measures the transition probability between tensor product states $\ket{10}$ and $\ket{01}$ of a periodic quantum circuit with $d$ $U$-gates inside. This differs from our QSPE method which initializes Bell states, though the main body of the quantum circuit is the same. To provide an estimation of angles, periodic calibration uses a black-box optimization to minimize the distance between the parametric ansatz and experimentally measured values. Periodic calibration is based on RPE \cite{kimmel2015robust, neill_accurately_2021, arute_observation_2020} and generalizes RPE to the estimation of multiple angles. Though RPE provably saturates the Heisenberg limit, the actual performance of periodic calibration highly relies on the satisfaction of the so-called phase-matching condition, namely, $\omega = \varphi$. Previous experiments suggest that the violation of the condition would lower the estimation accuracy of $\theta$ angle by a few magnitudes~\cite{neill_accurately_2021}. Because the phase angle $\varphi$ is vulnerable to time-dependent drift errors, the uncertainty of $\varphi$ ultimately ruins the estimation accuracy of the swap angle $\theta$ in periodic calibration. In \SPN 9,
we provide a comprehensive mathematical analysis of periodic calibration and prove that even without complex error and noise, the violation of phase-matching condition makes the estimation variance of $\theta$ scale as $1 / \log_2(d)$. This is exponentially worse than Heisenberg-limit scaling $1 / d^2$ when \REV{the} phase-matching condition is satisfied as depth increases. A formal statement can be found in Theorem 21
in \SPN 9. Moreover, the complex optimization landscape, detailed in \SPN 9 D,
renders the estimation using periodic calibration impractical. Though the periodic calibration with phase-matching, i.e. RPE, has higher Fisher information than QSPE, the hardness of satisfying \REV{the} phase-matching condition due to finite resource and time-dependent drift error renders the ideal high accuracy estimation challenging. In contrast, by averaging over $\omega$ points, our QSPE is more robust against error by separating $\theta$ and $\varphi$ estimation processes. This is also empirically justified using real data derived from quantum devices in \cref{fig:cz_calibrate_exp_qspc,fig:cz_calibrate_compare_var} in later sections.

 \subsection*{Robustness against realistic errors}

 We incorporate error mitigation against three different types of errors in our quantum-metrology routine, which we discuss separately below.
	
	\begin{enumerate}
	    \item Decoherence. Exploiting the analysis in the Fourier space provides a fruitful structure for mitigating decoherence. To illustrate, we propose a mitigation scheme for the globally depolarizing error in \SPN 7 A.
        Numerical simulation shows that the scheme can accurately mitigate the depolarizing error and can drastically improve the performance of QSPE estimators.
	    
	    \item Time-dependent noise. We numerically investigate the robustness of the QSPE estimators against realistic qubit frequency-drift error~\cite{wudarski2023characterizing} based on observation from real experiments. We show in \SPN 7 C
        that the QSPE estimators preserve their accuracy in the presence of this error.
	    
	    \item Readout errors.   In \SPN 7 D,
        we make an explicit resource estimation for sufficiently accurate mitigation of readout errors. 

        \item \REV{Initial state errors. In \SPN 7 E,
        we analyze the induced estimation error due to the error in the initial states. We demonstrate that these induced estimation errors are negligibly small in real experimental settings.} 
	\end{enumerate}
	We deploy these error mitigation techniques to realize QSPE on a real quantum device. The experimental results are given and discussed in the following section.

 \subsection*{Experimental deployment}
 In this section, we review the experimental deployment of our metrology method	and compare it against the leading alternative methods. We consider learning small swap angles in \fsim{}s, which are important for fermionic simulation and native to  Transmon superconducting qubits. \fsim{}s are two-qubit $U$-gates. The invariant subspace, referred to as single-excitation subspace, is spanned by single-excitation basis $\ket{01}$ and $\ket{10}$. CZ gate is a special \fsim{} with zero swap angle. Consequently, \fsim{}s with small swap angles model the imperfect production of CZ gates whose angle parameter estimation is crucial for applications of CZ gate including demonstrating surface code \cite{acharya2022suppressing}.  We refer to \SPN 1 C
 for details. 
 
We use the Google Quantum AI superconducting qubits~\cite{GoogleQuantumSupremacy2019} platform to conduct the experiments described in Algorithm displayed in Box \ref{tab:alg:main-qsp-pc} and \cref{fig:main-qspc-main}. We apply our QSPE method to calibrate $\theta$ and $\varphi$ angles of seventeen pairs of \REV{qubits on which} CZ gates \REV{act}.
Each CZ gate qubit pair is labeled by the coordinates $(x_1, y_1)$ and $(x_2, y_2)$ of both qubits on a two-dimensional-grid architecture. We plot the statistics of the learned gate angles in \cref{fig:cz_calibrate_exp_qspc}: the unwanted swap angle for most qubits are small, of order below $10^{-2}$. In comparison, periodic calibration yields unstable estimates with a highly variant standard deviation across different runs (see Supplementary Figure 2 in \SPN 1 A), a result of its sensitivity to time-dependent errors.

The performance advantage of QSPE over prior art lies in its robustness against time-dependent noise in the single-qubit phase $\varphi$.
In traditional methods, such as XEB and robust phase estimation \cite{neill_accurately_2021}, the measurement observables are nonlinear functions of both $\varphi$ and $\theta$, so if there is time-dependent drift in $\varphi$ during each experiment, or over different repetitions of the same experiment routine, the value of inferred $\theta$ will be directly affected~(see Supplementary Figure 2 in \SPN 1 A).
In comparison, QSPE is tolerant to realistic time-dependent error in $\varphi$ when estimating the swap angle $\theta$ due to the analytic separation between the two parameters through signal transformation, signal processing and Fourier analysis.

To validate the stability of QSPE method, we repeat the same phase estimation routine on each CZ gate pair over 10 independent repetitions.
This allows us to bootstrap the variance of the QSPE estimator on $\theta$ and $\varphi$.
We show the measured variance and mean of the $\theta$ and $\varphi$ estimates on seventeen CZ-gate pairs over different circuit depths $d$ used in QSPE in \cref{fig:cz_calibrate_exp_qspc}. It is important to note that we apply the technique discussed in \SPN 7 A to mitigate globally depolarizing errors using information from Fourier space. This error mitigation procedure estimates a globally depolarizing circuit fidelity $\alpha$ for each pair of qubits on which CZ gates act, as shown in the right-most column of \cref{fig:cz_calibrate_exp_qspc}. We observe that the circuit fidelity demonstrates a clear exponential decay with increasing circuit depth, which is consistent with our theoretical analysis \cite{BoixoIsakovSmelyanskiyEtAl2018}.
We show that on average the variance in $\theta$ estimates is around $10^{-7}$ for a depth-10 QSPE experiment.
This corresponds to $3\times 10^{-4}$ in STD, which is one to two orders of magnitude lower than the value of $\theta$ itself.
In comparison, we also performed the same set of experiments using XEB and compared the results to QSPE in \cref{fig:cz_calibrate_compare_var}.
The variance of $\theta$ inferred by XEB is of order $10^{-4}$~(three orders of magnitude larger than QSPE).
Consequently, we show that XEB and periodic calibration are insufficient to learn the value of $\theta$ in our experiments with a larger than unity signal-to-noise ratio.

\subsection*{Generalization of QSPE for an extended range of swap angles}
In earlier sections, we demonstrate the QSPE algorithm's effectiveness for estimating angles when the swap angle is of small magnitude. The actual use of QSPE is not limited to this parameter regime. In this subsection, we propose a generalization of QSPE for general swap angles. Theoretical analyses in \SPN 4 show that noise in Fourier space is significantly reduced, which consequently suggests the algorithm design using Fourier-transformed data. The key observation is that the exact expression of the amplitude of Fourier coefficients, which is referred to as $A_k(\theta)$, can be efficiently solved. Hence, given a set of experimental data, we can estimate the swap angle $\theta$ by aligning experimental amplitudes with theoretical expressions, effectively solving systems of nonlinear equations. When using $d$ $U$-gates per circuit, as discussed in \SPN 5, we outline an empirical noise-robust algorithm estimating $\theta$ to error $\epsilon$ using $\Or(d \log(d) \epsilon^{-1})$ classical operations. In the numerical results in \SPN 5, we demonstrate that the $\theta$-estimation error remains below $5 \times 10^{-4}$ for general $\theta$, even with only five $U$-gates. Thanks to Fourier transformation, the angle $\varphi$ is inferred from the phases of Fourier coefficients, decoupling its estimation from $\theta$. It also allows the use of the QSPE $\varphi$-estimator in \cref{eqn:main-qspcf-estimators} for varied swap angles, despite the signal-to-noise ratio varies with different swap angles. This is analyzed in the numerical results in \SPN 5.

 \subsection*{Learning quantum crosstalk with QSPE}

 An important  \REV{application} of  QSPE, thanks to its exceptional sensitivity in measuring small rotations, is in learning quantum crosstalk amplitudes. Quantum crosstalk errors arise from unintended quantum interactions between qubits, which become more problematic when gates operate together. These errors create correlations that are either spatial or temporal, posing a challenge to achieving fault-tolerant quantum computation. In the case of tunable superconducting transmon qubits~\cite{neill2017path}, interactions between two qubits are facilitated by a third ``coupler'' qubit placed between each pair. This setup allows for the two-qubit interaction to be controlled—turned on or off—by adjusting the coupler qubit's frequency. Yet, even with control over the coupler qubit's frequency, there remains  a non-zero amount of coupling between neighboring qubits' different levels. This coupling mimics the system's Bose-Hubbard coupling Hamiltonian: $H_{\text{crosstalk}}= g_{\text{crosstalk}}\left(\hat{a}_1\hat{a}_2^\dagger + \hat{a}_1^\dagger\hat{a}_2\right)$, where $\hat{a}_i$ denotes the bosonic annihilation operator for the $i$-th qubit.

Furthermore, the main effect of quantum crosstalk in qubit subspace can be described by a rotation within the single qubit subspaces $\text{span}\{\ket{10}, \ket{01}\}$ of the two qubits affected by crosstalk. Without any gate operation, the crosstalk's impact over a period $\Delta t$ follows the same pattern as shown in Supplementary Equation (4), where $\theta = g_{\text{crosstalk}} \Delta t$ depends on the crosstalk strength and the duration of crosstalk interaction. This insight allows for the learning of crosstalk effects using  QSPE  by substituting the $U$-gate in \cref{fig:qsp-pc-circuit} with an idle gate for an appropriate duration $\Delta t$, ensuring that $g_{\text{crosstalk}} \Delta t$ is sufficiently large to be measurable, yet not so large as to compromise the assumptions underlying QSPE. For example, by setting the circuit depth to $d=5$ and $\Delta t$ to 200 ns, the precision in measuring $g_{\text{crosstalk}}$ can reach around  10 MHz, markedly surpassing the accuracy of state-of-the-art results in similar systems, which are around 1 MHz~\cite{GoogleQuantumSupremacy2019}.

 \section*{Discussion}
 Our proposed QSPE estimators leverage the polynomial structure of periodic circuits through classical and quantum signal processing.
These analytics helped us to design an algorithm where the estimation of the swap angle $\theta$ is largely decoupled from that of the single-qubit phases $\varphi$ and $\chi$.
When some constant phase drift is imposed on the system, the inference is not affected, thanks to the robustness of the Fourier transformation and sequential phase difference.
We demonstrate such robustness against realistic errors, including drift errors in both numerical simulations and deployment on quantum devices. Such robustness is essential in achieving a record level of accuracy not demonstrated before in superconducting qubits.
An additional error-mitigation method against globally depolarizing error is further achieved here using the difference in the Fourier coefficients.

Prior to this work, error mitigation routines had been largely separated from quantum metrology protocols, preventing us from achieving the ultimate limit permitted by physics. Our successful combination of error mitigation with metrology hinges upon treating quantum metrology as a type of quantum signal processing: amplifying a given quantum signal while de-amplifying the unwanted experimental noise.
Our work, therefore, opens directions for using advanced quantum simulation techniques in the design of quantum metrology algorithms in order to achieve properties necessary for high-accuracy gate learning against realistic environmental noise.

Our future work aims to generalize the optimality of the metrology algorithm against more types of errors and schemes. First, our proof of QSPE estimators' optimality is based on analyzing the Monte-Carlo sampling error.
To fully optimize the design of statistical estimators against all types of error in addition to sampling errors requires modeling and studying the behavior and statistics of all types of dominant realistic error using tools from classical statistics, Bayesian inference, and machine learning.  Secondly, here we do not optimize all possible state initialization and measurement schemes.
Although theoretical analysis and numerical simulation prove that QSPE estimators are optimal in the given parameter regime and the given state preparation and measurement scheme, it remains an open question whether we can derive the optimal estimators in the most generic setting in \cref{prob:qspc} by optimizing circuit structure at initialization and measurement steps.
Thirdly, the QSPE estimators are designed for the low-depth regime~($d<10$). 
The restriction to this finite-depth setting is tied in with our main objective of mitigating the detrimental effect of time-dependent noise where deeper circuit depths introduce more drift and decoherence error. A related work employing a dynamical-decoupling-based scheme~\cite{gross_characterizing_2024} achieves similar accuracy; however, it is not effective in the low-depth limit of our algorithm and lacks optimality guarantees due to its reliance on decoupling sequences, necessitating deeper circuits.
This low-depth limit can be lifted if we can fully mitigate various noise effects that kick in at deeper depth. Furthermore, generalizing a deterministic estimator from our work to a variational one can also offer greater flexibility and optimality but requires a deeper understanding of the  landscape inherited from the QSPE structure. Lastly, qubitization techniques \cite{LowChuang2019,GilyenSuLowEtAl2019} and cosine-sine decomposition \cite{Dong2023Dissertation,TangTian2024} provide powerful two-dimensional subspace representations associated with any unitary matrix. These techniques could potentially generalize our estimation method to large systems with a large number of qubits, which will be our future work.

\section*{Methods}

\subsection*{Quantum experiments}

Quantum experiments in our work are conducted using the Google Quantum AI superconducting qubits~\cite{GoogleQuantumSupremacy2019} platform. Our QSPE method is detailed in the algorithm shown in Box \ref{tab:alg:main-qsp-pc} and \cref{fig:main-qspc-main}. The details of the XEB method, used in \cref{fig:cz_calibrate_compare_var}, are discussed in \SPN 1 A.

\subsection*{Numerical simulations}

All numerical tests are implemented in \textsf{python}. The numerical studies in \SPN 4 E and \SPN 7 C are conducted using \textsf{Cirq}. The study of our method's robustness against quantum errors in \SPN 7 C is performed using \textsf{Cirq}'s built-in noisy simulator. Other numerical studies, which do not directly involve quantum circuits, are performed using \textsf{numpy}.

 \bigskip

\noindent{\large \textbf{Data availability}}\\
All data presented in this work are visualized in the figures and tables within the main text and the Supplementary Information file.\\

\noindent{\large \textbf{Code availability}}\\
The codes that support the finding are available at \cite{Niu_Sinha_2023}. \\

\noindent {\large \textbf{Acknowledgments}}\\
	This work is partially supported by the NSF Quantum Leap Challenge Institute (QLCI) program through grant number OMA-2016245 (Y.D.). The authors thank discussions with Ryan Babbush, Connor Clayton, Zhiyan Ding, Zhang Jiang, Lin Lin, Shi Jie Samuel Tan, Vadim Smelyanskiy and K. Birgitta Whaley. The authors thank Rajarishi Sinha from Google Cloud AI for his help in open-sourcing the code supporting this finding.
\\

\noindent {\large \textbf{Author contributions}}\\
 M.N. conceived the project, carried out some of the theoretical analysis, and led and coordinated the project. Y.D. contributed to the original idea and carried out the theoretical analysis and numerical simulation to support the study. M.N. and J.G. discussed and carried out experiments on real quantum devices to support the study. All authors contributed to the writing of the manuscript.\\

\noindent{\large \textbf{Competing interests}}\\
\noindent The authors declare no competing interests.

\renewcommand{\tablename}{Box}
\renewcommand{\thetable}{\arabic{table}}
\begin{table}[htbp]
  \renewcommand{\arraystretch}{1.2}
  \begin{tabular}{*{1}{@{}L{8.5cm}}}
    \toprule
    {\bfseries Algorithm \thealgorithm} \quad Inferring unknown angles in $U$-gate with small swap angle using QSPE \tabularnewline
    \bottomrule
    \textbf{Input:} A $U$-gate $U(\theta,\varphi,\chi,\psi)$, an integer $d$ (the number of applications of the $U$-gate). \tabularnewline
    Initiate a complex-valued data vector $\boldsymbol{\mf{h}} \in \CC^{2d-1}$. \tabularnewline
    \textbf{for} $j = 0, 1, \cdots, 2d-2$ \textbf{do} \tabularnewline
    \quad Set the tunable $Z$-phase angle as $\omega_j = \frac{j}{2d-1} \pi$.  \tabularnewline
    \quad Perform the quantum circuit in \cref{fig:qsp-pc-circuit} and measure the transition probabilities $p_X(\omega_j)$ and $p_Y(\omega_j)$.  \tabularnewline
    \quad Set $\mf{h}_j \leftarrow p_X(\omega_j) - \frac{1}{2} + \I\left(p_Y(\omega_j) - \frac{1}{2}\right)$.  \tabularnewline
    \textbf{end for}  \tabularnewline
    Compute the Fourier coefficients $\mathbf{c} = \mathsf{FFT}\left(\boldsymbol{\mf{h}}\right)$.  \tabularnewline
    Compute estimates $\hat{\theta}$ and $\hat{\varphi}$ according to \cref{eqn:main-qspcf-estimators}. \tabularnewline
    \textbf{Output:} Estimators $\hat{\theta}, \hat{\varphi}$.
    \\ \bottomrule
    \end{tabular}
    \caption{Algorithm for inferring unknown angles in $U$-gate with small swap angle using QSPE}
    \label{tab:alg:main-qsp-pc}
\end{table}

\widetext
\clearpage
\newpage

	\begin{figure}[H]
 \includegraphics[width=\linewidth]{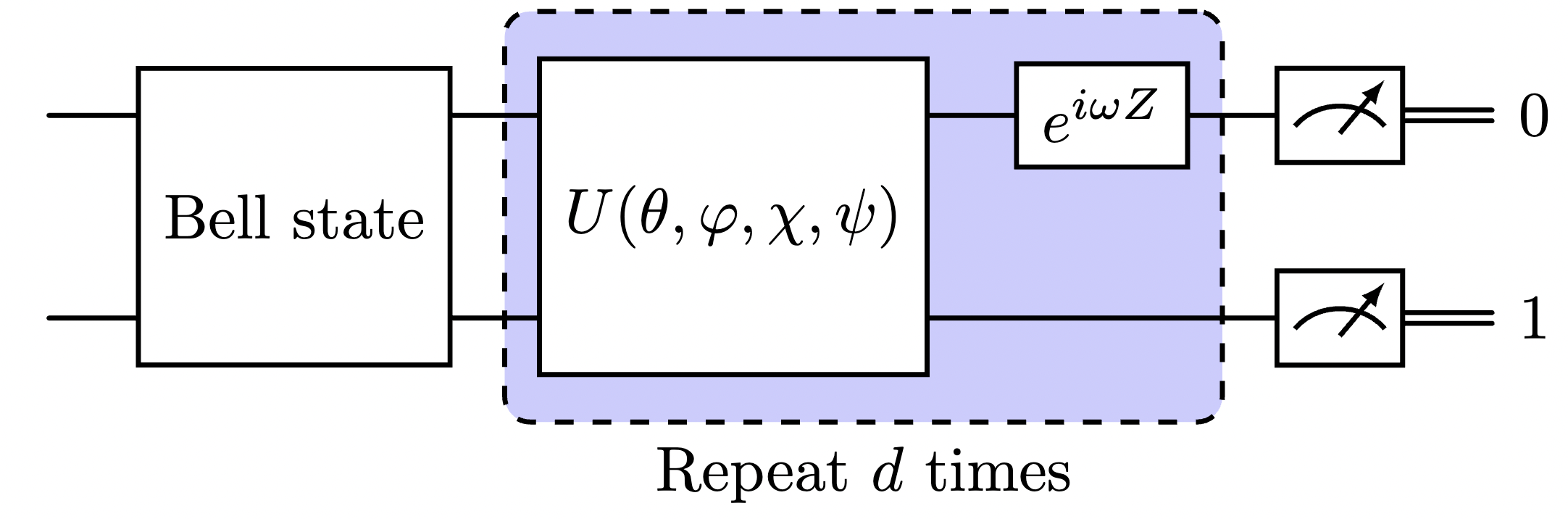}
		\caption{Quantum circuit for QSPE with an exemplified two-qubit $U$-gate. The input quantum state is prepared to be Bell state in either $\ket{+_\ell}$ or $\ket{\I_\ell}$ according to the type of experiment. The quantum circuit enjoys a periodic structure of the unknown $U$-gate and a tunable $Z$ rotation.\label{fig:qsp-pc-circuit}
  }
	\end{figure}

    	\begin{figure*}[htbp]
		\centering
		\includegraphics[width= \textwidth]{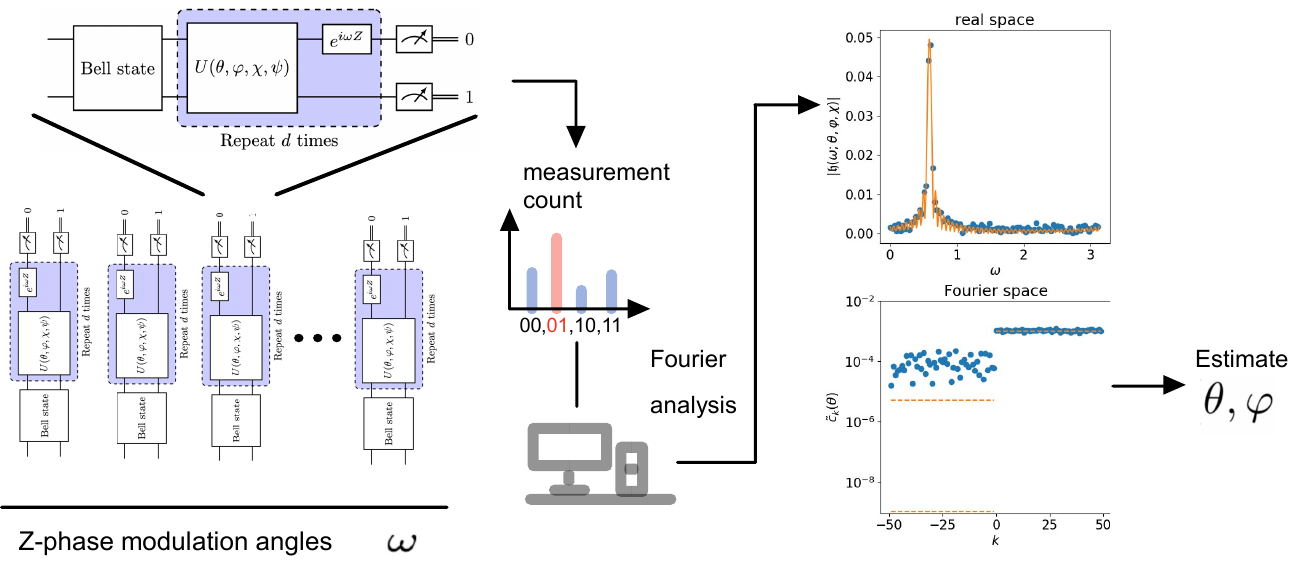}
		\caption{Flowchart of main procedures in QSPE. \REV{The experimental data are collected from depth $d$ quantum circuit experiments featuring equally-spaced phase modulation angles $\omega$, as shown in the left panels. Probabilities from each experiment of different phase modulations are analyzed using Fourier transformation. As illustrated in the right panels, the Fourier-space data are better structured compared to real-space data. Gate angles are then derived using our QSPE estimators.}}
		\label{fig:main-qspc-main}
	\end{figure*}

    \begin{figure}[htbp]
        \centering
        \includegraphics[width=\columnwidth]{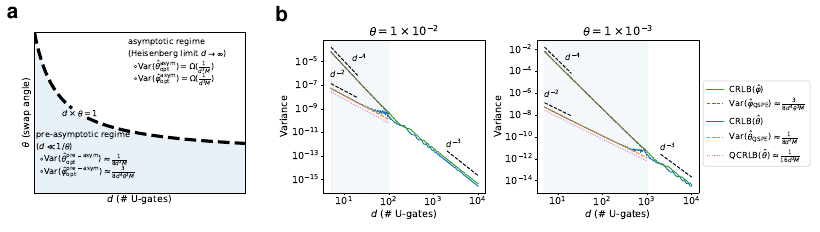}
        \caption{A nontrivial transition of the optimal variance in solving QSPE. The theoretical analysis of the transition is in \SPN 6. \textbf{a} Phase diagram showing the nontrivial transition of the optimal variance in solving QSPE. QSPE estimators attain the optimal variance in the pre-asymptotic regime. \textbf{b} Cram\'{e}r-Rao lower bound (CRLB) and the theoretically derived estimation variance. The single-qubit phases are set to $\varphi = \pi/16$ and $\chi = 5\pi/32$. The number of measurement samples is set to $M = 1\times10^5$.}
        \label{fig:crlb-qspe}
    \end{figure}

	\begin{figure}[htbp]
		\centering
		\includegraphics[width=\columnwidth]{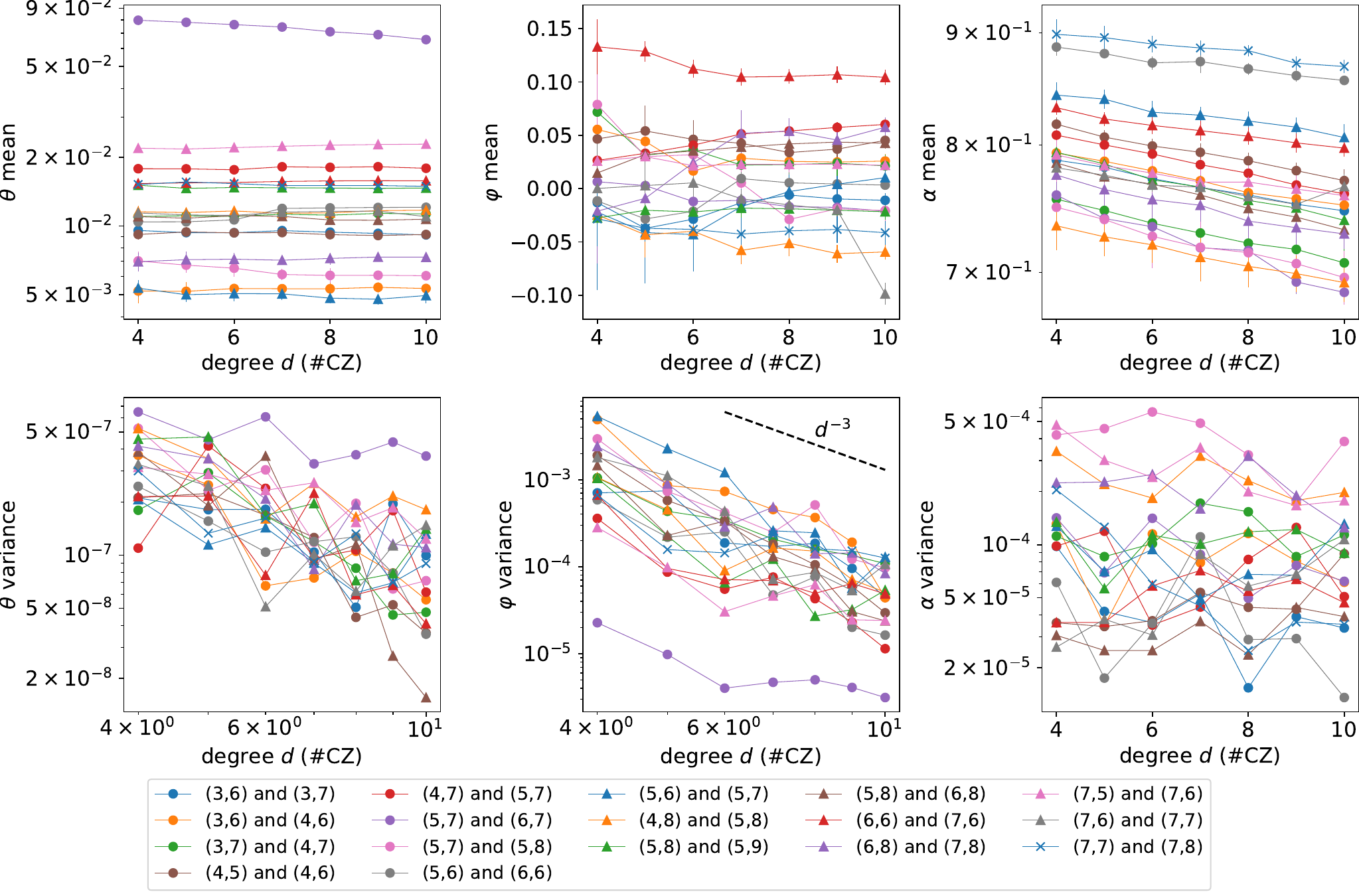}
		\caption{Learning CZ gate with small unwanted swap angle. Each data point is the average of $10$ independent repetitions and the error bars in the top panels stand for the standard deviation across those repetitions. The number of measurement samples is set to $M = 1\times 10^4$. \REV{These columns display the estimated values of gate angles $\theta, \varphi$, and circuit fidelity $\alpha$.}}
		\label{fig:cz_calibrate_exp_qspc}
	\end{figure}
	
	\begin{figure}[htbp]
		\centering
		\includegraphics[width=\columnwidth]{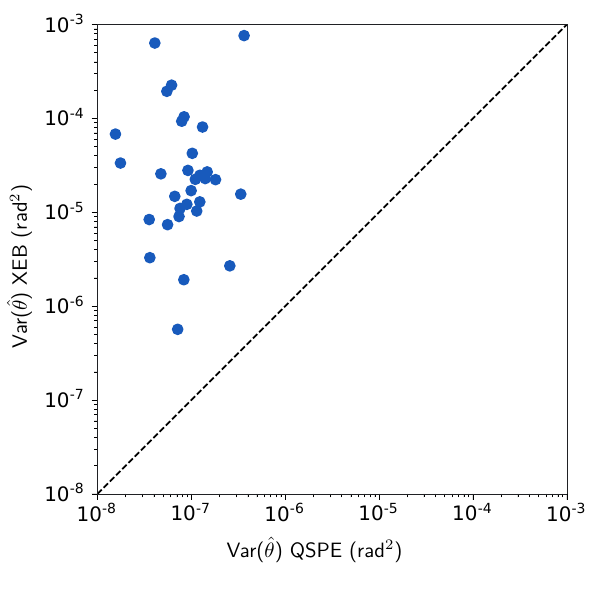}
		\caption{Comparison of the variance in learning swap angle $\theta$ of CZ gates over seventeen pairs of qubits between QSPE and XEB each repeated for 10 times.}
		\label{fig:cz_calibrate_compare_var}
	\end{figure}

\newpage
\clearpage
\onecolumngrid
    
\begin{center}
    {\bf{\Large Supplementary Information for Optimal Low-Depth Quantum Signal-Processing Phase Estimation}}
\end{center}

\setcounter{equation}{0}
\setcounter{figure}{0}
\setcounter{table}{0}
\setcounter{theorem}{0}
\setcounter{algorithm}{0}
\renewcommand{\figurename}{Supplementary Figure}
\renewcommand{\tablename}{Supplementary Table}
\renewcommand{\thetable}{\arabic{table}}
\renewcommand{\thesection}{\SP\ \arabic{section}}

	\section{Preliminaries}
 
	\subsection{Prior art in quantum gate calibration}
	\label{sec:prior-art}

    A widely used gate calibration technique is called Periodic or Floquet calibration~\cite{neill_accurately_2021,arute_observation_2020}, which is an extension of robust phase estimation~\cite{kimmel2015robust} to multi-parameter regime.    It leverages the excitation-preserving structure of the \fsim to measure the parameters using a restricted set of circuits (compared to full-process tomography).
	This technique amplifies unitary errors in the gate through repeated applications between measurements. When phase-matching condition is attained, it leads to variance in the estimated parameters that scales inversely with the square of the number of gate applications, thus achieving the Heisenberg limit. In \cite{neill_accurately_2021}, two calibration circuit types are discussed, each utilizing different initialization and measurement techniques. They are referred to as the phase method and the population method. The phase method is able to calibrate the phase accumulation angle. However, given the simpler parametric expression, the population method is more frequently used for calibrating swap angles and phase differences \cite{arute_observation_2020, neill_accurately_2021}. Consequently, in this paper, we focus on the population method for periodic calibration. The quantum circuit is depicted in \cref{fig:population-method-circuit}. 

    \begin{figure}[H]
         \centering
         \includegraphics[width=.4\textwidth]{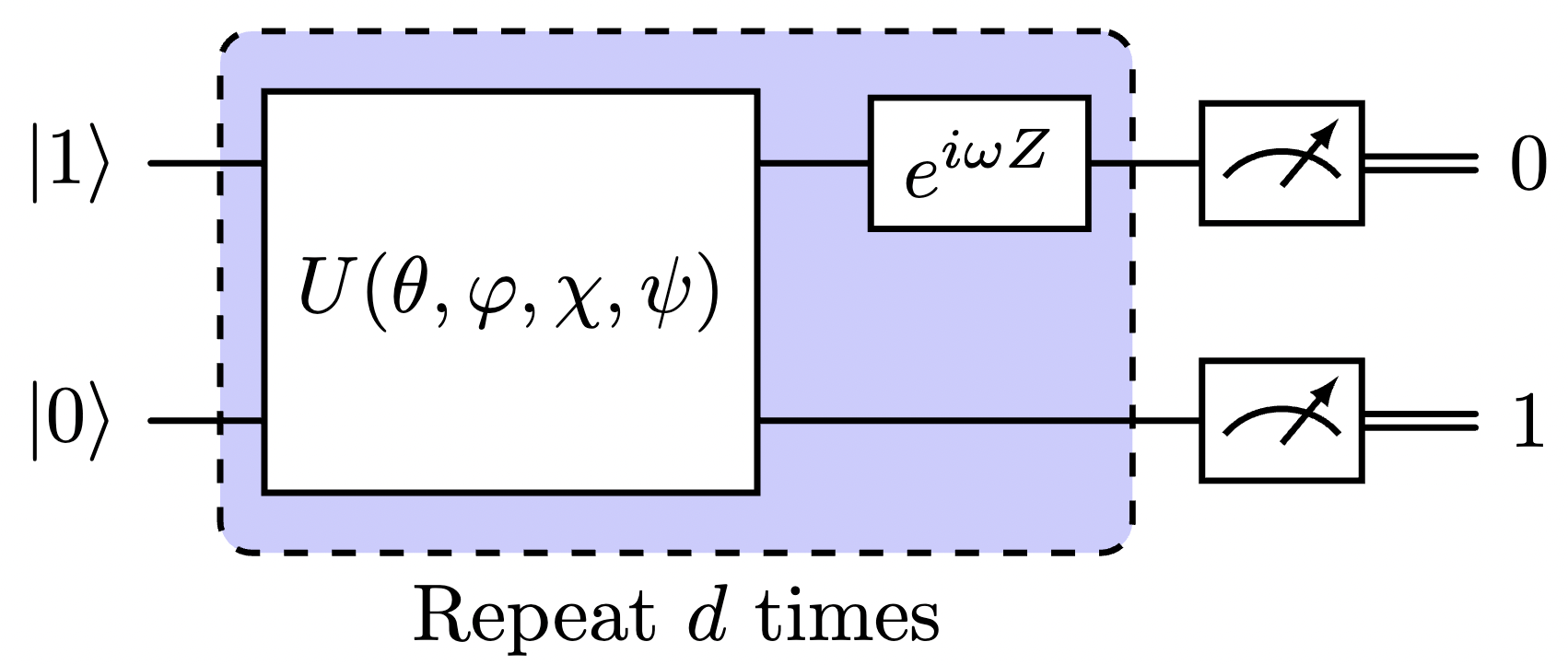}
                \caption{Quantum circuit for periodic calibration using the population method.\label{fig:population-method-circuit}}
    \end{figure}
    
	One difficulty with these techniques is that small values of the swap angle $\theta$ are difficult to \REV{amplify} in the presence of larger single-qubit phases.
	This can be addressed adaptively, by first measuring the unwanted single-qubit phases and applying compensating pulses, but this strategy is limited by the precision with which one can compensate, and the speed with which these single-qubit phases drift relative to the experiment time.
	For these reasons, in practice estimation of the swap angle is often done with a depth-1 circuit, commonly referred to as unitary tomography~\cite{foxen2020}. 

 In \cref{app:periodic-calibration}, we provide a comprehensive analysis of periodic calibration. By bounding the Fisher information, we show that the optimal estimation variance depends heavily on the satisfaction of \REV{the} phase-matching condition. Though Heisenberg-limit scaling is achieved with perfect phase-matching, realistic errors, e.g. time-dependent drift errors, render the perfect satisfaction of phase-matching condition challenging. According to the analysis, when the swap angle is small, a slight violation of the phase-matching condition can completely undermine the optimal Heisenberg-limit scaling and render the estimation variance exponentially worse in terms of depth dependency. These make the use of periodic calibration impractical when the swap angle is small.

 \begin{figure}[htbp]
\centering
\includegraphics[width=.5\columnwidth]{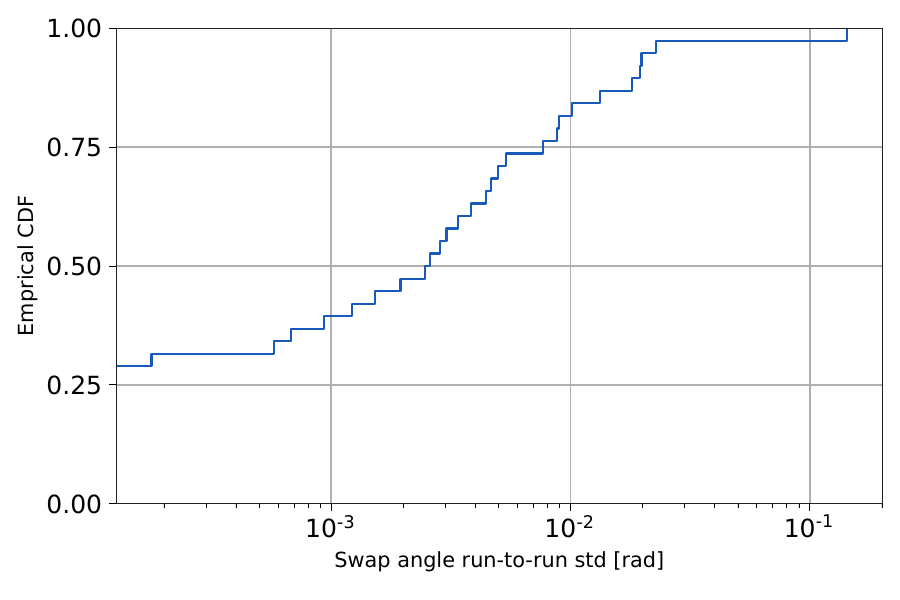}
\caption{Distribution of run-to-run variation of swap-angle estimation across a device.
    The swap angles were estimated using periodic calibration on four independent datasets for each CZ gate, with 10,000 samples per circuit and a maximum depth of 30.
    Due to the behavior of the periodic-calibration estimator for particularly small swap angles, a substantial fraction of swap angles were estimated to be identically 0, leading to the portion of the cumulative distribution function that extends off the plot to the left.
    Discarding these instances leaves us with a median run-to-run standard deviation of close to $4\times10^{-3}$ radians.
    } 
\label{fig:cz_calibrate_exp}
\end{figure}

In \cref{fig:cz_calibrate_exp}, we demonstrate the experimental estimation of swap angles across multiple runs on a Google Quantum AI superconducting device. The results show that periodic calibration fails to achieve the accuracy needed to reliably distinguish such small gate angles from zero. Moreover, the run-to-run standard deviation in swap-angle estimations varies widely, ranging from $10^{-4}$ to $10^{-1}$ radians, with a median of approximately $4\times10^{-3}$ radians. This instability in calibration outcomes stems from the periodic calibration's high sensitivity to time-dependent errors. Linking this to the theoretical analysis, the presence of time-dependent errors also impedes the precise fulfillment of the phase-matching condition, thereby significantly restricting the accuracy of periodic calibration estimations.

    An alternative characterization scheme using cross-entropy-benchmarking (XEB) circuits was described in Sec. C. 2. of the supplemental material for~\cite{GoogleQuantumSupremacy2019}.
    This characterization tool randomizes various noise sources into an effective depolarizing channel, allowing noise to be simply characterized along with unitary parameters.
    Randomization comes at a cost, though, requiring a large number of random circuits to get a representative sample of the distribution.
    Also, randomization interferes with the ability of unitary errors to build up coherently, keeping this method from achieving the Heisenberg limit.
    This makes it difficult for XEB characterization to resolve angles below $10^{-2}$ radians in practice.

	\subsection{Quantum signal processing and polynomial analysis}
	The quantum circuit used in QSPE (see Figure 1 in the main text) contains a periodic circuit structure in which the $U$-gate and a Z-rotation are interleaved. This circuit structure coincides with a quantum algorithm called quantum signal processing (QSP) \cite{LowChuang2017,GilyenSuLowEtAl2019}. QSP is a useful quantum algorithm for solving numerical linear algebra problems such as quantum linear system problems and Hamiltonian simulation by properly choosing a set of phase factors~\cite{DongMengWhaleyEtAl2020,martyn2021grand}. Specifically, in this paper, we will use the polynomial structure induced by the theory of QSP \cite{LowChuang2017,GilyenSuLowEtAl2019,WangDongLin2021}. Though the quantum circuit used in the paper is a special case of general QSP circuit by fixing all phase modulation angles, the general QSP structure may provide a more robust paradigm against stochastic phase drift errors by relaxing the fixed angle constraint. Hence, for completeness, we provide a concise overview of the theory of QSP structure in this subsection.
 
 The following theorem is a simplified version of \cite[Theorem 1]{WangDongLin2021}.
	\begin{theorem}[Polynomial structure of symmetric QSP]\label{thm:qsp}
		Let $d\in \NN$ and $\varOmega := (\omega_0, \cdots, \omega_d) \in \RR^{d+1}$ be a set of phase factors. Then, for any $x \in [-1,1]$, the following product of $\mathrm{SU}(2)$-matrices admits a representation
		\begin{equation}
		\label{eqn:qsp-gslw}
		U(x, \varOmega) = e^{\I \omega_0 Z} \prod_{j=1}^{d} \left( e^{\I \arccos(x) X} e^{\I \omega_j Z} \right) = \left( \begin{array}{cc}
		P(x) & \I Q(x) \sqrt{1 - x^2}\\
		\I Q^*(x) \sqrt{1 - x^2} & P^*(x)
		\end{array} \right)
		\end{equation}
		for some $P,Q\in \CC[x]$ satisfying that
		\begin{enumerate}
			\item[(1)] \label{itm:1} $\deg(P) \leq d, \deg(Q) \leq d-1$,
			\item[(2)] \label{itm:2} $P(x)$ has parity $(d\mod2)$ and $Q(x)$ has parity $(d-1 \mod 2)$,
			\item[(3)]  \label{itm:3} $|P(x)|^2 + (1-x^2) |Q(x)|^2 = 1, \forall x \in [-1, 1]$.
		\end{enumerate}
		Here, the superscript $*$ denotes the complex conjugate of a polynomial, namely $P^*(x) = \sum_i \overline{p_i} x^i$ if $P(x) = \sum_i p_i x^i$ with $p_i \in \CC$. Furthermore, if $\varOmega$ is chosen to be symmetric, namely $\omega_j = \omega_{d-j}$ for any $j$, then $Q \in \RR[x]$ is a real polynomial.
	\end{theorem}
	\begin{proof}
	    We give a concise proof for completeness. 
	    
	    ``Condition (1)'': Note that $\mathrm{SU}(2)$ matrices satisfy
	    \begin{equation*}
	        e^{\I \arccos(x) X} = x I + \I \sqrt{1-x^2} X \quad \mathrm{and} \quad X e^{\I \omega Z} = e^{-\I \omega Z} X.
	    \end{equation*}
	    The polynomial representation in \cref{eqn:qsp-gslw} follows the expansion and rearranging Pauli $X$ matrices. The condition (1) follows the observation that the leading term is at most $x^d$ when there are even number of Pauli $X$ matrices in the expansion while it is at most $x^{d-1}$ when the number of Pauli $X$ matrices is odd. 
	    
	    ``Condition (2)'': To see condition (2), we note that under the transformation $x \mapsto -x$, we have
	    \begin{equation*}
	        e^{\I \arccos(-x) X} = e^{\left(\pi - \arccos(x)\right)X} = - e^{-\I \arccos(x) X} = - Z e^{\I\arccos(x) X} Z.
	    \end{equation*}
	    Therefore
	    \begin{equation*}
	       U(-x,\varOmega) = (-1)^d Z U(x,\varOmega) Z = \left( \begin{array}{cc}
		(-1)^d P(x) & \I (-1)^{d-1} Q(x) \sqrt{1 - x^2}\\
		\I (-1)^{d-1} Q^*(x) \sqrt{1 - x^2} & (-1)^d P^*(x)
		\end{array} \right)
	    \end{equation*}
	    which implies that
	    \begin{equation*}
	        P(-x) = (-1)^d P(x) \quad \mathrm{and} \quad Q(-x) = (-1)^{d-1} Q(x)
	    \end{equation*}
	    which is the parity condition.
	    
	    ``Condition (3)'': Condition (3), which is equivalent to $\det U(x,\varOmega) = 1$, directly follows the special unitarity.
	    
	    ``Symmetric QSP'': Note that when $\Phi$ is symmetric, $U(x, \Phi)$ is invariant under the matrix transpose which reverses the order of phase factors. Using $U(x, \Phi) = U(x,\Phi)^\top$, the condition on the polynomial $Q(x) = Q^*(x)$ follows the transformation of the off-diagonal element. Therefore, $Q \in \RR[x]$ is a real polynomial.
	\end{proof}
	
	The previous theorem bridges the gap between the periodic circuits and the analysis of polynomials. In the paper, we will frequently invoke an important inequality of polynomials to bound error, which is stated as follows.
	
	\begin{theorem}[Markov brothers' inequality \cite{Markov1890}]\label{thm:Markovs-ineq}
	    Let $P \in \RR_d[x]$ be any algebraic polynomial of degree at most $d$. For any nonnegative integer $k$, it holds that
	    \begin{equation}
	        \max_{x \in [-1,1]} \abs{P^{(k)}(x)} \le \max_{x \in [-1,1]} \abs{P(x)} \prod_{j=0}^{k-1} \frac{d^2-j^2}{2j+1}.
	    \end{equation}
	    The equality is attained for Chebyshev polynomial of the first kind $T_d(x)$.
	\end{theorem}

 	\subsection{Fermionic simulation gate (\fsim)}\label{appendix:fsimgate}
	A special instance of $U$-gate is a  fermionic simulation gate (\fsim). It is a class of two-qubit quantum gates preserving the excitation and \REV{describes} all two-qubit gates realizable in Google's superconducting qubit system. Acting on two qubits $A_0$ and $A_1$, the \fsim is parametrized by a few parameters. Ordering the basis as $\mc{B} := \{ \ket{00}, \ket{01}, \ket{10}, \ket{11} \}$ where the qubits are ordered as $\ket{a_0a_1} := \ket{a_0}_{A_0}\ket{a_1}_{A_1}$, the unitary matrix representation of the \fsim is given by
	\begin{equation}
	U_\fsim(\theta,\varphi,\chi,\psi,\phi) = \left(\begin{array}{*{4}c}
	1 & 0 & 0 & 0 \\
	0 & e^{-\I \varphi-\I\psi}\cos\theta & -\I e^{\I \chi - \I \psi} \sin\theta & 0\\
	0 & -\I e^{-\I\chi-\I\psi}\sin\theta & e^{\I\varphi-\I\psi}\cos\theta & 0\\
	0 & 0 & 0 & e^{-\I(\phi + 2 \psi)}
	\end{array}\right).
	\end{equation}
	As a consequence of the preservation of excitation, there is a two-dimensional invariant subspace of the \fsim, which is referred to as the single-excitation subspace spanned by basis states $\mc{B} = \{ \ket{01}, \ket{10} \}$. Restricted on the single-excitation subspace $\mc{E} := \spans\ \mc{B}$, the matrix representation of the \fsim is (up to a global phase)
	\begin{equation}
	\begin{split}
	& \matrepwrt{U_\fsim(\theta, \varphi, \chi, \phi, \psi)}{\mc{B}} =: U_\fsim^{\mc{B}}(\theta, \varphi, \chi)\\
	&= \left(\begin{array}{cc}
	e^{-\I\varphi}\cos\theta & -\I e^{\I\chi}\sin\theta \\
	-\I e^{-\I\chi}\sin\theta & e^{\I\varphi}\cos\theta
	\end{array}
	\right) = e^{-\I \frac{\varphi - \chi - \pi}{2} Z} e^{\I \theta X} e^{-\I \frac{\varphi + \chi + \pi}{2} Z}.
	\end{split}
 \label{eq:single-excite-fsim-gate}
	\end{equation}
	Here, $X$ and $Z$ are logical Pauli operators by identifying logical quantum states $\ket{0}_\ell := \ket{01}$ and $\ket{1}_\ell := \ket{10}$. As a remark, it provides a parametrization of any general $\mathrm{SU}(2)$ matrix. 
	
	One of the most important two-qubit quantum gates is controlled-Z gate (CZ). It forms universal gate sets with several single-qubit gates and it is a pivotal building block for demonstrating surface code \cite{acharya2022suppressing}. CZ is in the gate class of \fsim 's, which can be generated by setting $\theta = \varphi = \chi = \psi = 0$ and $\phi = \pi$. Due to the noisy implementation of CZ, the resulting quantum gate is an \fsim slightly deviating \REV{from} the perfect CZ. In order to perform high-fidelity quantum computation, one has to characterize the erroneous parameters of an \fsim which include CZ as a special case. The characterization of gate parameters relies on quantum phase estimation techniques.

	\subsection{Notation}
	Throughout the paper, $M$ refers to the number of measurement samples unless otherwise noted. For a matrix $A\in\CC^{m\times n}$, the transpose, Hermitian conjugate and complex conjugate are denoted by $A^{\top}$, $A^{\dag}$, $\overline{A}$, respectively. The same notations are also used for the operations on a vector. The complex conjugate of a complex number $a$ is denoted as $\overline{a}$. We define the basis kets of the state space of a qubit as follows
\[
\ket{0} := \begin{pmatrix}
1\\0
\end{pmatrix}, \quad \ket{1} := 
\begin{pmatrix}
0\\ 1
\end{pmatrix}.
\]

\section{Details on QSPE in general cases}\label{app:general-QSPE}
In the main text, we provide a quantum circuit for QSPE exemplified by two-qubit $U$-gates. However, QSPE can be applied to general $U$-gate with a two-level invariant subspace spanned by logical basis states $\mc{B} := \{\ket{0_\ell}, \ket{1_\ell}\}$. Intuitively, we apply a logical $Z$-rotation to modulate the phase angles in the invariant subspace to form the desired functional form of the output signal. Consequently, to perform the modulation in the invariant subspace, the rotation gate is a logical $Z$-rotation $e^{\I \omega Z_\ell}$ with the logical Pauli operator defined on the logical basis $Z_\ell =  \ket{0_\ell} \bra{0_\ell} - \ket{1_\ell}\bra{1_\ell}$. Then, the matrix representation under the two-level basis states coincides with the conventional $Z$-ration, namely $[e^{\I \omega Z_\ell}]_\mc{B} = e^{\I \omega Z}$. The initial state is prepared to two superposition states of logical basis states, which are referred to as logical Bell states: $\ket{+_\ell} := \frac{1}{\sqrt{2}}\left( \ket{0_\ell} + \ket{1_\ell} \right)$ and $\ket{\I_\ell} := \frac{1}{\sqrt{2}}\left( \ket{0_\ell} + \I \ket{1_\ell} \right)$. To read the output signal, we perform a measurement onto the logical state $\ket{0_\ell}$ to measure the transition probability for further analysis. To summarize, the quantum circuit for QSPE in general cases is given in \cref{fig:general-QSPE}.

\begin{figure}[htbp]
    \centering
    \includegraphics[width=.75\textwidth]{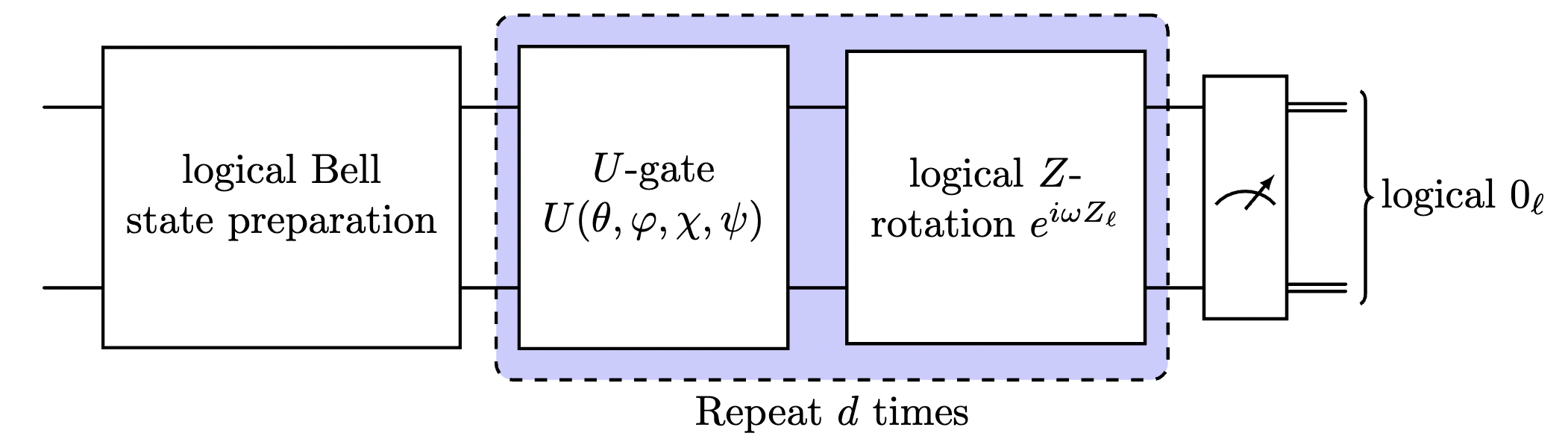}
    \caption{Quantum circuit for QSPE in general cases. The input quantum state is prepared to be logical Bell states in either $\ket{+_\ell}$ or $\ket{\I_\ell}$ according to the type of experiment. The quantum circuit enjoys a periodic structure of the unknown (multi-qubit) $U$-gate and a tunable logical $Z$-rotation. The measurement is performed to measure the transition probability onto the logical basis state $\ket{0_\ell}$.}
    \label{fig:general-QSPE}
\end{figure}

In the invariant subspace spanned by $\mc{B}$, the matrix representation of the periodic part of the circuit is
\begin{equation}
    \mc{U}^\scp{d}(\omega; \theta, \varphi, \chi) = \matrepwrt{\left( e^{\I \omega Z_\ell} U(\theta, \varphi, \chi, \psi) \right)^d}{\mc{B}} = e^{- \I d \psi} \left(e^{\I \omega Z} e^{-\I \frac{\varphi - \chi - \pi}{2} Z} e^{\I \theta X} e^{-\I \frac{\varphi + \chi + \pi}{2} Z} \right)^d.
\end{equation}
Here, the Euler-angle decomposition of the two-dimensional unitary $[U(\theta, \varphi, \chi, \psi)]_\mc{B}$ is used. We also remark that the $\psi$-dependency in the unitary representation is hidden for simplicity because the global phase $e^{- \I d \psi}$ does not affect measurement probability. In the later presentation, we will drop the global phase or mark it as $\psi \leftarrow *$ for simplicity. Note that the adjacent $Z$-rotations can be combined to simplify the expression. Hence, the unitary representation is
\begin{equation}\label{eqn:general-QSPE-matrix}
     \mc{U}^\scp{d}(\omega; \theta, \varphi, \chi) =  e^{\I\frac{\chi+\pi+\varphi}{2} Z} \left( e^{\I(\omega - \varphi) Z} e^{\I \theta X} \right)^d e^{- \I\frac{\chi+\pi+\varphi}{2} Z}.
\end{equation}
It is worth noting that the periodic part is independent of the phase accumulation angle $\chi$. Consequently, despite the amplification of $\theta$ and $\varphi$, the angle $\chi$ is not amplified on this basis. Hence, the estimation accuracy of $\chi$ is not significantly affected by increasing the depth parameter $d$ in the quantum circuit outlined in \cref{fig:general-QSPE}. To conclude, \cref{eqn:general-QSPE-matrix} will be the starting point of our analysis in the later section.
	
	\section{Analytical structure of periodic circuit}\label{sec:analytical-result-qspc}
	The QSPE circuit in Figure 1 in the main text (or \cref{fig:general-QSPE}) enjoys a periodic structure by interleaving $U$-gate and $Z$-rotation. This periodic structure is studied by the theory of QSP (\cref{thm:qsp}). Consequentially, the QSPE circuit admits some polynomial representation. In this section, we will derive the analytical form of the structure of the QSPE circuit. We start from the exact closed-form results of the QSPE circuit in \cref{sec:exact-pc}. In \cref{sec:app-pc}, we derive a good approximation to the closed-form exact results. The analysis in this section proves the following theorem which summarizes findings about the structure of QSPE.

 	\begin{theorem}[Structure of QSPE]\label{thm:structure-of-qsp-pc}
		Let $d \in \NN$ be the number of $U$-gate applications in the QSPE circuit, and 
		\begin{equation}\label{qspc-h-eq}
		    \mf{h}(\omega; \theta, \varphi, \chi) := p_X(\omega; \theta, \varphi, \chi) - \frac{1}{2} 
		  + \I \left(p_Y(\omega; \theta, \varphi, \chi) - \frac{1}{2}\right)
		\end{equation}
		be the reconstructed function derived from the measurement probability. Then, it admits a finite Fourier series expansion
		\begin{equation}
		\mf{h}(\omega; \theta, \varphi, \chi) = \sum_{-d+1}^{d-1} c_k(\theta, \varphi, \chi) e^{2\I k \omega}.
		\end{equation}
		Furthermore, for nonnegative indices $k = 0, 1, \cdots, d-1$, the Fourier coefficients take the form
		\begin{equation}
		c_k(\theta,\chi,\varphi) = \I e^{-\I\chi} e^{-\I(2k+1)\varphi} \theta + \mathrm{max} \left\{ \Or\left(\theta^3\right), \Or\left((d\theta)^5\right) \right\}.
		\end{equation}
	\end{theorem}
 
	\subsection{Exact representation of the periodic circuit}\label{sec:exact-pc}
	We abstract a simple $\mathrm{SU}(2)$-product model which can be shown as the building block of the QSPE circuit in Figure 1 in the main text (or \cref{fig:general-QSPE}). It turns out that the model admits a polynomial representation. 
	\begin{definition}[Building block of QSPE]\label{def:qspc-unitary}
		Let $\theta, \omega \in \RR$ be any angles, $d \in \NN$ by any positive integer. Then, the matrix representation of a periodic circuit with $d$ repetitions and $Z$-phase modulation angle $\omega$ is
		\begin{equation}
		U^\scp{d}(\omega,\theta) = \left(e^{\I \omega Z} e^{\I \theta X}\right)^{d} e^{\I\omega Z}.
		\end{equation}
	\end{definition}
	In the quantum circuit defined above, the $X$- and $Z$-rotations are interleaved, which agrees with the structure of QSP in \cref{thm:qsp}. The theory of QSP implies that the $\mathrm{SU}(2)$-product model enjoys a structure represented by polynomials which is given by the following lemma.
	\begin{lemma}\label{lma:qsp-unitary-polynomial}
		Let $x = \cos(\theta) \in [-1,1]$. There exists a complex polynomial $P_\omega^\scp{d} \in \CC_d[x]$ and a real polynomial $Q_\omega^\scp{d} \in \RR_{d-1}[x]$ so that
		\begin{equation}\label{eqn:lma:qsp-unitary-polynomial}
		U^\scp{d}\left(\omega, \arccos(x)\right) = \left(\begin{array}{cc}
		P^\scp{d}_\omega(x) & \I \sqrt{1-x^2} Q^\scp{d}_\omega(x) \\
		\I \sqrt{1-x^2} Q^{\scp{d} *}_\omega(x) & P^{\scp{d} *}_\omega(x)
		\end{array}\right).
		\end{equation}
		Furthermore, the special unitarity of $U^\scp{d}(\omega, \arccos(x))$ yields 
		\begin{equation}\label{eqn:P2+Q2=1}
		P_\omega^\scp{d}(x) P_\omega^{\scp{d} *}(x) + (1-x^2) \left(Q_\omega^\scp{d}(x)\right)^2 = 1.
		\end{equation}
	\end{lemma}
	\begin{proof}
		Following \cite[Theorem 4]{GilyenSuLowEtAl2019}, there exists two polynomials $P^\scp{d}_\omega, Q^\scp{d}_\omega \in \CC[x]$ so that \cref{eqn:lma:qsp-unitary-polynomial} holds. Because $U^\scp{d}(\omega, \arccos(x))$ is a QSP unitary with a set of symmetric phase factors, $Q_\omega^\scp{d} \in \RR_{d-1}[x]$ is a real polynomial according to \cite[Theorem 2]{DongMengWhaleyEtAl2020}. \cref{eqn:P2+Q2=1} holds by taking the determinant of \cref{eqn:lma:qsp-unitary-polynomial}.
	\end{proof}
	The exact presentation of the pair of polynomials $(P_\omega^\scp{d}, Q_\omega^\scp{d})$ can be determined via recurrence on a special set of points $d = 2^j, j \in \NN$ (see \cref{lma:poly-rep-special-pts} in Appendix). Based on it, we prove the generalized result to any positive integer $d$ by using induction. This gives a complete characterization of the structure of the $\mathrm{SU}(2)$-product model in \cref{def:qspc-unitary}.
	
	\begin{theorem}\label{thm:QSP-PC-P-Q}
		Let $d = 1, 2, \ldots$ be any positive integer. Then 
		\begin{equation}
		P_\omega^\scp{d}(x) = e^{\I\omega} \left(\cos\left(d \sigma\right) + \I \frac{\sin\left(d \sigma\right)}{\sin\sigma} \left(\sin \omega\right) x\right) \text{ and } Q_\omega^\scp{d}(x) = \frac{\sin\left(d \sigma\right)}{\sin\sigma}
		\end{equation}
		where $\sigma = \arccos\left(\left(\cos \omega\right) x \right)$.
	\end{theorem}
	\begin{proof}
		Let us prove the theorem by induction. The base case is $d=1$, where $P^\scp{1}_\omega(x) = e^{2\I\omega} x$ and $Q^\scp{1}_\omega(x) = 1$. Assuming that the induction hypothesis holds for $d$, we will prove it also holds for $d+1$. Using \cref{def:qspc-unitary,lma:qsp-unitary-polynomial}, the polynomials can be determined by a recurrence relation
		\begin{equation}
		U^\scp{d+1}(\omega,\theta) = U^\scp{d}(\omega,\theta) e^{\I\theta X} e^{\I \omega Z} \Rightarrow \left\{
		\begin{array}{l}
		P^\scp{d+1}_\omega(x) = e^{\I\omega}\left(x P^\scp{d}_\omega(x) - (1-x^2) Q^\scp{d}_\omega(x)\right),\\
		Q^\scp{d+1}_\omega(x) = e^{-\I\omega} \left(P^\scp{d}_\omega(x) + x Q^\scp{d}_\omega(x)\right).
		\end{array}
		\right.
		\end{equation}
		Using the induction hypothesis, we have
		\begin{equation}
		\begin{split}
		P_\omega^\scp{d+1}(x) &= e^{\I\omega} \left(\cos\sigma\cos(d\sigma) - \left(1-\left(1-\sin^2\omega\right)x^2\right) \frac{\sin(d\sigma)}{\sin\sigma} \right.\\
		&\quad\quad\quad\quad\quad\quad \left.+ \I \left(\sin\omega\right)x \frac{\sin\sigma\cos(d\sigma)+\cos\sigma\sin(d\sigma)}{\sin\sigma}\right)\\
		&= e^{\I\omega} \left(\cos\left((d+1) \sigma\right) + \I \frac{\sin\left((d+1) \sigma\right)}{\sin\sigma} \left(\sin \omega\right) x\right)
		\end{split}
		\end{equation}
		and
		\begin{equation}
		Q_\omega^\scp{d+1}(x) = \cos(d\sigma) + \cos\sigma\frac{\sin(d\sigma)}{\sin\sigma} = \frac{\sin\left((d+1)\sigma\right)}{\sin\sigma}.
		\end{equation}
		Therefore, the theorem follows induction.
	\end{proof}
The closed-form results above help us to analyze the dynamics of Figure 1 in the main text (or \cref{fig:general-QSPE}) where we  apply  a Pauli $Z$ modulation $e^{\I\omega Z_{A_0}}$ to the periodic circuit. Restricted to the single-excitation subspace, the matrix representation of the QSPE circuit in Figure 1 in the main text (or \cref{fig:general-QSPE}) is
	\begin{equation}\label{eqn:circuit-rep-of-qspc-with-building-block}
	\mc{U}^\scp{d}(\omega; \theta, \varphi, \chi) = \matrepwrt{\left(e^{\I \omega Z_\ell} U\left(\theta,\varphi,\chi,*\right)\right)^d}{\mc{B}} = e^{\I\frac{\chi+\pi+\varphi}{2} Z} U^\scp{d}(\omega-\varphi, \theta) e^{-\I\left(\omega + \frac{\chi+\pi-\varphi}{2} \right)Z}.
	\end{equation}
	The initial two-qubit state of the QSP circuit can be prepared as Bell states $\ket{+_\ell}$ or $\ket{\I_\ell}$ by using Hadamard gate, phase gate and CNOT gate. Recall that we denote the probability by measuring qubits $A_0A_1$ with $01$ as
	\begin{equation}
	p_X(\omega; \theta, \varphi, \chi) = \abs{\bra{0_\ell}\mc{U}^\scp{d}(\omega; \theta, \varphi, \chi)\ket{+_\ell}}^2
	\end{equation}
	when the initial state is $\ket{+_\ell}$, and
	\begin{equation}
	p_Y(\omega; \theta, \varphi, \chi) = \abs{\bra{0_\ell}\mc{U}^\scp{d}(\omega; \theta, \varphi, \chi)\ket{\I_\ell}}^2
	\end{equation}
	when the initial state is $\ket{\I_\ell}$ respectively. These bridge the gap between the analytical results derived based on \cref{def:qspc-unitary} and the measurement probabilities from the experimental setting. We are ready to prove the first half of \cref{thm:structure-of-qsp-pc}. 
	
	\begin{theorem}\label{thm:reconstruction-h-Fourier-expansion}
		The function reconstructed from the measurement probability admits the following Fourier series expansion:
		\begin{equation}
		\mf{h}(\omega; \theta, \varphi, \chi) := p_X(\omega; \theta, \varphi, \chi) + \I p_Y(\omega; \theta, \varphi, \chi) - \frac{1+ \I}{2} = \sum_{k = -d+1}^{d-1} c_{k}(\theta, \chi, \varphi) e^{2 \I k \omega}
		\end{equation}
		where
		\begin{equation}
		c_k(\theta,\chi,\varphi) = \I e^{-\I\chi} e^{-\I(2k+1)\varphi} \wt{c}_k(\theta)\quad \text{and}\quad \wt{c}_k(\theta) \in \RR.
		\end{equation}
	\end{theorem}
	\begin{proof}
		For simplicity, let $\beta \in \mathrm{U}(1)$ and $\ket{\beta} := \frac{1}{\sqrt{2}}\left(\ket{0_\ell} + \beta \ket{1_\ell}\right)$. Then, $\ket{\beta=1} = \ket{+_\ell}$ and $\ket{\beta = \I} = \ket{\I_\ell}$. Given the input quantum state is $\ket{\beta}$, we have the measurement probability
		\begin{equation}
		\begin{split}
		p_\beta(\omega;\theta,\varphi) &= \abs{\bra{0_\ell } \mc{U}^\scp{d}(\omega; \theta, \varphi, \chi)\ket{ \beta}}^2\\
		&= \abs{\frac{1}{\sqrt{2}} e^{\I\frac{\varphi+\chi+\pi}{2}} \bra{0_\ell} U^\scp{d}(\omega-\varphi, \theta) \left(e^{-\I\left(\omega + \frac{\chi+\pi-\varphi}{2} \right)}\ket{0_\ell} + \beta e^{\I\left(\omega + \frac{\chi+\pi-\varphi}{2} \right)}\ket{1_\ell}\right)}^2\\
		&= \frac{1}{2} + \Re\left( \overline{\beta} e^{\I(\varphi-\chi-2\omega)} P_{\omega-\varphi}^\scp{d}(\cos\theta) \I \sin\theta Q_{\omega-\varphi}^\scp{d}(\cos\theta) \right).
		\end{split}
		\end{equation}
		Then, $p_X = p_{\beta=1}$ and $p_Y = p_{\beta = \I}$. Furthermore, it holds that
		\begin{equation}\label{eqn:pX-Re-pY-Im}
		\begin{split}
		& p_X(\omega; \theta, \varphi, \chi) - \frac{1}{2} = \Re\left( e^{\I(\varphi-\chi-2\omega)} P_{\omega-\varphi}^\scp{d}(\cos\theta) \I \sin\theta Q_{\omega-\varphi}^\scp{d}(\cos\theta) \right),\\
		& p_Y(\omega; \theta, \varphi, \chi) - \frac{1}{2} = \Im\left( e^{\I(\varphi-\chi-2\omega)} P_{\omega-\varphi}^\scp{d}(\cos\theta) \I \sin\theta Q_{\omega-\varphi}^\scp{d}(\cos\theta) \right).
		\end{split}
		\end{equation}
		Therefore, the reconstructed function is
		\begin{equation}
		\mf{h}(\omega; \theta, \varphi, \chi) = \I e^{-\I(\chi+\varphi)} \sin\theta e^{-2\I(\omega - \varphi)} P_{\omega-\varphi}^\scp{d}(\cos\theta) Q_{\omega-\varphi}^\scp{d}(\cos\theta) =: \I e^{-\I(\chi+\varphi)} \wt{\mf{h}}(\omega-\varphi,\theta).
		\end{equation}
		Note that following \cref{thm:QSP-PC-P-Q}
		\begin{equation}
		\begin{split}
		&P_{\omega+\pi-\varphi}^\scp{d}(\cos\theta) = (-1)^{d+1} P_{\omega-\varphi}^\scp{d}(\cos\theta),\ Q_{\omega+\pi-\varphi}^\scp{d}(\cos\theta) = (-1)^{d-1} Q_{\omega-\varphi}^\scp{d}(\cos\theta)\\
		& \Rightarrow\ \wt{\mf{h}}(\omega+\pi-\varphi,\theta) = \wt{\mf{h}}(\omega-\varphi,\theta).
		\end{split}
		\end{equation}
		That means $\wt{\mf{h}}(\omega-\varphi,\theta)$ is $\pi$-periodic in the first argument. Furthermore, $\wt{\mf{h}}(\omega-\varphi,\theta)$ is a trigonometric polynomial in $(\omega-\varphi)$. Thus, it admits the Fourier series expansion:
		\begin{equation}
		\wt{\mf{h}}(\omega-\varphi,\theta) = \sum_{k=-d+1}^{d-1} \wt{c}_k(\theta) e^{2\I k (\omega-\varphi)}
		\end{equation}
		with coefficients
		\begin{equation}\label{eqn:tilde-c-k-expr}
		    \wt{c}_k(\theta) = \frac{\sin\theta}{\pi} \int_0^\pi e^{-2\I(k+1)\omega} P_\omega^\scp{d}(\cos\theta) Q_\omega^\scp{d}(\cos\theta) \rd \omega.
		\end{equation}
		The upper limit and lower limit of the summation index $\pm (d-1)$ can be verified by straightforward computation. According to \cref{thm:QSP-PC-P-Q}, we also have
		\begin{equation}
		P_{-\omega}^\scp{d}(\cos\theta) = \overline{P_\omega^\scp{d}(\cos\theta)},\ Q_\omega^\scp{d}(\cos\theta) \in \RR\ \Rightarrow\ \wt{\mf{h}}(\omega,\theta) = \overline{\wt{\mf{h}}(-\omega,\theta)}\ \Rightarrow\ \wt{c}_k(\theta) \in \RR.
		\end{equation}
		The proof is completed.
	\end{proof}
	It is also useful to study the magnitude of the reconstructed function. It gives the intuition of the distribution of the magnitude over different modulation angle $\omega$. The following corollary indicated that the magnitude of the reconstructed function attains its maximum $d\theta$ when the phase matching condition $\omega = \varphi$ is achieved.
	\begin{corollary}\label{cor:modulus-magnitude-sin2theta}
		The magnitude of $p_X(\omega; \theta, \varphi, \chi)-\frac{1}{2}$ and $p_Y(\omega; \theta, \varphi, \chi)-\frac{1}{2}$ are of order $\sin\theta$. Furthermore
		\begin{equation}\label{eqn:px2+py2-calibration}
		\begin{split}
		\mf{p}(\omega-\varphi, \theta) &:= \abs{\mf{h}(\omega;\theta,\varphi,\chi)}^2 =  \left(p_X(\omega; \theta, \varphi, \chi) - \frac{1}{2}\right)^2 + \left(p_Y(\omega; \theta, \varphi, \chi) - \frac{1}{2}\right)^2\\
		&= \sin^2(\theta) \frac{\sin^2(d\sigma)}{\sin^2(\sigma)} \left(1 - \sin^2(\theta) \frac{\sin^2(d\sigma)}{\sin^2(\sigma)}\right).
		\end{split}
		\end{equation}
		Here $\sigma = \arccos\left(\cos(\omega-\varphi)\cos(\theta)\right)$.
	\end{corollary}
	\begin{proof}
		Using \cref{thm:QSP-PC-P-Q,eqn:P2+Q2=1,eqn:pX-Re-pY-Im} as intermediate steps, we have
		\begin{equation}
		\begin{split}
		\mf{p}(\omega-\varphi, \theta) &= \abs{e^{\I(\varphi-\chi-2\omega)} P_{\omega-\varphi}^\scp{d}(\cos\theta) \I \sin\theta Q_{\omega-\varphi}^\scp{d}(\cos\theta)}^2\\
		&= \sin^2(\theta) \abs{Q_{\omega-\varphi}^\scp{d}(\cos\theta)}^2 \abs{P_{\omega-\varphi}^\scp{d}(\cos\theta)}^2\\
		&= \sin^2(\theta) \abs{Q_{\omega-\varphi}^\scp{d}(\cos\theta)}^2 \left(1 - \sin^2(\theta) \abs{Q_{\omega-\varphi}^\scp{d}(\cos\theta)}^2\right)\\
		&= \sin^2(\theta) \frac{\sin^2(d\sigma)}{\sin^2(\sigma)} \left(1 - \sin^2(\theta) \frac{\sin^2(d\sigma)}{\sin^2(\sigma)}\right)
		\end{split}
		\end{equation}
		which completes the proof.
	\end{proof}
	
Notice that if the transition probability between tensor-product states is measured, the magnitude of the signal (the nontrivial $\theta$ dependence in the transition probability) is $\Or\left(\sin^2\theta\right)$. Nonetheless, by preparing the input quantum state as Bell states, \cref{cor:modulus-magnitude-sin2theta} reveals that the magnitude of the signal is lifted to $\Or(\sin\theta)$ instead. Therefore, when $\theta$ is small, it is a significant improvement of the SNR especially in the presence of realistic errors.
	
	\subsection{Approximate Fourier coefficients}\label{sec:app-pc}
	\cref{thm:reconstruction-h-Fourier-expansion} shows that the $\theta$ and $\varphi$ dependence are factored completely in the amplitude and the phase of the Fourier coefficients of the reconstructed function $\mf{h}$ respectively. Given the angle $\omega$ of the $Z$-rotation is tunable, we can sample the data point by performing the QSPE circuit in Figure 1 in the main text (or \cref{fig:general-QSPE}) with equally spaced angles $\omega_j = \frac{j}{2d-1} \pi$ where $j = 0, \cdots, 2d-2$. These $2(2d-1)$ quantum experiments yield two sequences of measurement probabilities $\vp_X^\expl := \left(p^\expl_X(\omega_0), p^\expl_X(\omega_1), \cdots, p^\expl_X(\omega_{2d-2})\right)$ and $\vp_Y^\expl := \left(p^\expl_Y(\omega_0), p^\expl_Y(\omega_1), \cdots, p^\expl_Y(\omega_{2d-2})\right)$. Therefore, we can compute $\mf{h}^\expl = \vp_X^\expl + \vp_Y^\expl - \frac{1+\I}{2}$ from experimental data. The Fourier coefficients of $\mf{h}$ can be computed by fast Fourier transform (FFT). The Fourier coefficients $\vc^\expl := \left(c_{-d+1}^{(\expl)}, c_{-d+2}^{(\expl)}, \cdots, c_{d-1}^{(\expl)}\right) = \mathsf{FFT}(\mf{h}^\expl)$ can be computed efficiently using FFT. In order to infer $\theta$ and $\varphi$ accurately and efficiently from the data, we need to study the approximate structure of the Fourier coefficients first.
	
	\begin{theorem}\label{thm:approx-coef-first-order}
		Let $\hat{\mf{h}}(\omega, \cos\theta) := \wt{\mf{h}}(\omega,\theta)/(\sin\theta e^{-\I \omega})$. There is an approximation to it: 
		\begin{equation}\label{eqn:hat_c_k_star_expr}
		\begin{split}
		& \hat{\mf{h}}^\star(\omega,\cos\theta) = \sum_{k=-d+1}^{d-1} \hat{c}_k^\star(\theta) e^{\I(2k+1)\omega},\quad \text{where }\\
		 & \hat{c}_k^\star(\theta) = \left\{
		\begin{array}{ll}
		1 - \frac{1}{2}\left(3d^2-k^2-(k+1)^2-\left(d-(2k+1)\right)^2\right) (1-\cos\theta) \ \text{ if } 0 \le k \le d-1, \\
		- \frac{1}{2} \left(d^2 + (d+2k+1)^2 - k^2 - (k+1)^2\right) (1-\cos\theta) \quad \text{ if } -d+1 \le k \le -1.
		\end{array} \right.
		\end{split}
		\end{equation}
		The approximation error is upper bounded as
		\begin{equation}
		\max_{\omega \in [0,\pi]} \abs{\hat{\mf{h}}(\omega,\cos\theta) - \hat{\mf{h}}^\star(\omega,\cos\theta)} \le 2 d^5 \theta^4
		\end{equation}
		and for any $k$
		\begin{equation}
		\abs{\wt{c}_k(\theta) - \sin\theta \hat{c}_k^\star(\theta)} \le 2(d\theta)^5.
		\end{equation}
	\end{theorem}
	\begin{proof}
		Following \cref{thm:QSP-PC-P-Q}, we have
		\begin{equation}
		\begin{split}
		\wt{\mf{h}}(\omega,\theta) &= \sin\theta e^{-2\I \omega} P_\omega^\scp{d}(\cos\theta) Q_\omega^\scp{d}(\cos\theta) = \sin\theta e^{-\I\omega} \left(\cos(d\sigma) + \I \sin\omega \cos\theta \frac{\sin(d\sigma)}{\sin\sigma}\right)\frac{\sin(d\sigma)}{\sin\sigma}\\
		&= \sin\theta e^{-\I\omega} \left(T_d(\cos\sigma) + \I \sin\omega \cos\theta U_{d-1}(\cos\sigma)\right)U_{d-1}(\cos\sigma)
		\end{split}
		\end{equation}
		where $\cos\sigma = \cos\omega \cos\theta$ and $T_d \in \RR_d[x], U_{d-1} \in \RR_{d-1}[x]$ are Chebyshev polynomials of the first and second kind respectively. Then
		\begin{equation}
		\begin{split}
		\hat{\mf{h}}(\omega,\cos\theta) &:= \left(T_d(\cos\sigma) + \I \sin\omega \cos\theta U_{d-1}(\cos\sigma)\right)U_{d-1}(\cos\sigma) \\
		&= \frac{1}{2} U_{2d-1}(\cos\sigma) + \I \sin\omega\cos\theta U_{d-1}^2(\cos\sigma).
		\end{split}
		\end{equation}
		Therefore, for a given $\omega$, $\hat{\mf{h}}(\omega, \cos\theta)$ is a polynomial in $\cos\theta$ of degree at most $2d-1$. According to \cref{cor:modulus-magnitude-sin2theta}, we have for any $\omega$
		\begin{equation}\label{eqn:h-re-im-upper-bound}
		\max_{\theta \in [0,\pi]} \max\left\{ \abs{\Re\left(\hat{\mf{h}}(\omega,\cos\theta)\right)}, \abs{\Im\left(\hat{\mf{h}}(\omega,\cos\theta)\right)} \right\} \le \max_{\theta \in [0,\pi]} \abs{\hat{\mf{h}}(\omega,\cos\theta)} \le \max_{\theta \in [0,\pi]} \abs{\frac{\sin(d\sigma)}{\sin\sigma}} \le d.
		\end{equation}
		Applying Taylor's theorem and expanding $\hat{\mf{h}}(\omega,\cos\theta)$ with respect to $1-\cos\theta$, there exists $\xi \in (\cos\theta, 1)$ so that
		\begin{equation}\label{eqn:taylor-remainder}
		\hat{\mf{h}}(\omega,\cos\theta) = \hat{\mf{h}}(\omega,1) + \frac{\partial \hat{\mf{h}}(\omega,x)}{\partial x}\bigg|_{x=1} (\cos\theta - 1) + \frac{1}{2} \frac{\partial^2 \hat{\mf{h}}(\omega,x)}{\partial x^2}\bigg|_{x=\xi}(\cos\theta - 1)^2.
		\end{equation}
		Here
		\begin{equation}
		\hat{\mf{h}}(\omega,1) = e^{\I d \omega} U_{d-1}(\cos\omega) = e^{\I \omega} \sum_{k=0}^{d-1} e^{2\I k \omega}.
		\end{equation}
		Furthermore,
		\begin{equation}
		\begin{split}
		& \Re\left(\frac{\partial \hat{\mf{h}}(\omega,x)}{\partial x}\bigg|_{x=1} \right) = \frac{1}{2} \frac{\partial U_{2d-1}(x\cos\omega)}{\partial x}\bigg|_{x=1} = \frac{\cos\omega}{2} U_{2d-1}^\prime(\cos\omega)\\
		&= \cos\omega \sum_{j=0}^{d-1} (2j+1) U_{2j}(\cos\omega) = \cos\omega \sum_{j=0}^{d-1} (2j+1) \sum_{k=-j}^j e^{2\I k \omega}\\
		&= \cos\omega \sum_{k=-(d-1)}^{d-1} (d^2-k^2) e^{2\I k \omega}
		\end{split}
		\end{equation}
		and
		\begin{equation}
		\begin{split}
		& \Im\left(\frac{\partial \hat{\mf{h}}(\omega,x)}{\partial x}\bigg|_{x=1} \right) = \sin\omega U_{d-1}(\cos\omega) \left(U_{d-1}(\cos\omega) + 2\cos\omega U_{d-1}^\prime(\cos\omega)\right)\\
		&= \sin(d\omega)\left(\sum_{\substack{k=-(d-1)\\ \text{stepsize 2}}}^{d-1} e^{\I k \omega} + \frac{1}{2} \left(e^{\I\omega} + e^{-\I\omega}\right) \sum_{\substack{k=-(d-2)\\ \text{stepsize 2}}}^{d-2} (d^2-k^2) e^{\I k \omega}\right)\\
		&= \sin(d\omega) \sum_{\substack{k=-(d-1)\\ \text{stepsize 2}}}^{d-1} (d^2-k^2) e^{\I k \omega} = \frac{1}{2\I} \left(\sum_{\substack{k=-2d+1\\ \text{stepsize 2}}}^{-1} (2d+k)k e^{\I k \omega} + \sum_{\substack{k=1\\ \text{stepsize 2}}}^{2d-1} (2d-k)k e^{\I k \omega}\right).
		\end{split}
		\end{equation}
		Let the approximation of $\hat{\mf{h}}(\omega,\cos\theta)$ be
		\begin{equation}
		\hat{\mf{h}}^\star(\omega,\cos\theta) := \hat{\mf{h}}(\omega,1) + \frac{\partial \hat{\mf{h}}(\omega,x)}{\partial x}\bigg|_{x=1} (\cos\theta - 1).
		\end{equation}
		Then, the previous computation shows it admits a Fourier series expansion:
		\begin{equation}
		\begin{split}
		& \hat{\mf{h}}^\star(\omega,\cos\theta) = \sum_{k=-d+1}^{d-1} \hat{c}_k^\star(\theta) e^{\I(2k+1)\omega},\ \text{where } \\
		& \hat{c}_k^\star(\theta) = \left\{
		\begin{array}{ll}
		1 + \frac{1}{2}\left(3d^2-k^2-(k+1)^2-\left(d-(2k+1)\right)^2\right) (\cos\theta - 1) \ \text{ if } 0 \le k \le d-1, \\
		\frac{1}{2} \left(d^2 + (d+2k+1)^2 - k^2 - (k+1)^2\right) (\cos\theta - 1) \quad \text{ if } -d+1 \le k \le -1.
		\end{array} \right.
		\end{split}
		\end{equation}
		The approximation error can be bounded by using \cref{eqn:taylor-remainder}. For any $\omega \in [0,\pi]$, we have
		\begin{equation}
		\begin{split}
		& \abs{\hat{\mf{h}}(\omega,\cos\theta) - \hat{\mf{h}}^\star(\omega,\cos\theta)} \le \frac{(1-\cos\theta)^2}{2} \max_{x \in [-1,1]} \abs{\frac{\partial^2 \hat{\mf{h}}(\omega,x)}{\partial x^2}}\\
		& \le \frac{\theta^4}{8} \sqrt{\left(\max_{x \in [-1,1]}\abs{\frac{\partial^2 \Re\left(\hat{\mf{h}}(\omega,x)\right)}{\partial x^2}}\right)^2 + \left(\max_{x \in [-1,1]}\abs{\frac{\partial^2 \Im\left(\hat{\mf{h}}(\omega,x)\right)}{\partial x^2}}\right)^2}.
		\end{split}
		\end{equation}
		Note that $\Re\left(\hat{\mf{h}}(\omega,x)\right)$ and $\Im\left(\hat{\mf{h}}(\omega,x)\right)$ are real polynomials in $x$ of degree at most $2d-1$. Invoking the Markov brothers' inequality (\cref{thm:Markovs-ineq}), we further get
		\begin{equation}
		\begin{split}
		& \abs{\hat{\mf{h}}(\omega,\cos\theta) - \hat{\mf{h}}^\star(\omega,\cos\theta)}\\
		& \le \frac{\sqrt{2} (2d-1)^2 d(d-1)}{6} \theta^4 \max_{x \in [-1,1]} \max\left\{ \abs{\Re\left(\hat{\mf{h}}(\omega,x)\right)}, \abs{\Im\left(\hat{\mf{h}}(\omega,x)\right)} \right\}\\
		& \le \frac{\sqrt{2} (2d-1)^2 d^2(d-1)}{6} \theta^4 \le 2 d^5 \theta^4,
		\end{split}
		\end{equation}
		where \cref{eqn:h-re-im-upper-bound} is used. The error bound can be transferred to that of the Fourier coefficients. Using the previous result and triangle inequality, one has
		\begin{equation}
		\begin{split}
		\abs{\wt{c}_k(\theta) - \sin\theta \hat{c}_k^\star(\theta)} &= \abs{\frac{\sin\theta}{\pi} \int_0^\pi e^{-\I (2k+1) \omega} \left(\hat{\mf{h}}(\omega,\cos\theta) - \hat{\mf{h}}^\star(\omega,\cos\theta)\right) \rd \omega}\\
		& \le \sin\theta \frac{1}{\pi} \int_0^\pi \abs{\hat{\mf{h}}(\omega,\cos\theta) - \hat{\mf{h}}^\star(\omega,\cos\theta)} \rd \omega \le 2 \left(d\theta\right)^5. 
		\end{split}
		\end{equation}
		The proof is completed.
	\end{proof}
	
	There are two implications of the previous theorem. First, it suggests that the magnitude of the Fourier coefficients of negative indices are $\Or\left(\sin^3\theta\right)$. If they are included in the formalism of the inference problem in the Fourier space, the accuracy of inference may be heavily contaminated because of the nearly vanishing SNR when $\theta \ll 1$. On the other hand, the amplitude of the Fourier coefficients of nonnegative indices tightly concentrates at $\sin\theta$ when $\theta \ll 1$. Therefore, a nice linear approximation of the Fourier coefficients holds in the case of small swap angle: for any $k = 0, \cdots, d-1$
	\begin{equation}
	c_k(\theta,\chi,\varphi) = \I e^{-\I\chi} e^{-\I(2k+1)\varphi} \theta + \mathrm{max} \left\{ \Or\left(\theta^3\right), \Or\left((d\theta)^5\right) \right\}.
	\end{equation}
	This proves the second half of \cref{thm:structure-of-qsp-pc}.

To numerically demonstrate the previous theorem, we depicted the exactly computed Fourier coefficients and their approximation in \cref{fig:approximate-Fourier-coef}. The numerical results support that the approximated coefficients are close to the exact ones when $d \theta \le 1$. Furthermore, the vanishing values of negatively indexed Fourier coefficients and the approximately linear growth of nonnegatively indexed Fourier coefficients are visualized in the figure.
    \begin{figure}[htbp]
        \centering
        \subfigure[]{
            \includegraphics[width=.475\columnwidth]{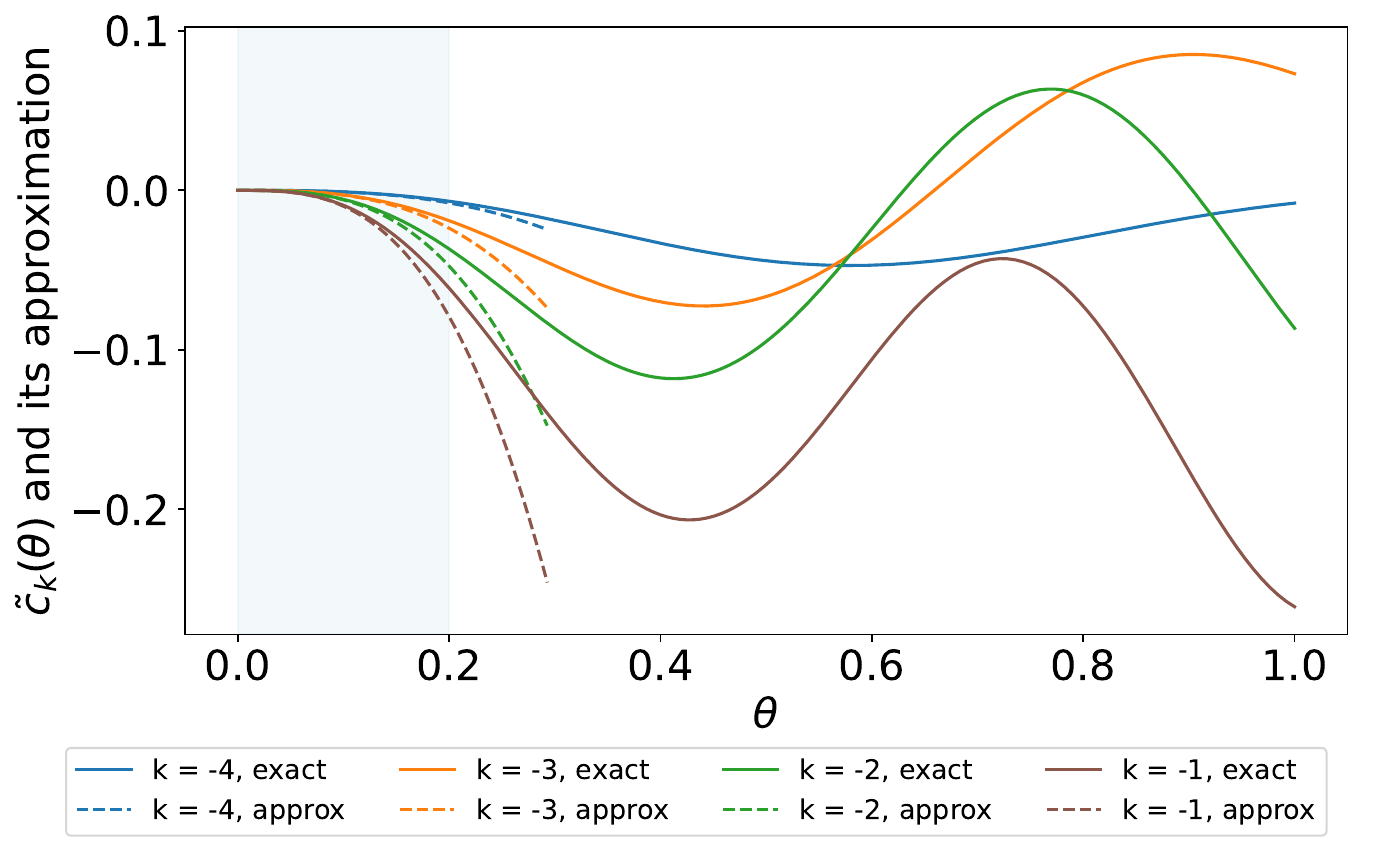}
        }
        \subfigure[]{
            \includegraphics[width=.475\columnwidth]{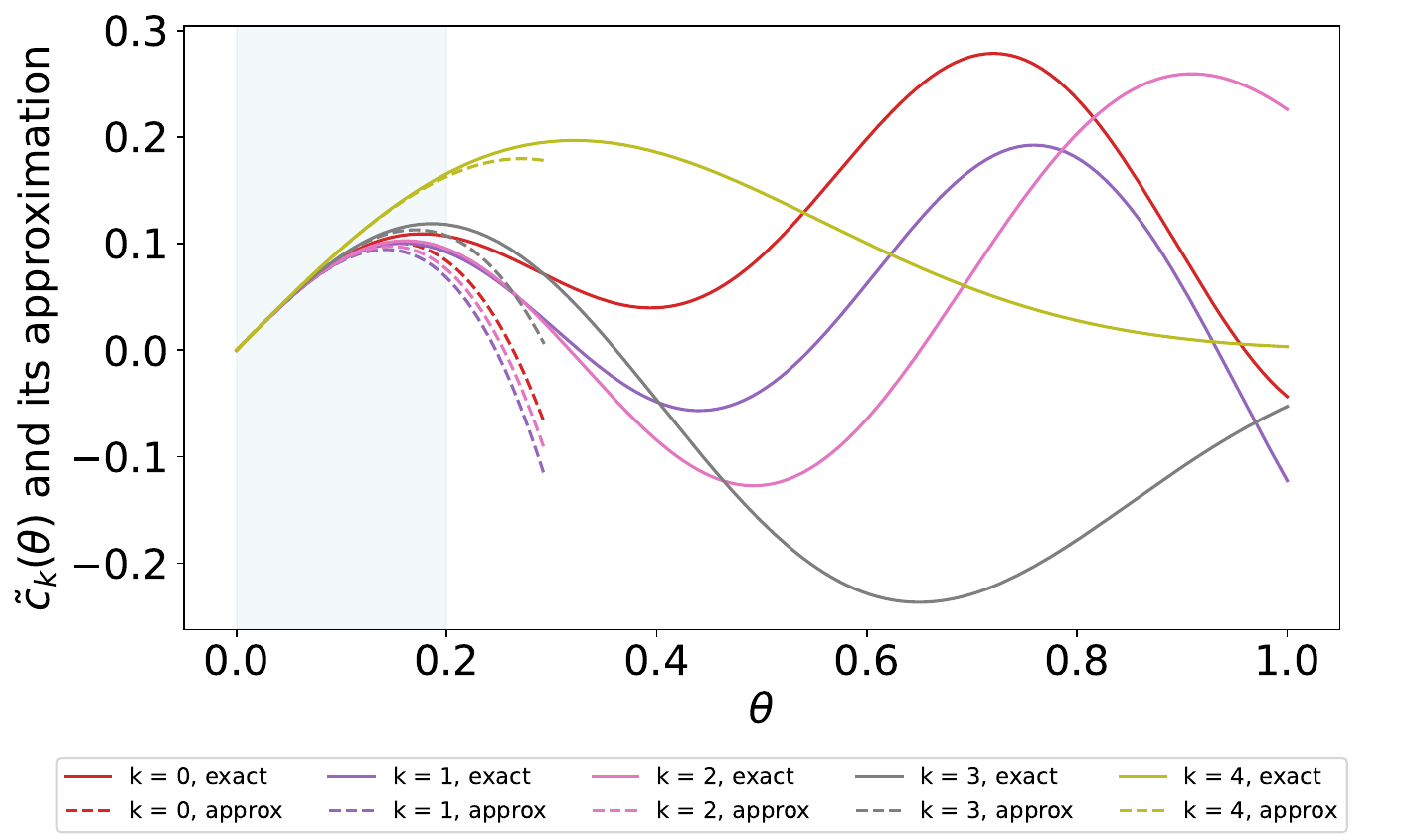}
        }
        \caption{Numerical justification of the approximated Fourier coefficients. We set $d = 5$. The solid curve is computed by integrating the Fourier transformation. The dashed curve is derived from the approximation in the previous theorem. The shaded area stands for $d \theta \le 1$. (a) Fourier coefficients of negative indices. (b) Fourier coefficient of nonnegative indices.\label{fig:approximate-Fourier-coef}}
    \end{figure}
 
	\section{Robust estimator against Monte Carlo sampling error}\label{sec:Monte-Carlo-sampling-error}
 The estimation problem of QSPE is formalized as follows.

\begin{problem}[Calibrating $U$-gate using  QSPE]\label{prob:inference-Fourier}
    (1) QSPE: Given experimentally measured probabilities of QSPE circuits in Figure 1 in the main text (or \cref{fig:general-QSPE}) on $\{\omega_j: j = 0, \cdots, 2d-2\}$, infer $\theta$ and $\varphi$ accurately. 
    
    (2) QSPE in Fourier space: Given experimentally measured Fourier coefficients of nonnegative indices, infer $\theta$ and $\varphi$ accurately.
\end{problem}
A dominant and unavoidable source of errors in quantum metrology is Monte Carlo sampling error due to the finite sample size in quantum measurements. Such limitation derives from both practical concerns of the efficiency of quantum metrology, and realistic constraints where some system parameters can drift over time and can only be monitored by sufficiently fast protocols. In this section, we  analyze the effect of Monte Carlo sampling error in our proposed metrology algorithm by characterizing the sampling error as a function of quantum circuit depth, \fsim parameters and sample size. The result will also be used in \cref{sec:lower-bound-qspc-metrology} to prove that our estimator based on QSPE is optimal. In the following analysis, we annotate  with superscript ``$\expl$'' to represent experimentally measured probability as opposed to expected probability from theory.
 
	\subsection{Modeling the Monte Carlo sampling error}\label{subsec:model_MC}
	We start the analysis by statistically modeling the Monte Carlo sampling error on the measurement probabilities. Furthermore, we also derive the sampling error induced on the Fourier coefficients derived from experimental data. The result is summarized in the following lemma.
	\begin{lemma}\label{lma:monte-carlo-error-magnitude}
		Let $M$ be the number of measurement samples in each experiment. When $M$ is large enough, the measurement probability $p_X^\expl(\omega_j)$ is approximately normal distributed
		\begin{equation}
		p_X^\expl(\omega_j) = p_X(\omega_j;\theta,\varphi,\chi) + \Sigma_{X,j} u_{X,j},\text{ where } u_{X,j} \sim N(0,1) \text{ and } \frac{1-4(d\theta)^2}{4M} \le \Sigma_{X,j}^2 \le \frac{1}{4M}.
		\end{equation}
		The same conclusion holds for $p_Y^\expl(\omega_j)$. Furthermore, by computing the Fourier coefficients via FFT, the Fourier coefficients are approximately complex normal distributed
		\begin{equation}
		c_k^\expl = c_k(\theta,\varphi,\chi) + \left\{ \begin{array}{ll}
		v_k &, k = 0,\cdots, d-1,   \\
		v_{2d-1+k} &, k = -d+1,\cdots, -1. 
		\end{array}\right.
		\end{equation}
		where $v_k$'s are complex normal distributed random variables so that
		\begin{equation}
		\begin{split}
		&\text{for any }k:\ \expt{v_k} = 0,\ \frac{1 - 2(d\theta)^2}{2M(2d-1)} \le \expt{\abs{v_k}^2} \le \frac{1}{2M(2d-1)},\\
		&\text{ and for any } k \ne k^\prime:\ \abs{\expt{v_k \overline{v_{k^\prime}}}} \le \frac{(d\theta)^2}{M(2d-1)}.
		\end{split}
		\end{equation}
		Consequentially, when $d\theta \ll 1$, these random variables $v_k$'s can be approximately assumed to be uncorrelated.
	\end{lemma}
	\begin{proof}
		Given a quantum experiment with angle $\omega_j$, the measurement generates i.i.d. Bernoulli distributed outcomes $b_i$'s, namely $\bP(b_i=0) = 1-\bP(b_i=1) = p_X(\omega_j;\theta,\varphi,\chi)$. Then, the measurement probability is estimated by $p_X^\expl(\omega_j) = \frac{1}{M}\sum_{i=1}^M (1-b_i)$. When the sample size $M$ is large enough, $p^\expl_X(\omega_j)$ is approximately normal distributed following the central limit theorem where the mean is $\expt{p_X^\expl(\omega_j)} = p_X(\omega_j;\theta,\varphi,\chi)$ and the variance is 
		\begin{equation}\label{eqn:variance-mc-error-concentration}
		\begin{split}
		\Sigma_{X,j}^2 &:= \var\left(p^\expl_X(\omega_j)\right) = \frac{p_X(\omega_j;\theta,\varphi,\chi)\left(1-p_X(\omega_j;\theta,\varphi,\chi)\right)}{M}\\
		& = \frac{1}{4M} - \frac{\left(p_X(\omega_j;\theta,\varphi,\chi) - \frac{1}{2}\right)^2}{M} \le \frac{1}{4M}.
		\end{split}
		\end{equation}
		The other side of the inequality $\Sigma_{X,j}^2 \ge \frac{1-4(d\theta)^2}{4M}$ follows that $(p_X-\frac{1}{2})^2 \le (d\theta)^2$ from \cref{cor:modulus-magnitude-sin2theta}. The same analysis is applicable to $p_Y^\expl(\omega_j)$. 
		
		To compute the Fourier coefficients from the experimental data, we perform FFT on $\mf{h}^\expl_j := p^\expl_X(\omega_j) + \I p_Y^\expl(\omega_j) - \frac{1+\I}{2}$ reconstructed from the experimental data. We have $\expt{\mf{h}^\expl_j} = \mf{h}(\omega_j;\theta,\varphi,\chi)$. Furthermore, let $\wt{u}_j = \mf{h}^\expl_j - \expt{\mf{h}^\expl_j} = \Sigma_{X,j}u_{X,j} + \I \Sigma_{Y,j}u_{Y,j}$, then it holds that
		\begin{equation}
		\begin{split}
		& \text{for any }j,\ \expt{\wt{u}_j} = 0,\ \expt{\abs{\wt{u}_j}^2} = \Sigma_{X,j}^2 + \Sigma_{Y,j}^2 = \frac{1}{2M} - \frac{\mf{p}(\omega_j-\varphi,\theta)}{M}, \\
		&\text{and for any }j \ne j^\prime, \expt{\wt{u}_j \overline{\wt{u}_{j^\prime}}} = 0.
		\end{split}
		\end{equation}
		The FFT gives the Fourier coefficients as
		\begin{equation}
		\left(\begin{array}{c}
		c_0^\expl\\ \vdots \\ c_{d-1}^\expl \\ c_{-d+1}^\expl \\ \vdots \\ c_{-1}^\expl
		\end{array}\right) = \frac{1}{2d-1} \Omega^\dagger \left(\begin{array}{l}
		\mf{h}^\expl_0 \\ \vdots  \\ \mf{h}^\expl_{2d-2}
		\end{array}\right), \text{ where } \Omega_{jk} = e^{\I \frac{2\pi jk}{2d-1}}.
		\end{equation}
		Using the linearity, we get 
		\begin{equation}\label{eqn:monte-carlo-error-in-Fourier-coef}
		v_k = \frac{1}{2d-1}\left(\Omega^\dagger \mf{h}^\expl\right)_k = \frac{1}{2d-1} \sum_{j=0}^{2d-2} \overline{\Omega_{kj}} \wt{u}_j.
		\end{equation}
		The mean is $\expt{v_k} = \frac{1}{2d-1} \sum_{j=0}^{2d-2} \overline{\Omega_{kj}} \expt{\wt{u}_j} = 0$. The covariance is
		\begin{equation}
		\expt{v_k \overline{v_{k^\prime}}} = \frac{1}{(2d-1)^2} \sum_{j, j^\prime = 0}^{2d-2} \overline{\Omega_{kj}} \Omega_{k^\prime j^\prime} \expt{\wt{u}_j \overline{\wt{u}_{j^\prime}}} = \frac{1}{(2d-1)^2} \sum_{j=0}^{2d-2} e^{\I \frac{2\pi}{2d-1}(k^\prime-k)j} \expt{\abs{\wt{u}_j}^2}.
		\end{equation}
		When $k = k^\prime$, it gives
		\begin{equation}
		\begin{split}
		& \expt{\abs{v_k}^2} = \frac{1}{2M(2d-1)} - \frac{1}{M (2d-1)^2} \sum_{j=0}^{2d-2} \mf{p}(\omega_j-\varphi, \theta)\\
		& \Rightarrow \frac{1 - 2(d\theta)^2}{2M(2d-1)} \le \expt{\abs{v_k}^2} \le \frac{1}{2M(2d-1)}
		\end{split}
		\end{equation}
		where $0 \le \mf{p}(\omega_j-\varphi, \theta) \le (d\theta)^2$ is used which follows \cref{cor:modulus-magnitude-sin2theta}.
		
		On the other hand, when $k \ne k^\prime$, the constant term $\frac{1}{2M}$ in $\expt{\abs{\wt{u}_j}^2}$ vanishes because $\sum_{j=0}^{2d-2} e^{\I \frac{2\pi}{2d-1}(k^\prime-k)j} = (2d-1) \delta_{k k^\prime}$. Then, using triangle inequality and \cref{cor:modulus-magnitude-sin2theta}, we get
		\begin{equation}
		\begin{split}
		\abs{\expt{v_k \overline{v_{k^\prime}}}} &= \abs{\frac{1}{(2d-1)^2} \sum_{j=0}^{2d-2} e^{\I \frac{2\pi}{2d-1}(k^\prime-k)j} \mf{p}(\omega_j-\varphi, \theta)}\\
		& \le \frac{1}{M (2d-1)^2} \sum_{j=0}^{2d-2} \abs{\mf{p}(\omega_j-\varphi, \theta)} \le \frac{(d\theta)^2}{M(2d-1)}.
		\end{split}
		\end{equation}
		The proof is completed.
	\end{proof}
	With a characterization of Monte Carlo sampling error, we are able to measure the robustness of the signal against error by the signal-to-noise ratio (SNR). The SNR of each Fourier coefficient is defined as the ratio between the squared Fourier coefficient and the variance of its associated additive sampling error. We define the SNR of QSPE in \cref{prob:inference-Fourier} by the minimal component-wise SNR. The following theorem gives a characterization of the SNR.
	\begin{theorem}\label{thm:snr-qspc}
		When $d^{\frac{5}{4}} \theta \ll 1$, the signal-to-noise ratio satisfies
		\begin{equation}
		\mathrm{SNR}_k := \frac{\abs{c_k(\theta,\varphi,\chi)}^2}{\expt{\abs{v_k}^2}} \ge \mathrm{SNR} := 2(2d-1)M\sin^2\theta \left(1 - \frac{4}{3}(d\theta)^2\left(1 + 3d^3\theta^2\right)\right).
		\end{equation}
	\end{theorem}
	\begin{proof}
		According to \cref{eqn:hat_c_k_star_expr}, for any $k = 0, \cdots, d-1$
		\begin{equation}
		1 - \frac{2}{3} (d\theta)^2 \le \hat{c}_k^\star(\theta) \le 1.
		\end{equation}
		Applying \cref{thm:approx-coef-first-order}, we have
		\begin{equation}
		c_k(\theta,\varphi,\chi) \ge \sin\theta\left(\hat{c}_k^\star(\theta) - 2d^5\theta^4 \right) \ge \sin\theta \left(1 - \frac{2}{3}(d\theta)^2\left(1 + 3d^3\theta^2\right) \right).
		\end{equation}
		Furthermore, by Bernoulli's inequality,
		\begin{equation}
		\abs{c_k(\theta,\varphi,\chi)}^2 \ge \sin^2\theta \left(1 - \frac{4}{3}(d\theta)^2\left(1 + 3d^3\theta^2\right) \right).
		\end{equation}
		Combing the derived inequality with \cref{lma:monte-carlo-error-magnitude}, it gives
		\begin{equation}
		\mathrm{SNR}_k = \frac{\abs{c_k(\theta,\varphi,\chi)}^2}{\expt{\abs{v_k}^2}} \ge 2(2d-1)M\sin^2\theta \left(1 - \frac{4}{3}(d\theta)^2\left(1 + 3d^3\theta^2\right)\right),
		\end{equation}
		which completes the proof.
	\end{proof}
	
	\subsection{Statistical estimator against Monte Carlo sampling error}\label{subsec:stat-estimator-MC}
	As a consequence of \cref{thm:snr-qspc}, when SNR is high, namely $d M \theta^2 \gg 1$, the noise modeling in Ref. \cite{Tretter1985} suggests that a linear model with normal distributed noise can well approximate the problem in which the $\theta$- and $(\varphi,\chi)$-dependence are decoupled following \cref{thm:reconstruction-h-Fourier-expansion}.
	When $k = 0, \cdots, d-1$ and $d^5 \theta^4 \ll 1$, we have 
	\begin{equation}\label{eqn:modeled-problem-Fourier-space}
	\begin{split}
	& \mathsf{amplitude}\left(c^\expl_k\right) = \wt{c}_k(\theta) + v_k^\scp{\mathrm{amp}} \approx \theta + v_k^\scp{\mathrm{amp}},\\
	& \mathsf{phase}(c^\expl_k) = \frac{\pi}{2} - \chi - (2k+1) \varphi + v_k^\scp{\mathrm{pha}} \text{ (up to } 2\pi\text{-periodicity)}
	\end{split}
	\end{equation}
	where $v_k^\scp{\mathrm{amp}}$ and $v_k^\scp{\mathrm{pha}}$ are normal distributed and are approximately $v_k^\scp{\mathrm{amp}} = \Re(v_k)$, $v_k^\scp{\mathrm{pha}} = \Im(v_k)/\wt{c}_k(\theta)$ according to Ref. \cite{Tretter1985}. Let the covariance matrices be $\mc{C}^\scp{\mathrm{amp}}$ and $\mc{C}^\scp{\mathrm{pha}}$. For any $k$ and $k^\prime$
	\begin{equation}
	\mc{C}_{k,k^\prime}^\scp{\mathrm{amp}} := \expt{v_k^\scp{\mathrm{amp}}v_{k^\prime}^\scp{\mathrm{amp}}} \text{ and } \mc{C}_{k,k^\prime}^\scp{\mathrm{pha}} := \expt{v_k^\scp{\mathrm{pha}}v_{k^\prime}^\scp{\mathrm{pha}}}.
	\end{equation}
	Let the data vectors be
	\begin{equation}
	\abs{\bvec{c}^\expl} := \left(\mathsf{amplitude}(c_0^\expl), \cdots, \mathsf{amplitude}(c_{d-1}^\expl)\right)^\top,\ \boldsymbol{\wt{\mymathbb{1}}} = (\underbrace{1,\cdots,1}_d)^\top.
	\end{equation}
	The maximum likelihood estimator (MLE) is found by minimizing the negated log-likelihood function
	\begin{equation}
	\hat{\theta} = \myargmin_\theta \left(\abs{\bvec{c}^\expl} - \theta \boldsymbol{\wt{\mymathbb{1}}}\right)^\top \left(\mc{C}^\scp{\mathrm{amp}}\right)^{-1} \left(\abs{\bvec{c}^\expl} - \theta \boldsymbol{\wt{\mymathbb{1}}}\right).
	\end{equation}
	which follows the normality in \cref{lma:monte-carlo-error-magnitude}.
	
	In order to estimate $\varphi$, we can apply the Kay's phase unwrapping estimator in Ref.\cite{Kay1989}, a.k.a. weighted phase average estimator (WPA). The estimator is based on the sequential phase difference of the successive coefficients:
	\begin{equation}\label{eqn:sequential-phase-difference}
	\mathsf{phase}\left(c_k^\expl \overline{c_{k+1}^\expl}\right) = 2 \varphi + v_k^\scp{\mathrm{pha}} - v_{k+1}^\scp{\mathrm{pha}},\ k = 0, 1, \cdots, d-2.
	\end{equation}
	Remarkably, by computing the sequential phase difference, the troublesome $(2\pi)$-periodicity in \cref{eqn:modeled-problem-Fourier-space} can be overcome. According to this equation, the noise is turned into a colored noise process. Let the covariance be
	\begin{equation}\label{eqn:covariance-WPA}
	\mc{D}_{k,k^\prime} := \expt{\left(v_k^\scp{\mathrm{pha}} - v_{k+1}^\scp{\mathrm{pha}}\right)\left(v_{k^\prime}^\scp{\mathrm{pha}} - v_{k^\prime+1}^\scp{\mathrm{pha}}\right)} = \mc{C}_{k,k^\prime}^\scp{\mathrm{pha}} + \mc{C}_{k+1,k^\prime+1}^\scp{\mathrm{pha}} - \mc{C}_{k,k^\prime+1}^\scp{\mathrm{pha}} - \mc{C}_{k+1,k^\prime}^\scp{\mathrm{pha}}.
	\end{equation}
	Then, the WPA estimator is derived by the following MLE:
	\begin{equation}
	\hat{\varphi} = \myargmin_{\varphi} \left(\boldsymbol{\Delta} - 2 \varphi \boldsymbol{\mymathbb{1}}\right)^\top \mc{D}^{-1} \left(\boldsymbol{\Delta} - 2 \varphi \boldsymbol{\mymathbb{1}}\right)
	\end{equation}
	where the data vectors are
	\begin{equation}
	\boldsymbol{\Delta} = \left( \mathsf{phase}\left(c_0^\expl \overline{c_{1}^\expl}\right), \cdots, \mathsf{phase}\left(c_{d-2}^\expl \overline{c_{d-1}^\expl}\right) \right)^\top, \text{ and } \boldsymbol{\mymathbb{1}} = (\underbrace{1,\cdots,1}_{d-1})^\top.
	\end{equation}
	To solve the MLE, we need to study the structure of covariance matrices, which is given by the following lemma.
	\begin{lemma}\label{lma:covariance-mod-pha-estimation}
		When $d\theta \le \frac{1}{5}$ and $d^3\theta^2 \le 1$, for any $k \ne k^\prime$
		\begin{equation}
		\abs{\mc{C}_{k,k^\prime}^\scp{\mathrm{amp}}} \le \frac{(d\sin\theta)^2}{M(2d-1)}, \text{ and } \abs{\mc{C}_{k,k^\prime}^\scp{\mathrm{pha}}} \le \frac{4 d^2}{3 M(2d-1)}.
		\end{equation}
		For any $k$
		\begin{equation}
		\abs{\mc{C}_{k,k}^\scp{\mathrm{amp}} - \frac{1}{4M(2d-1)}} \le \frac{2(d\sin\theta)^2}{M(2d-1)}, \text{ and } \abs{\mc{C}_{k,k}^\scp{\mathrm{pha}} - \frac{1}{4M(2d-1)\sin^2\theta}} \le \frac{10 d^2}{3 M(2d-1)}.
		\end{equation}
	\end{lemma}
	\begin{proof}
		We first estimate the covariance of the real and imaginary components of the Monte Carlo sampling error in Fourier coefficients. Following \cref{eqn:monte-carlo-error-in-Fourier-coef,lma:monte-carlo-error-magnitude}, for any $k \ne k^\prime$,
		\begin{equation}
		\abs{\expt{\Re(v_k)\Re(v_{k^\prime}) + \Im(v_k)\Im(v_{k^\prime})}} \le \abs{\expt{v_k \overline{v_{k^\prime}}}} \le \frac{(d\sin\theta)^2}{M(2d-1)},
		\end{equation}
		and for any $k$,
		\begin{equation}
		\frac{1-2(d\sin\theta)^2}{2M(2d-1)} \le \expt{\Re^2(v_k) + \Im^2(v_k)} = \expt{\abs{v_k}^2} \le \frac{1}{2M(2d-1)}.
		\end{equation}
		For any $k\ne k^\prime$,
		\begin{equation}
		\begin{split}
		&\abs{\expt{\Re(v_k)\Re(v_{k^\prime}) - \Im(v_k)\Im(v_{k^\prime})}}  \le \abs{\expt{v_k v_{k^\prime}}} = \abs{\frac{1}{(2d-1)^2} \sum_{j, j^\prime = 0}^{2d-2} \overline{\Omega_{kj} \Omega_{k^\prime j^\prime}} \expt{\wt{u}_j \wt{u}_{j^\prime}}}\\
		&\le  \frac{1}{(2d-1)^2} \sum_{j= 0}^{2d-2} \abs{\expt{\wt{u}^2_j }} = \frac{1}{(2d-1)^2} \sum_{j= 0}^{2d-2} \abs{\Sigma_{X,j}^2 - \Sigma_{Y,j}^2} \\
		&\le \frac{1}{M (2d-1)^2} \sum_{j= 0}^{2d-2} \mf{p}(\omega_j-\varphi,\theta) \le \frac{(d\sin\theta)^2}{M(2d-1)}.
		\end{split}
		\end{equation}
		Similarly for any $k$,
		\begin{equation}
		\begin{split}
		\abs{\expt{\Re^2(v_k) - \Im^2(v_k)}} &\le \abs{\expt{\Re^2(v_k) - \Im^2(v_k) + 2\I \Re(v_k)\Im(v_k)}}\\
		&= \abs{\expt{v_k^2}} \le \frac{1}{(2d-1)^2}\sum_{j=0}^{2d-2} \abs{\expt{\wt{u}_j^2}} \le \frac{(d\sin\theta)^2}{M(2d-1)}.
		\end{split}
		\end{equation}
		Using triangle inequality and the derived results, we have for any $k$
		\begin{equation}
		\begin{split}
		\abs{\expt{\Re^2(v_k)} - \frac{1}{4M(2d-1)}} \le& \frac{1}{2}\abs{\expt{\Re^2(v_k) + \Im^2(v_k)} - \frac{1}{2M(2d-1)}} \\
		&\quad + \frac{1}{2} \abs{\expt{\Re^2(v_k) - \Im^2(v_k)}} \le \frac{(d\sin\theta)^2}{M(2d-1)}.
		\end{split}
		\end{equation}
		The same argument is applicable to the imaginary component
		\begin{equation}
		\abs{\expt{\Im^2(v_k)} - \frac{1}{4M(2d-1)}} \le \frac{(d\sin\theta)^2}{M(2d-1)}.
		\end{equation}
		When $k \ne k^\prime$,
		\begin{equation}
		\begin{split}
		\abs{\expt{\Re(v_k)\Re(v_{k^\prime}}} \le& \frac{1}{2}\abs{\expt{\Re(v_k)\Re(v_{k^\prime}) + \Im(v_k)\Im(v_{k^\prime})}} \\
		&\quad + \frac{1}{2} \abs{\expt{\Re(v_k)\Re(v_{k^\prime}) - \Im(v_k)\Im(v_{k^\prime})}} \le \frac{(d\sin\theta)^2}{M(2d-1)}
		\end{split}
		\end{equation}
		and
		\begin{equation}
		\begin{split}
		\abs{\expt{\Im(v_k)\Im(v_{k^\prime}}} \le& \frac{1}{2}\abs{\expt{\Re(v_k)\Re(v_{k^\prime}) + \Im(v_k)\Im(v_{k^\prime})}} \\
		&\quad+ \frac{1}{2} \abs{\expt{\Re(v_k)\Re(v_{k^\prime}) - \Im(v_k)\Im(v_{k^\prime})}} \le \frac{(d\sin\theta)^2}{M(2d-1)}.
		\end{split}
		\end{equation}
		Estimating \cref{eqn:hat_c_k_star_expr} gives for any $k \ge 0$
		\begin{equation}
		1 - \frac{2}{3} (d\theta)^2 \le \hat{c}_k^\star(\theta) \le 1.
		\end{equation}
		Assuming that $d^3\theta^2 \le 1$, and applying \cref{thm:approx-coef-first-order}, it holds that for any $k$
		\begin{equation}
		\sin\theta\left(1 - \frac{8}{3} (d\theta)^2\right) \le \wt{c}_k(\theta) \le\sin\theta + 2 (d\theta)^5.
		\end{equation}
		Furthermore, if $d\theta \le \frac{1}{5}$, it holds that
		\begin{equation}
		\abs{\frac{\sin\theta}{\wt{c}_k(\theta)}} \le \frac{1}{1 - \frac{8}{3}(d\theta)^2} < \sqrt{\frac{4}{3}}.
		\end{equation}
		Then, for any $k \ne k^\prime$
		\begin{equation}
		\abs{\expt{v_k^\scp{\mathrm{pha}}v_{k^\prime}^\scp{\mathrm{pha}}}} = \frac{1}{\abs{\wt{c}_k(\theta)\wt{c}_{k^\prime}(\theta)}} \abs{\expt{\Im(v_k)\Im(v_{k^\prime}}} \le \frac{4 d^2}{3 M(2d-1)}.
		\end{equation}
		For any $k$, applying triangle inequality, it yields that 
		\begin{equation}
		\begin{split}
		& \abs{\expt{\left(v_k^\scp{\mathrm{pha}}\right)^2} - \frac{1}{4M(2d-1) \sin^2\theta}} \\
		&\le \frac{1}{\wt{c}_k^2(\theta)} \abs{\expt{\Im^2(v_k)} - \frac{1}{4M(2d-1)}} + \frac{1}{4M(2d-1)\sin^2\theta} \frac{\abs{\wt{c}_k^2(\theta) - \sin^2\theta}}{\wt{c}_k^2(\theta)}\\
		&\le \frac{\sin^2\theta}{\wt{c}_k^2(\theta)} \frac{d^2}{M(2d-1)} \left( 1 + \frac{4}{3} \frac{\theta^2}{\sin^2\theta} \left(1 + d^5\theta^4\right)\right) \le \frac{5}{2} \frac{\sin^2\theta}{\wt{c}_k^2(\theta)} \frac{d^2}{M(2d-1)} \le \frac{10 d^2}{3 M(2d-1)},
		\end{split}
		\end{equation}
		where the inequality $\frac{\theta^2}{\sin^2\theta} \le \frac{1}{25 \sin^2(1/5)} < \frac{9}{8}$ when $\theta \le \frac{1}{5d} \le \frac{1}{5}$ is used to simplify the constant.
		The proof is completed.
	\end{proof}
	
	Because of the sequential phase difference, we also need to study the structure of the covariance matrix of the colored noise in \cref{eqn:sequential-phase-difference}. It is given by the following corollary.
	\begin{corollary}\label{cor:elementwise-bound-discrete-Laplacian}
		Let
		\begin{equation}
		\wt{D} := \frac{1}{4M(2d-1)\sin^2\theta} \mf{D}, \text{ where } \mf{D}_{k, k^\prime} = \left\{\begin{array}{ll}
		2 & ,\ k = k^\prime, \\
		-1 & ,\ \abs{k-k^\prime} = 1,\\
		0 & ,\ \text{otherwise}.
		\end{array}\right.
		\end{equation}
		Then, when $d\theta \le \frac{1}{5}$ and $d^3\theta^2 \le 1$,
		\begin{equation}\label{eqn:covariance-WPA-elementwise-bound}
		\abs{D_{k,k^\prime} - \wt{D}_{k,k^\prime}} \le \frac{d^2}{M(2d-1)}  \times \left\{
		\begin{array}{ll}
		\frac{28}{3} & , k = k^\prime, \\
		\frac{22}{3} & , \abs{k - k^\prime} = 1, \\
		\frac{16}{3} & , \text{otherwise}.
		\end{array}\right.
		\end{equation}
	\end{corollary}
	\begin{proof}
		The element-wise bound \cref{eqn:covariance-WPA-elementwise-bound} follows immediately by applying triangle inequality with \cref{lma:covariance-mod-pha-estimation} and the defining equation \cref{eqn:covariance-WPA}. 
	\end{proof}
	
	Consequentially, the log-likelihood functions are well approximated by quadratic forms in terms constant matrices. The approximate forms yield the maximum likelihood estimators of QSPE.
		\begin{definition}[QSPE estimators]\label{def:estimator-qsp-pc}
		For any $k = 0, \ldots, d-2$, the sequential phase difference is defined as 
		\begin{align}\nonumber
		\Delta_k &:= \mathsf{phase}\left(c_k^\expl \overline{c_{k+1}^\expl}\right),\\
		&\text{ and }\ \boldsymbol{\Delta} := \left(\Delta_0, \Delta_1, \ldots, \Delta_{d-2}\right)^\top.
		\end{align}
		Let the all-one vector be $\boldsymbol{\mymathbb{1}} = (\underbrace{1,\ldots,1}_{d-1})^\top$ and the discrete Laplacian matrix be
		\begin{equation*}
		\mf{D} = \left(\begin{array}{rrrrr}
		2 & -1 & 0 & \cdots & 0\\
		-1 & 2 & -1 & \cdots & 0\\
		0 & -1 & 2 & \cdots & 0\\
		\vdots & \vdots & \vdots & & \vdots\\
		0 & 0 & 0 & \cdots & 2
		\end{array}\right) \in \RR^{(d-1)\times (d-1)}.
		\end{equation*}
		The statistical estimators solving QSPE are
		\begin{equation}\label{estimator_eq}
		\hat{\theta} = \frac{1}{d} \sum_{k=0}^{d-1} \abs{c_k^\expl} \quad \text{ and } \quad  \hat{\varphi} = \frac{1}{2} \frac{\boldsymbol{\mymathbb{1}}^\top \mf{D}^{-1} \boldsymbol{\Delta}}{\boldsymbol{\mymathbb{1}}^\top \mf{D}^{-1} \boldsymbol{\mymathbb{1}}}.
		\end{equation}
	\end{definition}
	Their variances can also be computed by using approximate covariance matrices, which gives
	\begin{equation}\label{eqn:var-qspc-estimator}
	\begin{split}
	& \mathrm{Var}\left(\hat{\theta}\right) \approx \frac{1}{4Md(2d-1)} \approx \frac{1}{8 M d^2} \quad \text{ and } \quad \mathrm{Var}\left(\hat{\varphi}\right) \approx \frac{3}{4Md(2d-1)(d^2-1)\theta^2} \approx \frac{3}{8Md^4\theta^2}.
	\end{split}
	\end{equation}
	In practice, an additional moving average filter in Ref. \cite{ShenLiu2019} can be applied to the data to further numerically boost the SNR. 
	For completeness, we exactly compute the optimal variance from Cram\'{e}r-Rao lower bound and discuss the optimality of the estimators in \cref{sec:lower-bound-qspc-metrology}.
	
	\subsection{Improving the estimator of swap angle using the peak information provided by \texorpdfstring{$\hat{\varphi}$}{hat varphi}}\label{appendix:peak-differentiation}
	In this subsection, we explicitly write down the dependence on $d$ as the subscript of relevant functions because $d$ is variable in the analysis.
	
	Once we have a priori $\hat{\varphi}_\mathrm{pri}$, it gives an accurate estimation of phase making $\abs{\mf{h}}$ attain its maximum, which is often referred to as the phase matching condition. The a priori phase $\hat{\varphi}_\mathrm{pri}$ can be some statistical estimator from other subroutines. For example, it can be the QSPE $\varphi$-estimator. By setting the phase modulation angle to $\omega = \hat{\varphi}_\mathrm{pri}$ in the QSPE circuit, we compute the amplitude of the reconstructed function for variable degrees and compute the differential signal by
	\begin{equation}
	\begin{split}
	&\{ \abs{\mf{h}_j^\expl} : j = d, d+2, d+4, \cdots, 3d \} \\
	&\Rightarrow \boldsymbol{\Gamma} := \left( \abs{\mf{h}_{d+2}^\expl} - \abs{\mf{h}_d^\expl}, \abs{\mf{h}_{d+4}^\expl} - \abs{\mf{h}_{d+2}^\expl}, \cdots, \abs{\mf{h}_{3d}^\expl} - \abs{\mf{h}_{3d-2}^\expl}\right)^\top \in \RR^d.
	\end{split}
	\end{equation}
	Let $\mf{D}$ be the $d$-by-$d$ discrete Laplacian matrix and $\mymathbb{1} := (1, 1, \cdots, 1) \in \RR^d$. The swap angle can be estimated by the statistical estimator
	\begin{equation}
	\hat{\theta}_\mathrm{pd} = \frac{1}{2} \frac{\mymathbb{1}^\top \mf{D}^{-1} \boldsymbol{\Gamma}}{\mymathbb{1}^\top \mf{D}^{-1} \mymathbb{1}}. 
	\end{equation}
	The performance guarantee of this estimator is given in the following theorem. We also discuss the case that the a priori is given by the QSPE estimator in the next corollary.
	\begin{theorem}\label{thm:bias-prog-diff}
		Assume an unbiased estimator $\hat{\varphi}_\mathrm{pri}$ with variance $\mathrm{Var}\left(\hat{\varphi}_\mathrm{pri}\right)$ is used as a priori. When $d\theta \le \frac{1}{9}$, the estimator $\hat{\theta}_\mathrm{pd}$ is a biased estimator with bounded bias
		\begin{equation}
		\abs{\mathrm{Bias}_\mathrm{pd}} := \abs{\expt{ \hat{\theta}_\mathrm{pd}} - \theta} \le \frac{13}{2} d^2\theta \mathrm{Var}\left( \hat{\varphi} \right) + 37 (d\theta)^3
		\end{equation}
		and variance
		\begin{equation}
		\mathrm{Var}\left( \hat{\theta}_\mathrm{pd} \right) = \frac{3}{4Md(d+1)(d+2)} \approx \frac{3}{4d^3 M}.
		\end{equation}
	\end{theorem}
	\begin{proof}
		Let the amplitude of the reconstructed function be
		\begin{equation}
		\begin{split}
			\mf{f}_d(\omega-\varphi, \theta) := \abs{\mf{h}_d(\omega;\theta,\varphi,\chi)} = \sin\theta\abs{\frac{\sin(d\sigma)}{\sin\sigma}} \sqrt{1 - \sin^2(\theta) \frac{\sin^2(d\sigma)}{\sin^2(\sigma)}},
		\end{split}
		\end{equation}
		which follows \cref{cor:modulus-magnitude-sin2theta} and $\sigma := \arccos\left(\cos(\theta)\cos(\omega-\varphi)\right)$. Furthermore, let
		\begin{equation}
			\wt{\mf{f}}^\circ_d(\omega-\varphi, \theta) := \sin\theta \frac{\sin\left(d(\omega-\varphi)\right)}{\sin(\omega-\varphi)},\quad \mf{f}^\circ_d(\omega-\varphi, \theta) := \abs{\wt{\mf{f}}^\circ_d(\omega-\varphi, \theta)}.
		\end{equation}
		Note that when $\abs{\omega-\varphi} \le \frac{\pi}{d}$, the defined function agrees with the amplitude of itself $\mf{f}^\circ_d(\omega-\varphi, \theta) = \wt{\mf{f}}^\circ_d(\omega-\varphi, \theta)$. Furthermore, for any $\omega$, we have the following bound by using triangle inequality
		\begin{equation}
			\begin{split}
				&\abs{\mf{f}_d^\circ(\omega-\varphi,\theta) - \mf{f}_d(\omega-\varphi,\theta)} \le \abs{\sin\theta \frac{\sin(d\sigma)}{\sin\sigma} \sqrt{1 - \sin^2(\theta) \frac{\sin^2(d\sigma)}{\sin^2(\sigma)}} - \sin\theta \frac{\sin\left(d(\omega-\varphi)\right)}{\sin(\omega-\varphi)}}\\
				&\le \sin\theta \abs{\frac{\sin(d\sigma)}{\sin\sigma}} \left(1 - \sqrt{1 - \sin^2(\theta) \frac{\sin^2(d\sigma)}{\sin^2(\sigma)}} \right) + \sin\theta \abs{ \frac{\sin(d\sigma)}{\sin\sigma} - \frac{\sin\left(d(\omega-\varphi)\right)}{\sin(\omega-\varphi)} }\\
				& := J_1(d) + J_2(d).
			\end{split}
		\end{equation}
		The first term can be further upper bounded by using the fact that $\max_x \abs{\frac{\sin(dx)}{\sin x}} = d$
		\begin{equation}
			J_1(d) = \frac{\sin^3\theta \abs{\frac{\sin(d\sigma)}{\sin\sigma}}^3}{1 + \sqrt{1 - \sin^2(\theta) \frac{\sin^2(d\sigma)}{\sin^2(\sigma)}}}  \le \frac{(d\theta)^3}{1 + \sqrt{1-(d\theta)^2}} \le \frac{(d\theta)^3}{1 + 2\sqrt{2}/3}
		\end{equation}
		where the last inequality uses the condition $3d\theta \le \frac{1}{3}$. The last inequality is established so that it holds for any $J_1(d), \cdots, J_1(3d)$. Note that the Chebyshev polynomial of the second kind is $U_{d-1}(\cos\sigma) = \frac{\sin(d\sigma)}{\sin\sigma}$ and it is related to the derivative of the Chebyshev polynomial of the first kind as $U_{d-1} = \frac{1}{d-1} T_{d-1}^\prime$. Using the intermediate value theorem, there exists $\xi$ in between $\cos\theta\cos(\omega-\varphi)$ and $\cos(\omega-\varphi)$ so that
		\begin{equation}
			\begin{split}
				J_2(d) &= \sin\theta\abs{U_{d-1}\left(\cos\theta\cos(\omega-\varphi)\right) - U_{d-1}\left(\cos(\omega-\varphi)\right)}\\
				&= \sin\theta \abs{U_{d-1}^\prime(\xi)}\abs{\cos(\omega-\varphi)}\left( 1 - \cos\theta \right) \le \frac{\theta^3}{2(d-1)}  \max_{-1\le x \le 1} \abs{T_{d-1}^{\prime\prime}(x)}\\
				&\le \frac{\theta^3}{2(d-1)} \frac{(d-1)^2\left((d-1)^2-1\right)}{3} \max_{-1\le x \le 1} \abs{T_{d-1}(x)} = \frac{d(d-1)(d-2)\theta^3}{6} \le \frac{(d\theta)^3}{6}.
			\end{split}
		\end{equation}
		Here, the Markov brothers' inequality (\cref{thm:Markovs-ineq}) is invoked to bound the second order derivative. Thus, the approximation error is
		\begin{equation}\label{eqn:prog-diff-err1}
			\max_{\omega \in [0,\pi]} \abs{\mf{f}_d^\circ(\omega-\varphi,\theta) - \mf{f}_d(\omega-\varphi,\theta)} \le C (d\theta)^3 \quad \text{where } C = \frac{1}{1 + 2\sqrt{2}/3} + \frac{1}{6} \approx 0.6814.
		\end{equation}
		When $\abs{\omega-\varphi} \le \frac{\pi}{d}$, the absolute value can be discarded and we can consider $\wt{\mf{f}}_d^\circ$ instead. Taking the difference of the function, it yields
		\begin{equation}\label{eqn:prog-diff-err2}
		\wt{\mf{f}}_{d+2}^\circ(\omega-\varphi, \theta) - \wt{\mf{f}}_{d}^\circ(\omega-\varphi, \theta) = 2 \sin\theta \cos\left( (d+1) (\omega-\varphi) \right).
		\end{equation}
		Let the differential signal be
		\begin{equation}
		\Gamma_d(\omega-\varphi, \theta)  := \mf{f}_{d+2}(\omega-\varphi, \theta) - \mf{f}_{d}(\omega-\varphi, \theta) = 2\theta + \delta_d(\omega-\varphi, \theta)
		\end{equation}
		where $\delta_d(\omega-\varphi, \theta)$ is the systematic error raising in the linearization of the model. Using \cref{eqn:prog-diff-err1,eqn:prog-diff-err2}, when $\abs{\omega-\varphi} \le \frac{\pi}{d+2}$, the systematic error is bounded as
		\begin{equation}\label{eqn:systematic-error-bound}
		\begin{split}
		&\abs{\delta_d(\omega-\varphi, \theta)} \le \abs{\wt{\mf{f}}_{d+2}^\circ(\omega-\varphi, \theta) - \wt{\mf{f}}_{d}^\circ(\omega-\varphi, \theta) - 2\theta} + C\theta^3\left( d^3 + (d+2)^3 \right)\\
		&\le 2\abs{\cos\left( (d+1) (\omega-\varphi) \right)}\left(\theta - \sin\theta\right) + 2\theta\left( 1 - \cos\left( (d+1) (\omega-\varphi) \right) \right) + C \theta^3\left( d^3 + (d+2)^3 \right)\\
		&\le \theta \left(d+1\right)^2\left( \omega - \varphi \right)^2 + C \theta^3\left( d^3 + (d+2)^3 \right) + 2\theta^3.
		\end{split}
		\end{equation}
		Furthermore, the differential signal is also bounded
		\begin{equation}
			\begin{split}
				&\abs{\Gamma_d(\omega-\varphi,\theta)} \le \abs{\wt{\mf{f}}_{d+2}^\circ(\omega-\varphi, \theta) - \wt{\mf{f}}_{d}^\circ(\omega-\varphi, \theta)} + C\theta^3\left( d^3 + (d+2)^3 \right)\\
				&= 2\sin\theta \abs{\cos\left( (d+1) (\omega-\varphi) \right)} + C\theta^3\left( d^3 + (d+2)^3 \right) \le 2\theta + C\theta^3\left( d^3 + (d+2)^3 \right).
			\end{split}
		\end{equation}
		In the experimental implementation, we perform the QSPE circuit with $\omega = \hat{\varphi}_\mathrm{pri}$ and degree $d, d+2, d+4, \cdots, 3d$. The resulted dataset contains $\left\{ \mf{f}_j^\expl := \abs{\mf{h}^\expl_j} : j = d, d+2, \cdots, 3d \right\}$ and the differential signal can be computed respectively
		\begin{equation}\label{eqn:prog-diff-lm}
		\Gamma_j^\expl := \mf{f}_{j+2}^\expl - \mf{f}_j^\expl = \Gamma_j(\hat{\varphi}-\varphi, \theta) + w_{j+2} - w_j = 2\theta + \delta_j(\hat{\varphi}-\varphi, \theta) + w_{j+2} - w_j
		\end{equation}
		where $w_j := \mf{f}_j^\expl - \mf{f}_j(\hat{\varphi} - \varphi)$ is the noise of the sampled data. When the SNR is large, Ref. \cite{Tretter1985} suggests the noise can be approximated by the real component of the noise on the complex-valued data $\mf{h}_j^\expl$. Analyzed in the proof of \cref{lma:monte-carlo-error-magnitude}, the variance of the noise concentrates around a constant
		\begin{equation}
		\expt{w_j} = 0 \quad \text{ and } \quad \frac{1}{4M} - \frac{(j\theta)^2}{M} \le \mathrm{Var}(w_j) \le \frac{1}{4M}.
		\end{equation}
		Assume $3d\theta \ll 1$, the covariance matrix of the colored noise $w_{j+2} - w_j$ is well approximated by a constant matrix
		\begin{equation}
		\expt{\left(w_{d + 2(j+1)} - w_{d + 2j}\right)\left(w_{d + 2(k+1)} - w_{d + 2k}\right)} \approx \frac{1}{4M} \mf{D}_{j,k}.
		\end{equation}
		Let the data vector be 
		\begin{equation}
		\boldsymbol{\Gamma} = \left( \Gamma_d^\expl, \Gamma_{d+2}^\expl, \cdots, \Gamma_{3d-2}^\expl \right)^\top \in \RR^d
		\end{equation}
		and the systematic error vector be
		\begin{equation}
		\mathrm{\boldsymbol{\delta}}(\hat{\varphi}-\varphi, \theta) = \left( \delta_d(\hat{\varphi}-\varphi, \theta), \delta_{d+2}(\hat{\varphi}-\varphi, \theta), \cdots, \delta_{3d-2}(\hat{\varphi}-\varphi, \theta) \right)^\top \in \RR^d.
		\end{equation}
		The statistical estimator solving the linearized problem of \cref{eqn:prog-diff-lm} is
		\begin{equation}
		\hat{\theta}_\mathrm{pd} = \frac{1}{2} \frac{\mymathbb{1}^\top \mf{D}^{-1} \boldsymbol{\Gamma}}{\mymathbb{1}^\top \mf{D}^{-1} \mymathbb{1}}. 
		\end{equation}
		 According to Ref. \cite{Kay1989}, the matrix-multiplication form can be exactly represented as a convex combination: for any $d$-dimensional vector $\bvec{X} = (X_0, \cdots, X_{d-1})^\top$
		\begin{equation}
			\frac{\mymathbb{1}^\top \mf{D}^{-1} \bvec{X}}{\mymathbb{1}^\top \mf{D}^{-1} \mymathbb{1}} = \sum_{k=0}^{d-1} \mu_k X_k
		\end{equation}
		where
		\begin{equation}
		\mu_k := \frac{\frac{3}{2} (d+1)}{(d+1)^2-1} \left( 1 - \left( \frac{k - \frac{d-1}{2}}{\frac{d+1}{2}}\right)^2 \right) > 0 \text{ and } \sum_{k=0}^{d-1} \mu_k = 1.
		\end{equation}
		The variance of the estimator is
		\begin{equation}
		\mathrm{Var}\left( \hat{\theta}_\mathrm{pd} \right) = \frac{1}{4} \frac{1}{4M} \frac{1}{\mymathbb{1}^\top \mf{D}^{-1} \mymathbb{1}} = \frac{3}{4Md(d+1)(d+2)} \approx \frac{3}{4d^3 M}.
		\end{equation}
		The conditional mean of the estimator is bounded as
		\begin{equation}
			\abs{\expt{\hat{\theta}_\mathrm{pd} \bigg| \hat{\varphi}_\mathrm{pri}}} = \abs{\frac{1}{2} \sum_{k=0}^{d-1} \mu_k \Gamma_{d+2k}(\hat{\varphi}_\mathrm{pri}-\varphi,\theta)} \le \frac{1}{2} \max_{k=0,\cdots,d-1} \abs{\Gamma_{d+2k}(\hat{\varphi}_\mathrm{pri}-\varphi,\theta)} \le \theta + C(3d\theta)^3.
		\end{equation}
		To make the bound in \cref{eqn:systematic-error-bound} justified, we first assume that $\abs{\hat{\varphi}_\mathrm{pri}-\varphi} \le \frac{\pi}{3d}$. Invoking Chebyshev's inequality, the assumption fails with probability
		\begin{equation}
		\bP\left(\abs{\hat{\varphi}_\mathrm{pri}-\varphi} > \frac{\pi}{3d}\right) \le \frac{9d^2}{\pi^2}\mathrm{Var}\left(\hat{\varphi}_\mathrm{pri}\right) \le d^2 \mathrm{Var}\left(\hat{\varphi}_\mathrm{pri}\right).
		\end{equation}
		When $\abs{\hat{\varphi}_\mathrm{pri}-\varphi} \le \frac{\pi}{3d}$, the conditional expectation of the estimator is
		\begin{equation}
		\expt{\hat{\theta}_\mathrm{pd} \mathbb{I}_{\abs{\hat{\varphi}_\mathrm{pri}-\varphi} \le \frac{\pi}{3d}} \bigg| \hat{\varphi}_\mathrm{pri}} = \left(\theta + \frac{1}{2} \frac{\mymathbb{1}^\top \mf{D}^{-1} \boldsymbol{\delta}(\hat{\varphi}_\mathrm{pri}-\varphi, \theta)}{\mymathbb{1}^\top \mf{D}^{-1} \mymathbb{1}}\right)\mathbb{I}_{\abs{\hat{\varphi}_\mathrm{pri}-\varphi} \le \frac{\pi}{3d}}. 
		\end{equation}
		Invoking \cref{eqn:systematic-error-bound}, when $\abs{\hat{\varphi}_\mathrm{pri}-\varphi} \le \frac{\pi}{3d}$, the bias of the estimator is bounded as
		\begin{equation}
		\begin{split}
		&\abs{\expt{\left(\hat{\theta}_\mathrm{pd} - \theta\right)\mathbb{I}_{\abs{\hat{\varphi}_\mathrm{pri}-\varphi} \le \frac{\pi}{3d}}}} = \abs{\expt{\expt{\left(\hat{\theta}_\mathrm{pd} - \theta\right)\mathbb{I}_{\abs{\hat{\varphi}_\mathrm{pri}-\varphi} \le \frac{\pi}{3d}} \bigg| \hat{\varphi}_\mathrm{pri}}}}\\
		&= \frac{1}{2}\abs{\sum_{k=0}^{d-1} \mu_k \expt{\delta_{d+2k}(\hat{\varphi}_\mathrm{pri}-\varphi, \theta)\mathbb{I}_{\abs{\hat{\varphi}_\mathrm{pri}-\varphi} \le \frac{\pi}{3d}}}} \le \frac{1}{2}\max_{k=0,\cdots,d-1} \expt{\abs{\delta_{d+2k}(\hat{\varphi}_\mathrm{pri}-\varphi, \theta)}} \\
		&\le \frac{1}{2} \theta\left(3d-1\right)^2 \mathrm{Var}\left( \hat{\varphi}_\mathrm{pri} \right) + C\theta^3\left( \frac{1}{C} + \frac{(3d-2)^3 + (3d)^3}{2} \right) \le \frac{1}{2}\theta (3d)^2 \mathrm{Var}\left( \hat{\varphi}_\mathrm{pri} \right) + C(3d\theta)^3.
		\end{split}
		\end{equation}
		On the other hand, when $\abs{\hat{\varphi}_\mathrm{pri}-\varphi} > \frac{\pi}{3d}$, the bias of the estimator is bounded as
		\begin{equation}\label{eqn:bias-upper-bound-outside}
		\begin{split}
			&\abs{\expt{\left(\hat{\theta}_\mathrm{pd} - \theta\right)\mathbb{I}_{\abs{\hat{\varphi}_\mathrm{pri}-\varphi} > \frac{\pi}{3d}}}} =  \abs{\expt{\left(\expt{\hat{\theta}_\mathrm{pd} \bigg| \hat{\varphi}_\mathrm{pri}} - \theta\right)\mathbb{I}_{\abs{\hat{\varphi}_\mathrm{pri}-\varphi} > \frac{\pi}{3d}}}}\\
			&\le \left(2\theta + C(3d\theta)^3\right) \bP\left(\abs{\hat{\varphi}_\mathrm{pri}-\varphi} > \frac{\pi}{3d}\right) \le 2d^2\theta \mathrm{Var}\left(\hat{\varphi}_\mathrm{pri}\right) + C (3d\theta)^3.
		\end{split}
		\end{equation}
		Combining these two cases and using triangle inequality, the bias is bounded as
		\begin{equation}
		\begin{split}
			\abs{\mathrm{Bias}_\mathrm{pd}} &\le \abs{\expt{\left(\hat{\theta}_\mathrm{pd} - \theta\right)\mathbb{I}_{\abs{\hat{\varphi}_\mathrm{pri}-\varphi} \le \frac{\pi}{3d}}}} + \abs{\expt{\left(\hat{\theta}_\mathrm{pd} - \theta\right)\mathbb{I}_{\abs{\hat{\varphi}_\mathrm{pri}-\varphi} > \frac{\pi}{3d}}}}\\
			&\le \frac{13}{2} d^2\theta \mathrm{Var}\left( \hat{\varphi}_\mathrm{pri} \right) + 37 (d\theta)^3.
		\end{split}
		\end{equation}
		Here, we use $54C \le 37$ to simplify the preconstant. The proof is completed.
	\end{proof}
	\begin{corollary}\label{cor:bias-prog-diff-QSPC-F}
		When $\hat{\varphi}_\mathrm{pri} = \hat{\varphi}$ is the QSPE $\varphi$-estimator in \cref{def:estimator-qsp-pc}, the bias of the estimator is bounded as
		\begin{equation}
			\abs{\mathrm{Bias}_\mathrm{pd}} \le \frac{39}{16 d^2 M \theta} + \frac{7 d\theta}{M} + 19 \left(d\theta\right)^3.
		\end{equation}
	\end{corollary}
	\begin{proof}
		The upper bound follows the substitution $\mathrm{Var}\left(\hat{\varphi}\right) \approx \frac{3}{8 d^4\theta^2 M}$. Furthermore, the second term comes from the refinement in the upper bound in \cref{eqn:bias-upper-bound-outside}
		\begin{equation}
			C (3d\theta)^3 \bP\left(\abs{\hat{\varphi}-\varphi} > \frac{\pi}{3d}\right) \le C (3d\theta)^3 d^2 \mathrm{Var}\left(\hat{\varphi}\right) \le \frac{81 C d \theta}{8 M} \le \frac{7d\theta}{M}.
		\end{equation}
	\end{proof}
	The analysis in this section indicates that trusting the a priori phase as the ``peak'' location and estimating $\theta$ from the differential signal at the ``peak'' will unavoidably introduce bias to the $\theta$-estimator. Unless the a priori is deterministic and is exactly equal to $\varphi$, the ``peak'' is not the exact peak even subjected to the controllable statistical fluctuation of $\hat{\varphi}_\mathrm{pri}$. Hence, it suggests that we need to interpret the a priori $\hat{\varphi}_\mathrm{pri}$ as an estimated peak location which is close to the exact peak location $\varphi$. This gives rise to the regression-based methods in the next subsection.
	
	\subsection{Peak regression and peak fitting}
	In order to circumvent the over-confident reliance on the a priori guess of $\varphi$, the method can be improved by regressing distinct samples with respect to analytical expressions on the unknown angle parameters. Suppose $n$ samples are made with $\{(\omega_j, d_j, \mf{h}^{\expl, j}) : j = 1, \cdots, n\}$. One can consider \REV{performing} a nonlinear regression on the data to infer the unknown parameters, which is given by the following minimization problem
	\begin{equation}
	    \hat{\theta}_\mathrm{pr}, \hat{\varphi}_\mathrm{pr}, \hat{\chi}_\mathrm{pr} = \myargmin_{\theta, \varphi, \chi} \sum_{j = 1}^n \abs{\mf{h}_{d_j}(\omega_j; \theta,\varphi,\chi) - \mf{h}^{\expl, j}}^2.
	\end{equation}
	When the number of additional samples $n$ is large enough, the estimator derived from the minimization problem is expected to be unbiased and the variance scales as $\Or(1/(d^2nM))$ according to the M-estimation theory \cite{KeenerTheoreticalStatistics2010}. However, the practical implementation of these estimators is easily affected by the complex landscape of the minimization problem. Meanwhile, the sub-optimality and the run time of black-box optimization algorithms also limit the use of these \REV{estimators}.
	
	To overcome the difficulty due to the complex landscape of nonlinear regression, we propose another technique to improve the accuracy of the swap-angle estimator by fitting the peak of the amplitude function $\mf{f}_d(\omega-\varphi, \theta)$. We observe that the amplitude function is well captured by a parabola on the interval $\mc{I} := \left[ \varphi - \frac{\pi}{2d}, \varphi + \frac{\pi}{2d} \right]$. Consider $n_\mathrm{pf}$ equally spaced sample points on the interval $\mc{I}$: $\omega_j^\mathrm{(pf)} = \hat{\varphi}_\mathrm{pri} + \frac{\pi}{d} \left(\frac{j}{n_\mathrm{pf}-1} - \frac{1}{2}\right)$ where $j = 0, 1, \cdots, n_\mathrm{pf}-1$. We find the best parabola fitting the sampled data $\mf{f}_d^\expl\left(\omega_j^\mathrm{(pf)} \right)$ whose maximum $\mf{f}_d^\mathrm{(pf\ max)}$ attains at $\omega^\mathrm{(pf\ max)}$. Given that $\hat{\varphi}_\mathrm{pri}$ is an accurate estimator of the angle $\varphi$, we accept the parabolic fitting result if the peak location does not deviate $\hat{\varphi}_\mathrm{pri}$ beyond some threshold $\varepsilon^\mathrm{thr}$, namely, the fitting is accepted if $\abs{\omega^\mathrm{(pf\ max)} - \hat{\varphi}_\mathrm{pri}} < \varepsilon^\mathrm{thr}$. Upon the acceptance, the estimator is $\hat{\theta}_\mathrm{pf} := \mf{f}_d^\mathrm{(pf\ max)} / d$. Ignoring the systematic bias caused by the overshooting of $\hat{\varphi} \neq \varphi$, the variance of the estimator is approximately $\Or\left( \frac{1}{d^2 n_\mathrm{pf}} \right)$. The detailed procedure is given in \cref{alg:qspc-peak-fitting}.
 
	\begin{algorithm}
		\caption{Improving $\theta$ estimation using peak fitting}
		\label{alg:qspc-peak-fitting}
		\begin{algorithmic}
			\STATE{\textbf{Input:} A $U$-gate $U(\theta,\varphi,\chi,*)$, an integer $d$ (the number of applications of $U$-gate), an integer $n$ (the number of sampled angles), a priori $\hat{\varphi}_\mathrm{pri}$ (can be generated by QSPE), a threshold $\beta^\mathrm{thr} \in [0, 1]$.}
			\STATE{\textbf{Output:} Estimators $\hat{\theta}_\mathrm{pf}$}
			\STATE{}
			\STATE{Initiate real-valued data vectors $\boldsymbol{\mf{p}}^\expl, \boldsymbol{\mf{w}} \in \RR^n$.}
			\FOR{$j = 0, 1, \cdots, n-1$}
			\STATE{Set the tunable $Z$-phase modulation angle as $\omega_j = \hat{\varphi}_\mathrm{pri} + \frac{\pi}{d}\left(\frac{j}{n-1}-\frac{1}{2}\right)$.}
			\STATE{Peform the quantum circuit in Figure 1 in the main text (or \cref{fig:general-QSPE}) and measure the transition probabilities $p_X^\expl(\omega_j)$ and $p_Y^\expl(\omega_j)$.}
			\STATE{Set $\boldsymbol{\mf{p}}^\expl_j \leftarrow \sqrt{\left(p_X^\expl(\omega_j) - \frac{1}{2}\right)^2 + \left(p_Y^\expl(\omega_j) - \frac{1}{2}\right)^2}$ and $\boldsymbol{\mf{w}}_j \leftarrow \omega_j$.}
			\ENDFOR
			\STATE{Fit $\left(\boldsymbol{\mf{w}}, \boldsymbol{\mf{p}}^\expl\right)$ with respect to to parabolic model $\mf{p} = \beta_0\left(\mf{w} - \beta_1\right)^2 + \beta_2$.}
            \IF{$\beta_0 < 0$ (concavity) and $\abs{\beta_1 - \hat{\varphi}_\mathrm{pri}} < \beta^\mathrm{thr}$ (small deviation from a priori)}
			\STATE{Set $\hat{\theta}_\mathrm{pf} \leftarrow \beta_2 / d$. The improvement is accepted.}
			\ELSE
			\STATE{Set $\hat{\theta}_\mathrm{pf} \leftarrow \mathrm{None}$. The improvement is rejected.}
			\ENDIF
		\end{algorithmic}
	\end{algorithm}
 
	\subsection{Numerical performance of QSPE on \fsim against Monte Carlo sampling error}\label{subsec:num-result-MC}
	To numerically test the performance of QSPE and validate the analysis in the presence of Monte Carlo sampling error, we simulate the quantum circuit and perform the inference. In \cref{fig:degree_mc}, we plot the squared error of each estimator as a function of the number of \fsim{}s $d$ in each quantum circuit. Consequentially, each data point is the mean squared error (MSE), which is a metric of the performance according to the bias-variance decomposition $\mathrm{MSE} = \mathrm{Var} + \mathrm{bias}^2$. The numerical results in \cref{fig:degree_mc} indicate that although $\theta = 1\times 10^{-3}$ is small, QSPE estimators achieve an accurate estimation with a very small $d$. The numerical results also show that the performance of the estimator does not significantly depend on the value of the single-qubit phase $\varphi$. Meanwhile, using the peak fitting in \cref{alg:qspc-peak-fitting}, the variance in $\theta$-estimation is improved so that the MSE curve is lowered. Zooming the MSE curve in log-log scale, the curve scales as a function of $d$ as the theoretically derived variance scaling in Theorem 2 in the main text. We will discuss the scaling of the variance in \cref{sec:lower-bound-qspc-metrology} in more detail.
	
	In \cref{fig:meas_mc}, we perform the numerical simulation with variable swap angle $\theta$ and \REV{the} number of measurement samples $M$. The numerical results show that the accuracy of $\varphi$-estimation is more vulnerable to decreasing $\theta$. This is explainable from the theoretically derived variance in Theorem 2 in the main text which depends on the swap angle as $1/\theta^2$. Although the theoretical variance of $\theta$ is expected to be invariant for different $\theta$ values, the numerical results show that the MSE of $\theta$-estimation gets larger when smaller $\theta$ is used, and the scaling of the curve differs from the classical scaling $1/M$. The reason is that when $\theta \le 5\times 10^{-4}$, the SNR is not large enough so that the theoretical derivation can be justified. When using a bigger $d$ or $M$, the curve will converge to the theoretical derivation. When $\theta = 1\times 10^{-3}$, the setting of the experiments is enough to get a large enough SNR. Hence, the scaling of the MSE curves in the bottom panels in \cref{fig:meas_mc} agrees with the classical scaling $1/M$ of Monte Carlo sampling error.
	
	\begin{figure}[h]
		\centering
		\includegraphics[width=\columnwidth]{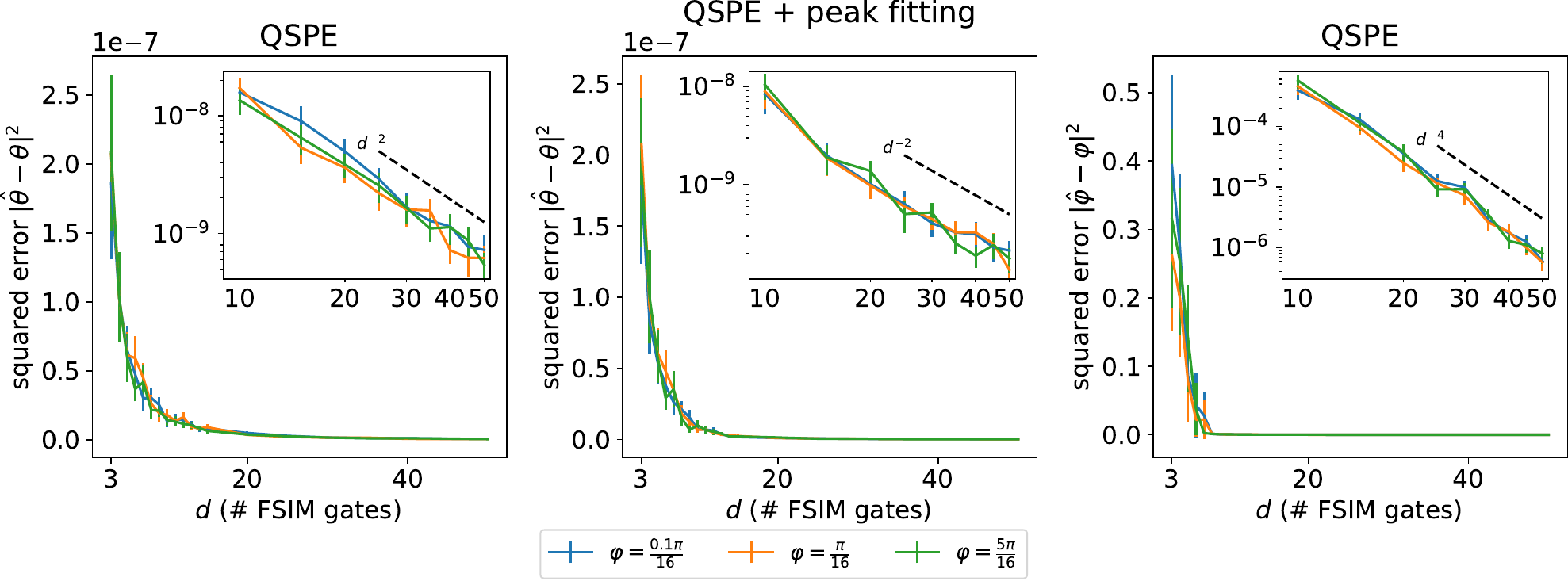}
		\caption{Squared error of estimators as a function of the number of \fsim's. The only source of noise in the numerical experiments is \REV{the} Monte Carlo sampling error. The number of measurement samples is $M = 1\times 10^5$, and $n_\mathrm{pf} = 15$ is used in the peak fitting. The swap angle is set to $\theta = 1 \times 10^{-3}$ and the  phase parameter is set to $\chi = 5\pi/32$. The error bar of each point stands for the confidence interval derived from $96$ independent repetitions.}
		\label{fig:degree_mc}
	\end{figure}
	
	\begin{figure}[h]
		\centering
		\includegraphics[width=\columnwidth]{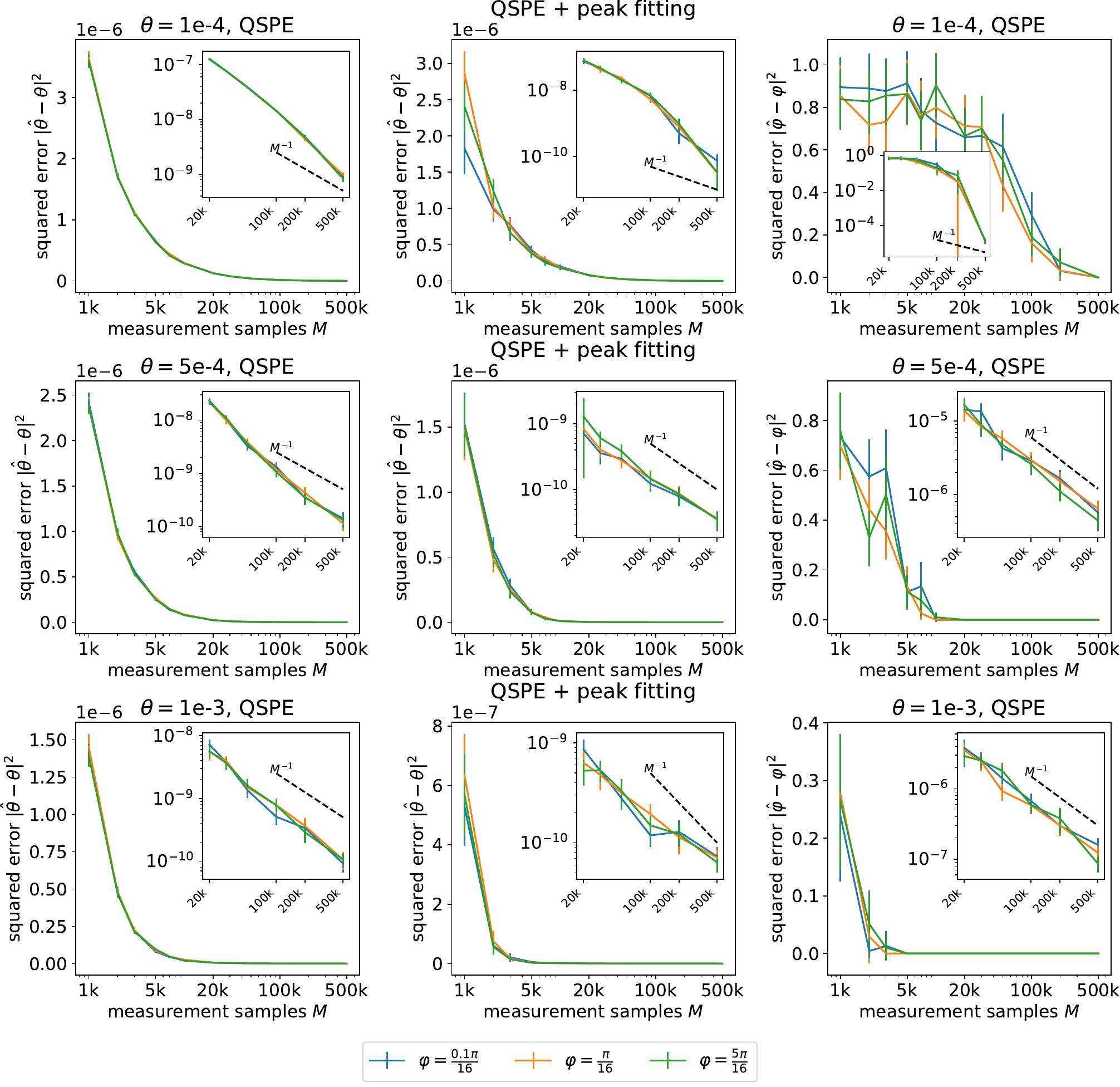}
		\caption{Squared error of estimators as a function of the number of measurement samples. The only source of noise in the numerical experiments is \REV{the} Monte Carlo sampling error. The circuit degree is set to $d = 50$ and the \fsim phase parameter is set to $\chi = 5\pi/32$. $n_\mathrm{pf} = 15$ is used in the peak fitting. The error bar of each point stands for the confidence interval derived from $96$ independent repetitions.}
		\label{fig:meas_mc}
	\end{figure}

    \section{Solving QSPE with arbitrary swap-angle value}\label{sec:qspe-general-theta}
    In the early section, we derive robust statistical estimators when $d \theta \le 1$. These estimators are based on linear statistical models which \REV{rely} on the approximation of Fourier coefficients in the desired regime. In this section, we aim to generalize the solution to arbitrary swap-angle value which expands the use of QSPE to generic $U$-gates. 

    Recall that in \cref{thm:reconstruction-h-Fourier-expansion}, we show that by performing Fourier transformation, the swap angle and phase angles are fully decoupled in terms of the dependencies in amplitude and phase. Furthermore, the analysis in \cref{sec:Monte-Carlo-sampling-error} reveals that the noise magnitude, namely, the variance of the noise, in the Fourier space is reduced by a factor of $d$. Consequently, these results suggest the design of algorithms in Fourier space. 

    According to \cref{thm:reconstruction-h-Fourier-expansion,eqn:tilde-c-k-expr}, the amplitudes of Fourier coefficients are exactly computable by solving the integral:
    \begin{equation}\label{eqn:def-eqn-Aks}
        A_k(\theta) := \wt{c}_k(\theta) = \frac{\sin\theta}{\pi} \int_0^\pi e^{-2\I(k+1)\omega} P_\omega^\scp{d}(\cos\theta) Q_\omega^\scp{d}(\cos\theta) \rd \omega.
    \end{equation}
    It is worth noting that the integrand is a finite-degree Laurent polynomial in terms of $e^{\pm \I \omega}$. Consequently, these coefficients can be exactly computed efficiently using $\Or(d \log(d))$ floating-point operations. Given the experimental amplitudes $\{\abs{c_k^\expl} : k = -d + 1, \cdots, d - 1\}$, the problem can be modeled as a system of nonlinear equations:
    \begin{equation}
        A_k(\theta) = \alpha_k \cdot \abs{c_k^\expl}, \quad k = -d + 1, \cdots, d - 1.
    \end{equation}
    Here, $\alpha_k \in \{-1, 1\}$ is an undetermined sign which is dropped when taking the amplitude. When formulating the nonlinear equation, we manually add it back. The set of candidate solutions of the $k$-th equation is $A^{-1}_k(\alpha_k \cdot \abs{c_k^\expl})$ which may contain more than one values as $A_k$ may not be injective. However, the solution should satisfy all equations. Hence, the solution set is $\cup_{\{\alpha_k\}} \cap_{k = - d + 1}^{d - 1} A^{-1}_k(\alpha_k \cdot \abs{c_k^\expl})$. The difficulties of solving this problem lie in three major aspects:
    \begin{enumerate}
        \item Sign problem. The undetermined sign factor renders it different from classic nonlinear equation problems. There are seemingly $2^{2 d - 1}$ potential systems of nonlinear equations associated with $\{\alpha_k\}$ to be solved. However, most of them do not have solution, namely, $\cap_{k = - d + 1}^{d - 1} A^{-1}_k(\alpha_k \cdot \abs{c_k^\expl}) = \emptyset$.
        \item Nonconvexity. Each function $A_k$ is nonconvex which contributes to the difficulty in solving nonlinear equation. Furthermore, the function is not injective which makes $A_k^{-1}$ multi-valued. These render the numerical solution to the system of nonlinear equations \REV{challenging}. 
        \item Imperfection due to noise and error. It is worth noting that the experimental values $\{\abs{c_k^\expl}\}$ may not match exact values due to the presence of sampling error and noise. The deviation from the exact value might make taking intersection among solutions to each equation hard because they may not match exactly. To account for noise and error, and to make the solution robust, we may slightly relax the range of the solution $\mathcal{S}_{k, \gamma}(\alpha_k \cdot \abs{c_k^\expl}) := \{ A_k^{-1}(\alpha_k \cdot \abs{c_k^\expl}) + x : \abs{x} \le \gamma \}$. Then, the set of relaxed solution is $\cup_{\{\alpha_k\}} \cap_{k = - d + 1}^{d - 1} \mathcal{S}_{k, \gamma}(\alpha_k \cdot \abs{c_k^\expl})$. The relaxation radius $\gamma$ \REV{accounts} for the uncertainty due to potential noise and error. It is worth noting that the choice of the relaxation radius $\gamma$ is a tradeoff between noise magnitude and solution precision which is also referred to as bias-variance tradeoff. When $\gamma$ is large, the variance of the solution is small which is robust, but the bias of the solution is large. Conversely, when $\gamma$ is small, the variance of the solution is large as it is more vulnerable to noise but the bias is small.
    \end{enumerate}

    It is worth noting that this hardness does not apply to the case when $d \theta \le 1$. According to the analysis in the previous sections, $A_k(\theta) \approx \theta \mathbb{I}_{k \ge 0}$ which is linear and positive. Hence, the sign problem and the nonconvexity in general cases are not applicable. The third issue is resolved by taking average in the construction of the statistical estimator which is proven to be robust against sampling error.
    
    As a consequence of the previous discussion, we can consider a simple interval-based solution to the general problem which approximately costs only $\Or(d \epsilon^{-1})$ operations where $\epsilon$ is the target precision. The procedure of the algorithm is as follows. Let the full $\theta$-range partitioned into small intervals $\{[l_i, r_i) : r_i - l_i \le \epsilon\}$. Then, by continuity, we are able to determine whether there exists points in this interval satisfying certain nonlinear equation by examining $\min\{A_k(l_i), A_k(r_i)\} \le \alpha_k \cdot \abs{c_k^\expl} \le \max\{A_k(l_i), A_k(r_i)\}$ for $\alpha_k = \pm 1$. We maintain a counter to count the number of satisfactions in each interval. The counter can be derived with $\Or(d \log(d) \epsilon^{-1})$ operations. Finally, we scan all intervals to get all intervals with full $2d-1$ satisfactions of equations and set the potential solution to $\theta = (l_i + r_i) / 2$. 

    This procedure is depicted in \cref{fig:generalized_theta_estimation} where each panel stands for an individual nonlinear equation with two distinct \REV{choices} of $\alpha_k = \pm 1$. In each panel, all intersection points form the potential solution set $A_k^{-1}(\alpha_k \cdot \abs{c_k^\expl})$. As the degree parameter is set to $d = 5$, the system has nine nonlinear equations. By solving them using the previously discussed interval-based method, we obtain the counter which is visualized as a histogram in \cref{fig:generalized_theta_estimation_interval_count}. We see the full satisfaction is attained at $\theta = 1$ and $\theta = \pi - 1$. They form the final solution output of the algorithm. It is worth noting that these two $\theta$-values are equivalent as the distinction in negating cosine can be accounted by redefining other phase angles.

    To estimate phase angle $\varphi$, we note that \cref{thm:reconstruction-h-Fourier-expansion} indicates that the $\varphi$-estimation procedure in \REV{the} general case is identical to that when $\theta$ is small. Hence, we may still use the derived estimator to estimate $\varphi$ by taking the sequential phase difference in the phases of Fourier coefficients. However, we note that the $\theta$-value may affect the estimation accuracy of $\varphi$ as the value of $A_k(\theta)$ \REV{modulates} the magnitude of the Fourier mode and \REV{limits} the signal-to-noise ratio as revealed in the analysis in \cref{sec:Monte-Carlo-sampling-error}.

    In \cref{fig:generalized_theta_varphi_estimation_error_meas1e5_tol1e-3}, we test the algorithm performance with variable $\theta$-values. We see that the absolute error in $\theta$-estimation remains well bounded below $5 \times 10^{-4}$ except for the singularity near $\pi / 2$. At this value, the variable part of the transition probabilities vanishes and $\mf{h} = 0$. Hence, angle inference becomes increasingly challenging due to the absence of information. It is also worth noting that as $\theta$ gets close to $\pi / 2$, the magnitude of $\cos(\theta)$ becomes more vanishing. Hence, the signal strength of the phase angle dependent part is increasingly weakened \REV{compared} to noise. This low signal-to-noise ratio leads to the increase in the $\varphi$-estimation error in \cref{fig:generalized_theta_varphi_estimation_error_meas1e5_tol1e-3}. 

\begin{algorithm} 
\caption{Inferring unknown angles in $U$-gate with general swap angle using QSPE procedure}
\label{alg:qspe-general-theta}
\begin{algorithmic}
    \STATE{\textbf{Input:} A  $U$-gate $U(\theta,\varphi,\chi,*)$, an integer $d$ (the number of applications of the $U$-gate), a precision parameter $\epsilon$.}
    \STATE{\textbf{Output:} Estimates $\hat{\theta}, \hat{\varphi}$}
    \STATE{}
    \STATE{Initiate a complex-valued data vector $\boldsymbol{\mf{h}}^\expl \in \CC^{2d-1}$.}
    \FOR{$j = 0, 1, \cdots, 2d-2$}
    \STATE{Set the tunable $Z$-phase modulation angle as $\omega_j = \frac{j}{2d-1} \pi$.}
    \STATE{Perform the quantum circuit in Figure 1 in the main text (or \cref{fig:general-QSPE}) and measure the transition probabilities $p_X^\expl(\omega_j)$ and $p_Y^\expl(\omega_j)$.}
    \STATE{Set $\boldsymbol{\mf{h}}^\expl_j \leftarrow p_X^\expl(\omega_j) - \frac{1}{2} + \I\left(p_Y^\expl(\omega_j) - \frac{1}{2}\right)$.}
    \ENDFOR
    \STATE{Compute the Fourier coefficients $\bvec{c}^\expl = \mathsf{FFT}\left(\boldsymbol{\mf{h}}^\expl\right)$.}
    \STATE{Compute estimates $\hat{\varphi}$ according to \cref{def:estimator-qsp-pc} using $\mathsf{phase}(\bvec{c}^\expl)$.}
    \STATE{Set number of intervals to $m = \lceil \pi / \epsilon \rceil$ and initiate a all-zero counter $z \in \RR^m$.}
    \STATE{Set $A_k^{(l)} = 0$ for $k = - d + 1, \cdots, d - 1$}
    \FOR{$j = 0, \cdots, m - 1$}
        \STATE{Set $l = j \pi / m$ and $r = (j + 1) \pi / m$.}
        \STATE{Set $A_k^{(r)} = A_k(r)$ which is derived by solving \cref{eqn:def-eqn-Aks} with FFT.}
        \FOR{$\alpha_k = \pm 1$}
        \IF{$\min\{A_k^{(l)}, A_k^{(r)}\} \le \alpha_k \cdot \abs{c_k^\expl} \le \max\{A_k^{(l)}, A_k^{(r)}\}$}
            \STATE{$z_j \leftarrow z_j + 1$}
        \ENDIF
        \ENDFOR
        \STATE{Set $A_k^{(r)} \leftarrow A_k^{(l)}$.}
    \ENDFOR
    \STATE{Initiate a list $\hat{\theta} = []$ and set $z_{\max} = \max_j z_j$.}
    \FOR{$j = 0, \cdots, m - 1$}
    \IF{$z_j = z_{\max}$}
        \STATE{Append $(j + 1 / 2) \pi / m$ into $\hat{\theta}$.}
    \ENDIF
    \ENDFOR
\end{algorithmic}
\end{algorithm}

    \begin{figure}[htbp]
        \centering
        \includegraphics[width=\linewidth]{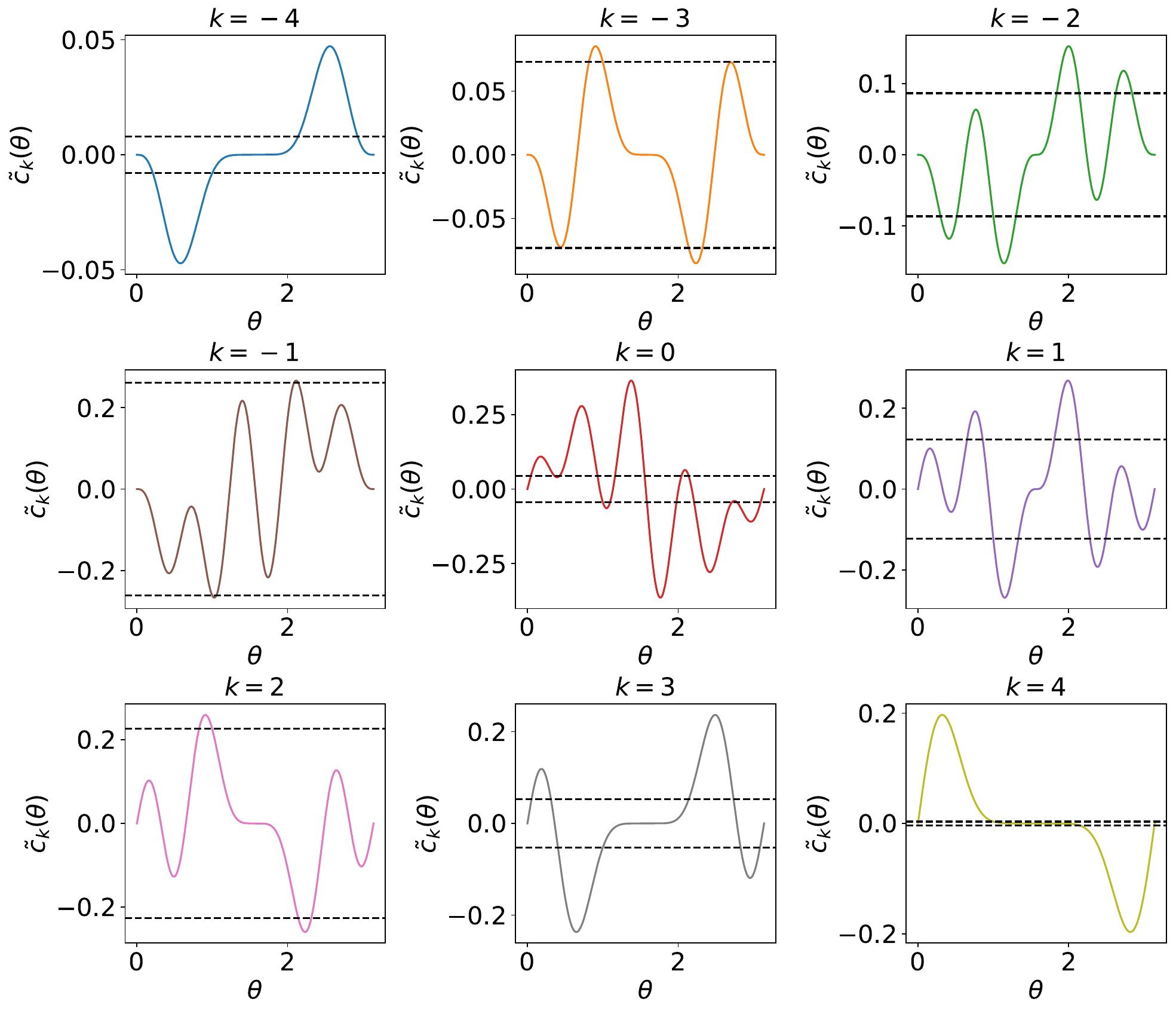}
        \caption{An example of formulating QSPE as solving systems of nonlinear equations. We set $d = 5$ and $M = 1 \times 10^5$. The relevant angles are set to $\theta = 1, \varphi = \pi / 16$. In each panel, the solid curve is the amplitude of the Fourier coefficients derived by solving the defining integral. The two horizontal dashed lines stand for the experimental amplitude with positive or negative signs ($\alpha_k = \pm 1$). The intersecting points are candidate solutions $A_k^{-1}(\alpha_k \cdot \abs{c_k^\expl})$.}
        \label{fig:generalized_theta_estimation}
    \end{figure}

    \begin{figure}[htbp]
        \centering
        \includegraphics[width=\linewidth]{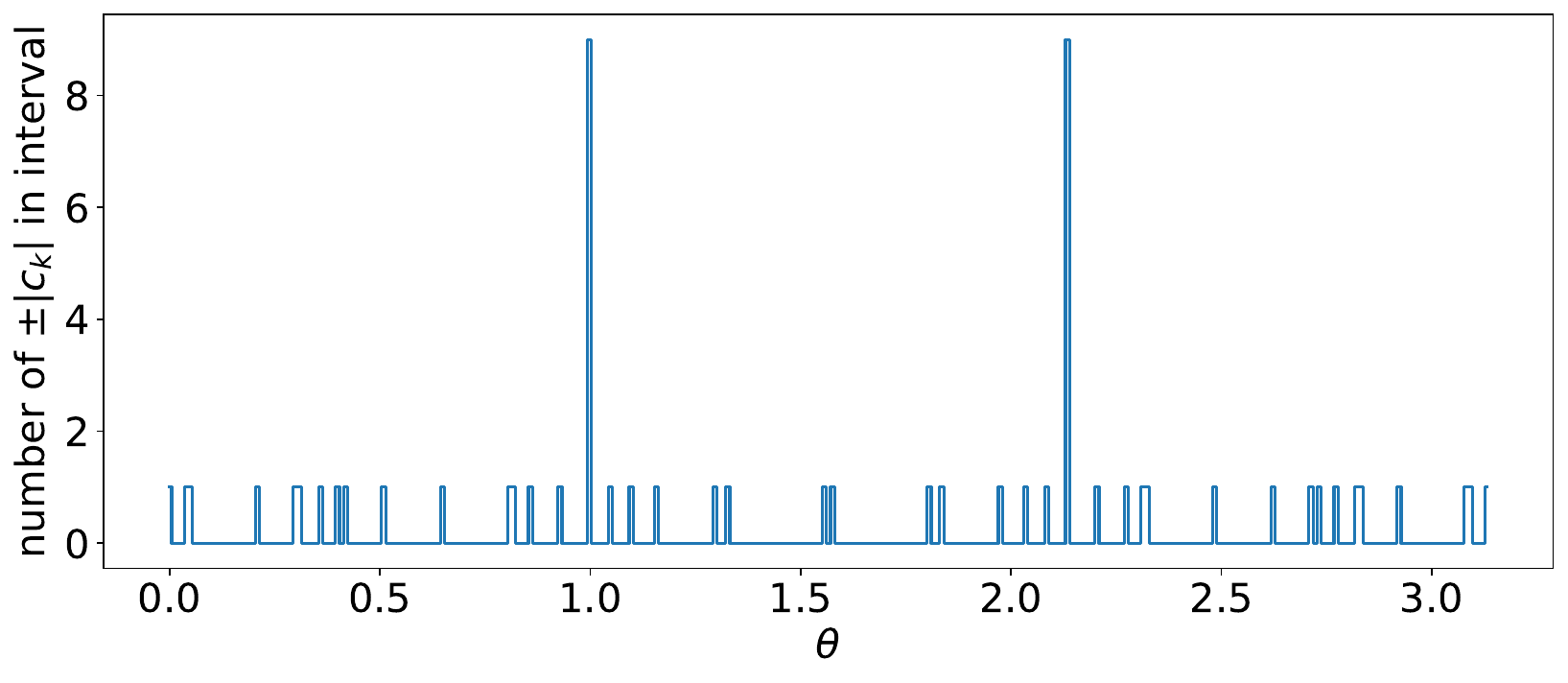}
        \caption{An example of the counter of interval-based algorithm for solving QSPE. The setup is identical to that in \cref{fig:generalized_theta_estimation}. The $\theta$-values satisfy all nine nonlinear equations are the output of the algorithm, which are estimators of the swap angle. }
        \label{fig:generalized_theta_estimation_interval_count}
    \end{figure}

    \begin{figure}[htbp]
        \centering
        \includegraphics[width=\linewidth]{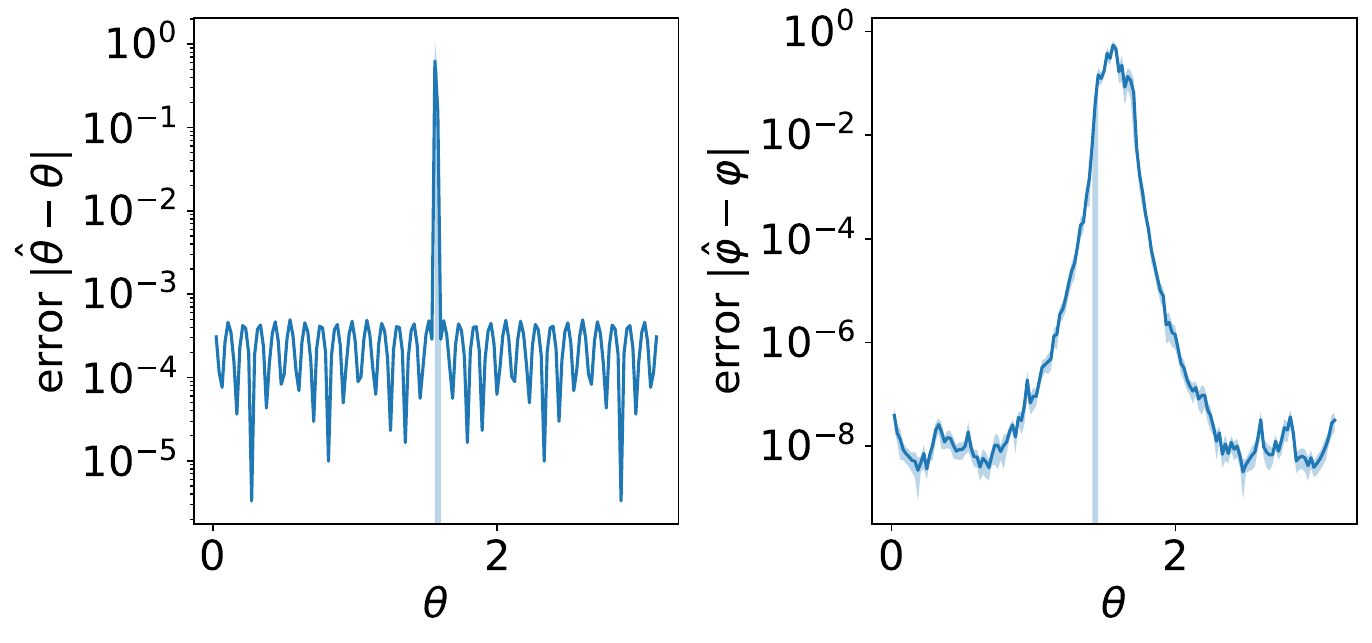}
        \caption{Estimation errors of $\theta$ and $\varphi$ using the generalized QSPE estimation algorithm for general swap angles. We set $d = 5$ and $M = 1 \times 10^5$. We fix $\varphi = \pi / 16$ and make $\theta$ variable. The shaded area stands for the confidence interval (error bar) determined from ten independent repetitions.}
        \label{fig:generalized_theta_varphi_estimation_error_meas1e5_tol1e-3}
    \end{figure}

	\section{Lower bounding the performance of QSPE}\label{sec:lower-bound-qspc-metrology}
	In the designed phase estimation algorithm, gate parameters are estimated from experimental data by running $2(2d-1)$ quantum circuits whose depths are $\Theta(d)$. If we simply think under the philosophy of the Heisenberg limit of quantum metrology in Ref. \cite{Lloyd2006}, we would expect the variance of statistical estimators bounded from below as 
	$$\Omega\left(1/\left((\text{classical repetition})\times(\text{quantum repetition})^2)\right)\right) = \Omega(1/d^3)$$
	when $d$ is large enough. However, theoretical analysis in Theorem 2 in the main text and numerical simulation in \cref{fig:degree_mc} show that the variance of the $\varphi$-estimator in QSPE depends on the parameter $d$ as $\mathrm{Var}\left(\hat{\varphi}\right) \sim 1/d^4$. In this section, we will analyze this nontrivial counterintuitive result. In the end, we prove that for a fixed unknown \fsim, the $1/d^4$-dependency only appears in the pre-asymptotic regime where the condition of the theorems holds, i.e., $d \theta \ll 1$. When passing to the limit of large enough $d$, the variances of statistical estimators agree with that suggested by the Heisenberg limit. Although such faster than Heisenberg limit scaling only applies in a finite range of circuit depth~($d \theta \ll 1$), it has drastically increased our metrology performance in practice against time-dependent errors, and thus deserves further investigation in its generalization to other domains of noise learning.
	
	\subsection{Pre-asymptotic regime \texorpdfstring{$d \ll 1/\theta$}{d ll 1/theta}}\label{subsec:CRLB-pre-asymptotic}
	We derive the optimal variance scaling permitted using our metrology method in finite circuit depth, i.e. pre-asymptotic regime in this subsection. More particularly, we require that for a given range of gate parameter $\theta \in [\theta_{\text{min}},\theta_{\text{max}}] $, our metrology circuit depth obeys: $d \ll 1/\theta_{\text{min}}$ in the pre-asymptotic regime. This also implies that for any $\theta$ under the consideration we have $d\theta \ll 1$.
	
	The quantum circuits in QSPE form a class of parametrized quantum circuits whose measurement probabilities are trigonometric polynomials in a tunable variable $\omega$. For simplicicty, the gate parameters of the unknown \fsim is denoted as $\Xi = (\xi_k) = (\theta, \varphi, \chi)$. According to the modeling of Monte Carlo sampling error in \cref{lma:monte-carlo-error-magnitude}, the experimentally estimated probabilities are approximately normal distributed. Given the normality and assuming the limit $M \gg 1$, the element of the Fisher information matrix is
	\begin{equation}\label{eqn:Fisher-information-element-definition}
		I_{kk^\prime}(\Xi) = \sum_{j=0}^{2d-2} \Sigma_{X, j}^{-2} \frac{\partial p_X(\omega_j; \Xi)}{\partial \xi_k}\frac{\partial p_X(\omega_j; \Xi)}{\partial \xi_{k^\prime}} + \sum_{j=0}^{2d-2} \Sigma_{Y, j}^{-2} \frac{\partial p_Y(\omega_j; \Xi)}{\partial \xi_k}\frac{\partial p_Y(\omega_j; \Xi)}{\partial \xi_{k^\prime}}.
	\end{equation}
	According to \cref{eqn:variance-mc-error-concentration}, the variance of the Monte Carlo sampling error concentrates near a constant. Hence
	\begin{equation}
	I_{kk^\prime}(\Xi) = 4M\left(1 + \Or(d^2\theta^2)\right) \sum_{j=0}^{2d-2} \left(\frac{\partial p_X(\omega_j; \Xi)}{\partial \xi_k}\frac{\partial p_X(\omega_j; \Xi)}{\partial \xi_{k^\prime}} + \frac{\partial p_Y(\omega_j; \Xi)}{\partial \xi_k}\frac{\partial p_Y(\omega_j; \Xi)}{\partial \xi_{k^\prime}}\right).
	\end{equation}
	Using the reconstructed function, the element of the Fisher information matrix can be expressed as
	\begin{align}
		I_{kk^\prime}(\Xi) &= 4M\left(1 + \Or(d^2\theta^2)\right) \Re\left(\sum_{j=0}^{2d-2} \frac{\partial \mf{h}(\omega_j; \Xi)}{\partial \xi_k}\frac{\partial \overline{\mf{h}(\omega_j; \Xi)}}{\partial \xi_{k^\prime}}\right)\\
		&= 4M(2d-1)\left(1 + \Or(d^2\theta^2)\right) \Re\left(\sum_{j=-d+1}^{d-1} \frac{\partial c_j(\Xi)}{\partial \xi_k}\frac{\partial \overline{c_j(\Xi)}}{\partial \xi_{k^\prime}}\right) \label{eqn:Fisher-information-Fourier-coefficients}\\
		&= \frac{4M(2d-1)}{\pi} \left(1 + \Or(d^2\theta^2)\right) \Re\left(\int_{-\pi/2}^{\pi/2} \frac{\partial \mf{h}(\omega; \Xi)}{\partial \xi_k}\frac{\partial \overline{\mf{h}(\omega; \Xi)}}{\partial \xi_{k^\prime}} \rd \omega\right). \label{eqn:Fisher-information-integral}
	\end{align}
	Here, we use the construction of QSPE in which the tunable angles are equally spaced in one period of the reconstructed function. The second equality (\cref{eqn:Fisher-information-Fourier-coefficients}) invokes \cref{thm:structure-of-qsp-pc} and the discrete orthogonality of Fourier factors. The last equality (\cref{eqn:Fisher-information-integral}) is due to the Parseval's identity. 
	
	When $d\theta \ll 1$ and $\theta \ll 1$, the Fourier coefficients are well captured by the approximation in \cref{thm:structure-of-qsp-pc} which gives $c_j(\Xi) \approx \I e^{-\I \chi} e^{-\I(2j+1)\varphi} \theta \bI_{j \ge 0}$. Consequentially, using \cref{eqn:Fisher-information-Fourier-coefficients}, in the pre-asymptotic regime $d \ll 1/\theta$, the Fisher information matrix is approximately
	\begin{equation}
	    I(\Xi) \approx 4M(2d-1) \left(\begin{array}{ccc}
	        d & 0 & 0 \\
	        0 & \frac{d(4d^2-1)}{3}\theta^2 & d^2\theta^2\\
	        0 & d^2\theta^2 & d \theta^2
	    \end{array}\right).
	\end{equation}
	Invoking Cram\'{e}r-Rao bound, the covariance matrix of any statistical estimator is lower bounded as
	\begin{equation}
	   \mathrm{Cov}\left(\hat{\theta}_\text{any}, \hat{\varphi}_\text{any}, \hat{\chi}_\text{any}\right) \succeq I^{-1}(\Xi) \approx \frac{1}{4Md(2d-1)} \left(\begin{array}{ccc}
	        1 & 0 & 0 \\
	        0 & \frac{3}{(d^2-1)\theta^2} & -\frac{3d}{(d^2-1)\theta^2}\\
	        0 & -\frac{3d}{(d^2-1)\theta^2} & \frac{4d^2-1}{(d^2-1)\theta^2}
	    \end{array}\right).
	\end{equation}
	Consequentially, in the pre-asymptotic regime, the optimal variances of the statistical estimator are 
	\begin{align}
	    &\mathrm{Var}\left(\hat{\theta}_\text{opt}\right) = \frac{1}{4Md(2d-1)} \approx \frac{1}{8 M d^2},\label{eqn:opt-theta-var}\\
	    &\mathrm{Var}\left(\hat{\varphi}_\text{opt}\right) = \frac{3}{4Md(2d-1)(d^2-1)\theta^2} \approx \frac{3}{8Md^4\theta^2},\label{eqn:opt-varphi-var}\\
	    &\mathrm{Var}\left(\hat{\chi}_\text{opt}\right) = \frac{1}{4Md(2d-1)\theta^2} \frac{4d^2-1}{d^2-1} \approx \frac{1}{2Md^2\theta^2}.\label{eqn:opt-chi-var}
	\end{align}
	
	Remarkably, the variances of QSPE estimators in Theorem 2 in the main text exactly match the optimality given in \cref{eqn:opt-theta-var,eqn:opt-varphi-var}. We thus proves the optimality of our QSPE estimator for inferring gate parameter $\theta$ and $\varphi$. Moreover, we like to point out that the faster than Heisenberg-limit scaling of parameter $\varphi$ in this asymptotic regime is critical to the successful  experimental deployment of our methods. This is because the dominant time-dependent error results in a   time-dependent drift error in $\varphi$, and a faster convergence in circuit depth provides faster metrology runtime to minimize such drift error during the measurements. 
	
	\subsection{Asymptotic regime \texorpdfstring{$d \to \infty$}{L to infty}}
	Thinking under the framework of Heisenberg limit in Ref. \cite{Lloyd2006}, for a fixed $\theta$, the optimal variances of $\theta$ and $\varphi$ estimators are expected to scale as $1/d^3$ while that of $\chi$ estimator scales as $1/d$ due to the absence of amplification in the quantum circuit. In contrast to these scalings, we show in the last subsection that the scalings of $\varphi$ and $\chi$ estimators can achieve $1/d^4$ and $1/d^2$ in the pre-asymptotic regime $d \ll 1/\theta$. In this subsection, we will argue that the scalings predicted by the Heisenberg scaling hold if further passing to the asymptotic limit $d \to \infty$. As a consequence, there is a nontrivial transition of variance scalings of QSPE estimators in pre-asymptotic regime and the asymptotic regime. We demonstrate such subtle transition in the fundamental   efficiency allowed for the given metrology protocol with both numerical simulation and analytic reasoning in this  section.
	
	As $d \to \infty$, the measurement probabilities no longer admit the property of concentration around constants. Using the variance derived in \cref{eqn:variance-mc-error-concentration}, the diagonal element of Fisher information matrix is exactly equal to
	\begin{equation}
	\begin{split}
	    &I_{kk}(\Xi) = M \sum_{j=0}^{2d-2} \left(\frac{1}{p_X(\omega_j; \Xi)\left(1-p_X(\omega_j; \Xi)\right)} \frac{\partial p_X(\omega_j; \Xi)}{\partial \xi_k}\frac{\partial p_X(\omega_j; \Xi)}{\partial \xi_k}\right.\\
	    &\hspace*{5em}\ignorespaces\left.+ \frac{1}{p_Y(\omega_j; \Xi)\left(1-p_Y(\omega_j; \Xi)\right)} \frac{\partial p_Y(\omega_j; \Xi)}{\partial \xi_k}\frac{\partial p_Y(\omega_j; \Xi)}{\partial \xi_k}\right)\\
	    &= M \sum_{j=0}^{2d-2} \left(- \frac{\partial \log p_X(\omega_j; \Xi)}{\partial \xi_k}\frac{\partial \log\left(1- p_X(\omega_j; \Xi)\right)}{\partial \xi_k} - \frac{\partial \log p_Y(\omega_j; \Xi)}{\partial \xi_k}\frac{\partial \log\left(1- p_Y(\omega_j; \Xi)\right)}{\partial \xi_k} \right).
	\end{split}
	\end{equation}
Moreover $p_X(\omega_j; \Xi)$ and $p_Y(\omega_j; \Xi)$ are trigonometric polynomials in $\theta$ and $\varphi$ of degree at most $d$ while in $\chi$ of degree $1$ due to the absence of amplification. Therefore the log-derivatives of $\theta$ and $\varphi$ are $\Or(d)$ in most regular cases while they are $\Or(1)$ for $\chi$. Hence, we expect from the Cram\'{e}r-Rao bound that
	\begin{equation}\label{eqn:crlb-asymptotic}
	    \mathrm{Var}\left(\hat{\theta}_\mathrm{opt}\right), \mathrm{Var}\left(\hat{\varphi}_\mathrm{opt}\right) = \Omega\left(\frac{1}{d^3}\right),\quad \text{and } \mathrm{Var}\left(\hat{\chi}_\mathrm{opt}\right) = \Omega\left(\frac{1}{d}\right) \quad \text{as } d \to \infty.
	\end{equation}
	These results match the scalings predicted by the Heisenberg limit which holds in the asymptotic limit $d \to \infty$.
	
	\subsection{Numerical results}
	
	We compute the Cram\'{e}r-Rao lower bound (CRLB) of the statistical inference problem defined by QSPE. The lower bound is given by the diagonal element of inverse Fisher information matrix
	\begin{equation}
	    \mathrm{CRLB}\left(\hat{\xi}_k\right) = \left(I^{-1}(\Xi)\right)_{kk}
	\end{equation}
	where the Fisher information matrix is element-wisely defined in \cref{eqn:Fisher-information-element-definition}. At the same time, we also compute the approximation to the optimal variance in the pre-asymptotic regime $d \ll 1/\theta$ derived in \cref{eqn:opt-theta-var,eqn:opt-varphi-var,eqn:opt-chi-var}. The numerical results are given in Figure 3 (b) in the main text. It can be seen that the approximated optimal variance agrees very well with the exact CRLB. In the asymptotic regime with large enough $d$, the optimal variance scaling given by the CRLB is as predicted in \cref{eqn:crlb-asymptotic}. Furthermore, the numerical results validate that there exists a nontrivial transition around $d \approx 1/\theta$ making the optimal variance scalings completely different in the pre-asymptotic and asymptotic regime.
	
	\begin{figure}[htbp]
	    \centering
	    \includegraphics[width=\columnwidth]{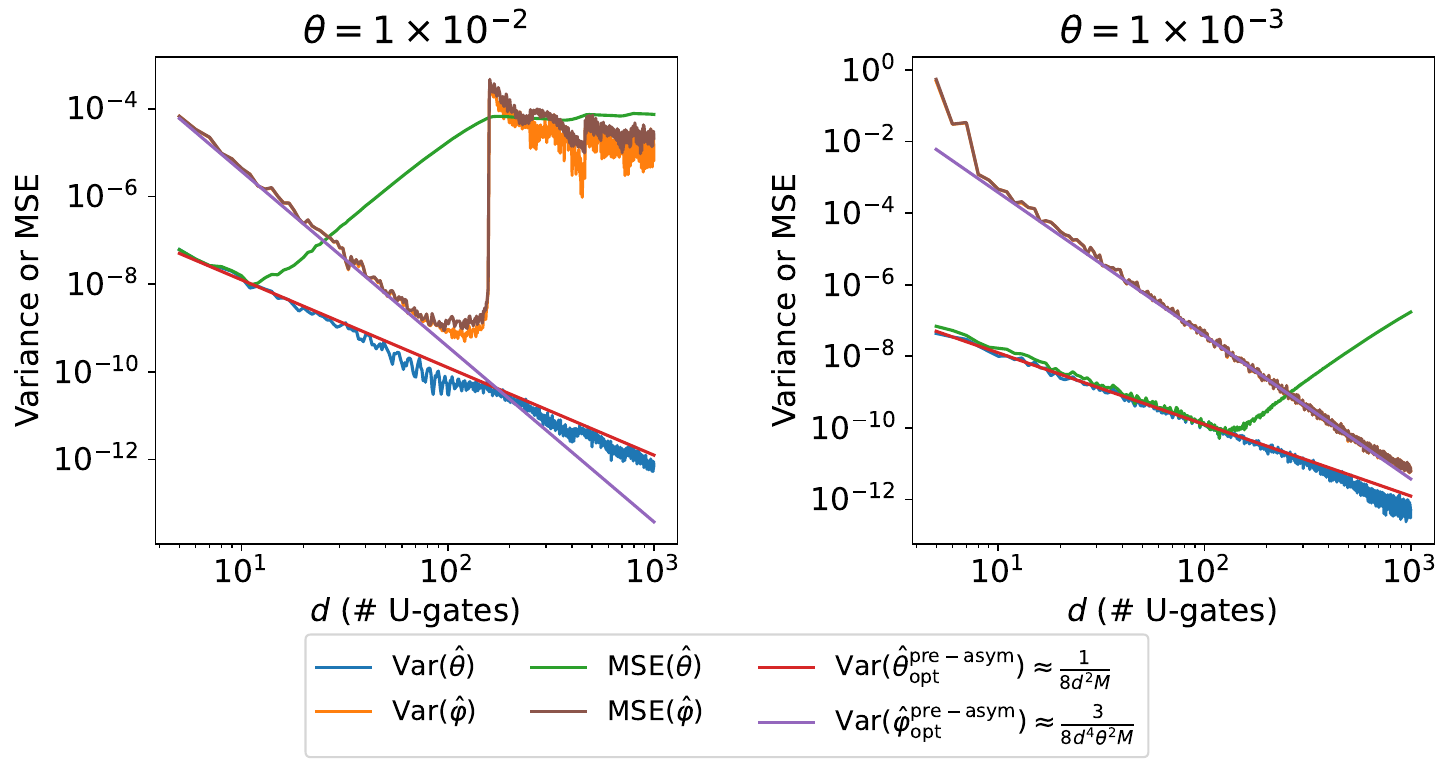}
	    \caption{Variance and mean-square error (MSE) of QSPE estimators. The left panel corresponds to the case where $\theta$ is relatively large and $d\theta \ll 1$ condition fails quickly at around $d=10$, beyond which bias dominates the estimator's  MSE since our inference model assumption~($d\theta \ll 1$) fails. The right panel corresponds to the case where $d\theta \ll 1$ condition holds all the way to around $d=100$.  The single-qubit phases are set to $\varphi = \pi/16$ and $\chi = 5\pi/32$. The number of measurement samples is set to $M = 1\times10^5$. Each data point is derived from $100$ independent repetitions.}
	    \label{fig:variance-qspcf}
	\end{figure}
	
	To prove the optimality of QSPE and investigate the situation where the conditions for deriving QSPE hold, we numerically estimate the variances of QSPE estimators and compare them with the derived optimal variances in the pre-asymptotic regime in \cref{eqn:opt-theta-var,eqn:opt-varphi-var}. The QSPE estimators are derived by approximating the original statistical inference problem by a linear model. When $d$ gets large, the model violation due to the approximation contributes to the bias of QSPE estimators. We compute the mean-square error (MSE) and using the bias-variance decomposition $\mathrm{MSE} = \mathrm{Var} + \mathrm{bias}^2$ to quantify the bias. The numerical results are displayed in \cref{fig:variance-qspcf}. Our simulation shows that the bias of $\theta$-estimator dominates the MSE and contaminates the inference accuracy after $d$ becomes larger than a threshold determined by the pre-asymptotic regime $d\theta \ll1$. Despite the bias due to the model violation, the MSE of the $\theta$-estimator still achieves some accuracy of order $\theta^2$ which suggests that the $\theta$-estimator might give a reasonable estimation of a similar order with model violation in larger $d$. The numerical results show that the $\varphi$-estimator is more robust where the MSE deviates significantly from the theoretical scaling in the pre-asymptotic regime after $d \ge 1/\theta$ is large enough to pass to the asymptotic regime. Furthermore, the MSE well matches the variance which implies that the bias in $\varphi$-estimator is always small. The difference in the robustness of the $\theta$- and $\varphi$-estimators is credited to the construction of QSPE in which the inferences of $\theta$ and $\varphi$ are completely decoupled due to the data post-processing using FFT. 
	
	\cref{fig:variance-qspcf} and Figure 3 (b) in the main text suggest the following. (1) In the pre-asymptotic regime, QSPE estimators achieve the optimality in the sense of saturating the Cram\'{e}r-Rao lower bound and exhibit   robustness against time-dependent errors in $\varphi$ in both simulation and experimental deployments. Furthermore, the construction of QSPE estimators only involves direct algebraic operations rather than iterative optimization, and the reduced inference problems in Fourier space are linear statistical models whose global optimum is unique for each realization. This not only  enables the fast and reliable data post-processing but also allows us to analyze its performance analytically. (2) Passing to the asymptotic regime, given the significant bias of $\theta$-estimator and the sharp transition of the variance of $\varphi$-estimator, one has to use other estimators to saturate the optimal variance scaling and unbiasedness, for example, maximum-likelihood estimators (MLE). Furthermore,  the analysis based on the Cram\'{e}r-Rao lower bound is made by fixing the data generation (measuring quantum circuits) but varying data post-processing.

    \subsection{Comparing with quantum Cram\'{e}r-Rao lower bound}\label{sec:quantum-CRLB}
    As a quantum analog of classical Fisher information, quantum Fisher information (QFI) lies in the center of quantum metrology by providing a fundamental lower bound on the accuracy one can infer from the system of a given resource limit. 

    In general case, the QFI is defined as
    \begin{equation}
        \mf{F}(\theta) := \Tr{\varrho(\theta) L_\varrho^2(\theta)}
    \end{equation}
    where the symmetric logarithmic derivative $L_\varrho(\theta)$ is defined implicitly \cite{BraunsteinCaves1994,KullGuerinVerstraete2020}. For a system whose density matrix evolves as $\varrho(\theta) = e^{\I \theta \mc{H}} \varrho_0 e^{-\I \theta \mc{H}}$, the QFI can be explicitly computed by diagonalizing $\varrho(\theta)$. According to the analysis in \cite{KullGuerinVerstraete2020}, the QFI is an upper bound on the Fisher information over all possible measurements. For brevity, we only consider the inference of $\theta$ and hold all other unknown parameters constant in the analysis. However, our analysis can be generalized to the multiple parameter inference by adopting the multi-variable QFI in Ref. \cite{KullGuerinVerstraete2020}.

    In our case, the collection of quantum circuits with variable modulation angle $\omega$ can be written in a density matrix which is a uniform average over all circuit realizations, namely
    \begin{equation}
        \varrho(\theta) = \frac{1}{2 d - 1} \sum_{j = 0}^{2d - 2} \mc{U}^\scp{d}(\omega_j; \theta, \varphi, \chi) \varrho_0 \mc{U}^\scp{d}(\omega_j; \theta, \varphi, \chi)^\dagger.
    \end{equation}

    Note that the optimization is intended to be performed over all potential initialization and measurement. Given that only $\theta$ is considered and other angle parameters are held constant, it suffices to consider a simpler case
    \begin{equation}
        \varrho(\theta) = \frac{1}{2 d - 1} \sum_{j = 0}^{2d - 2} U^\scp{d}(\omega_j - \varphi, \theta) \varrho_0 U^\scp{d}(\omega_j - \varphi, \theta)^\dagger.
    \end{equation}
    According to \cref{eqn:circuit-rep-of-qspc-with-building-block}, this alternative density matrix is equivalent to absorbing some constant rotation gates into initialization and measurement. When $\theta \ll 1$ is small, the following expansion holds:
    \begin{equation}\label{eqn:approximate-U-d-expansion}
        U^\scp{d}(\omega, \theta) = I + \I \theta \sum_{k = 1}^d e^{\I k \omega Z} X e^{\I (d - k + 1) \omega Z} + \Or((d\theta)^2) = I + \I \theta \underbrace{X \sum_{k = 1}^d e^{\I (d - 2 k + 1) \omega Z}}_{\mc{H}(\omega)} + \Or((d\theta)^2).
    \end{equation}
    Note that $\mc{H}(\omega)$ is Hermitian because
    \begin{equation}
        \mc{H}^\dagger(\omega) = \sum_{k = 1}^d e^{- \I (d - 2k + 1) \omega Z} X = X \sum_{k = 1}^d e^{\I (d - 2k + 1) \omega Z} = \mc{H}(\omega).
    \end{equation}
    Furthermore,
    \begin{equation}\label{eqn:approximate-U-d-H}
        \mc{H}(\omega) = X \diag\left\{\sum_{k = 1}^d e^{\I(d - 2k + 1) \omega}, \sum_{k = 1}^d e^{- \I(d - 2k + 1) \omega} \right\} = X \frac{e^{\I d \omega} - e^{-\I d \omega}}{e^{\I \omega} - e^{- \I \omega}} = X U_{d-1}(\cos(\omega))
    \end{equation}
    where $U_{d - 1}$ is the Chebyshev polynomial of the second kind.
    Then, when $d \theta \ll 1$ is small enough, it approximately holds that
    \begin{equation}
        \varrho_j(\theta) := U^\scp{d}(\omega_j - \varphi, \theta) \varrho_0 U^\scp{d}(\omega_j - \varphi, \theta)^\dagger = \varrho_0 + \I \theta \left[\mc{H}(\omega_j - \varphi), \varrho_0\right] + \Or((d\theta)^2) \approx e^{\I \theta \mc{H}(\omega_j - \varphi)} \varrho_0 e^{- \I \theta \mc{H}(\omega_j - \varphi)}.
    \end{equation}
    Then, as a two-dimensional density matrix, its QFI is 
    \begin{equation}
        \mf{F}_j(\theta) := 4 \frac{(\lambda_0 - \lambda_1)^2}{\lambda_0 + \lambda_1} \abs{\bra{\psi_0} \mc{H}(\omega_j - \varphi) \ket{\psi_1}}^2
    \end{equation}
    where $\lambda_i, \ket{\psi_i}$ are the eigenvalue and eigenvector of the density matrix. It is upper bounded as
    \begin{equation}\label{eqn:QFI-omega}
        \mf{F}_j(\theta) \le 4 \norm{\mc{H}(\omega_j - \varphi)}_2^2 = 4 U_{d - 1}^2(\cos(\omega_j - \varphi)).
    \end{equation}
    Because the overall density matrix is a uniform combination of $\varrho_j(\theta)$ and all generator Hamiltonians are scaled Pauli $X$ operators, the convexity of the QFI implies that
    \begin{equation}
        \mf{F}(\theta) \le \frac{1}{2 d - 1} \sum_{j = 0}^{2d - 2} \mf{F}_j(\theta) \le \frac{4}{2d - 1} \sum_{j = 0}^{2d - 2} U_{d - 1}^2(\cos(\omega_j - \varphi)) = 4 d
    \end{equation}
    where \cref{lem:integral-Usq} is used. Note that the upper bound on the QFI derived here is independent of the initialization $\varrho$, and the formalism of QFI provides a bound on the inference variance regardless of the choice of measurements. Hence, we achieve a lower bound on the variance of $\hat{\theta}$ with respect to any initialization, measurement, and classical data processing by invoking the quantum Cram\'{e}r-Rao bound:
    \begin{equation}\label{eqn:quantum-CRLB}
        \inf \mathrm{Var}(\hat{\theta}) \ge \frac{1}{M 2 (2d - 1) \mf{F}(\theta)} \ge \frac{1}{8 M d (2d - 1)} 
    \end{equation}
	where it uses the fact that there are $2 (2d-1) M$ experiments in total. 

    Before closing the analysis of inference error, we could gain more understanding of our method from the lower bounds derived in this section.
    \begin{enumerate}
        \item Compared with the result in \cref{eqn:opt-theta-var} derived by applying classical Cram\'{e}r-Rao bound, we see that the bound in \cref{eqn:quantum-CRLB} derived from QFI is lower, namely, $\text{quantum bound} = 0.5 \times \text{classical bound}$. This differentiation is explainable. Note that we use two logical Bell states to perform experiments. The advantage is the experimental probabilities of these two experiments form a conjugate pair to reconstruct a complex function that for the ease of analysis. This complex function and its properties (see \cref{thm:reconstruction-h-Fourier-expansion,thm:approx-coef-first-order}) eventually lead to a simple robust statistical estimator requiring only light computation. In contrast, the data generated from the initialization of one Bell state still contains full information of the parameters to be estimated. However, the highly nonlinear dependency renders the practical inference challenging. Hence, the factor of $2$ is due to the use of a pair of Bell states. Although the QFI indicates that inference variance can be lower by removing such redundancy in the initialization, the nature of ignoring practical ease makes it hard to achieve.
        \item The importance of the phase matching condition is also reflected in the analysis of QFI. It is worth noting that when we set $\omega_j = \varphi$, the QFI of a single experiment is $\mf{F}_j(\theta) | _{\omega_j = \varphi} \le 4 U_{d - 1}^2(1) = 4 d^2$ which attains the maximum of the Chebyshev polynomial. However, in practice, due to the absence of accurate information of $\varphi$, our method samples the data on equally spaced $\omega$ points and processes the data via FFT to isolate the dependencies of $\theta$ and $\varphi$. This further sampling procedure averages the QFI and lowers its value from $4d^2$ to $4d$. 
    \end{enumerate}

 \begin{lemma}\label{lem:integral-Usq}
 Given an integer $d$ and any $\varphi \in \RR$, it holds that
 \begin{equation}
     \frac{1}{2d - 1} \sum_{j = 0}^{2d-2} U_{d - 1}^2(\cos(\omega_j - \varphi)) = d
 \end{equation}
 where $\omega_j := j \pi / (2d - 1)$ where $j = 0, \cdots, 2d - 2$.
 \end{lemma}
 \begin{proof}
     Note that the discrete orthogonality implies that
     \begin{equation}
         \frac{1}{2d - 1} \sum_{j = 0}^{2d-2} U_{d - 1}^2(\cos(\omega_j)) = \frac{1}{\pi} \int_0^\pi U_{d - 1}^2(\cos(\omega - \varphi)) \rd \omega = \frac{1}{\pi} \left(\int_0^\pi + \int_{-\varphi}^0 - \int_{\pi - \varphi}^\pi\right) U_{d - 1}^2(\cos(\omega)) \rd \omega.
     \end{equation}
     Using the parity condition of Chebyshev polynomials, it holds that
     \begin{equation}
         \int_{\pi - \varphi}^\pi U_{d - 1}^2(\cos(\omega)) \rd \omega \stackrel{\omega \leftarrow \omega - \pi}{=} \int_{- \varphi}^0 \left((-1)^{d - 1} U_{d - 1}(\cos(\omega)) \right)^2 \rd \omega = \int_{- \varphi}^0 U_{d - 1}^2(\cos(\omega)) \rd \omega.
     \end{equation}
     Hence, for any $\varphi \in \RR$, it holds that
     \begin{equation}
         \frac{1}{2d - 1} \sum_{j = 0}^{2d-2} U_{d - 1}^2(\cos(\omega_j)) = \frac{1}{\pi} \int_0^\pi U_{d - 1}^2(\cos(\omega)) \rd \omega = \frac{1}{\pi} \int_{-1}^1 \frac{U_{d - 1}^2(x)}{\sqrt{1 - x^2}} \rd x.
     \end{equation}
     Using the product formula of Chebyshev polynomials, we have
     \begin{equation}
         U_{d - 1}^2(x) = \sum_{p = 0}^{d - 1} U_{2p}(x) = \sum_{p = 0}^{d - 1} \left(1 + 2 \sum_{k = 1}^p T_{2k}(x) \right).
     \end{equation}
     Then, the weighted orthogonality of Chebyshev polynomials of the first kind implies that
     \begin{equation}
         \frac{1}{\pi} \int_{-1}^1 \frac{U_{d - 1}^2(x)}{\sqrt{1 - x^2}} \rd x = \sum_{p = 0}^{d - 1} \left(1 + \frac{2}{\pi} \sum_{k = 1}^p \int_{-1}^1 \frac{T_{2k}(x) T_0(x)}{\sqrt{1 - x^2}} \rd x\right) = d
     \end{equation}
     which completes the proof.
 \end{proof}

	\section{Analysis of realistic error}\label{sec:realistic-error}
	Although QSPE estimators are derived from modeling Monte Carlo sampling error, we numerically show their robustness against realistic errors in this section. This section is organized as follows. We discuss the sources of realistic errors including depolarizing error, time-dependent error, and readout error in each subsection. We study the methods for correcting some realistic errors by analyzing experimental data. Furthermore, we perform numerical experiments to validate the robustness of our proposed quantum metrology scheme.
	
	\subsection{Depolarizing error}\label{sec:depolarizing}
	The quantum error largely contaminates the signal. In the two-qubit system, we assume the quantum error is captured by a depolarizing quantum channel, where the density matrix is transformed to the convex combination of the correctly implemented density matrix and that of the uniform distribution on bit-strings. Therefore, assuming the infinite number of measurement samples (vanishing Monte Carlo sampling error), the measurement probability is 
	\begin{equation}
	p_{X(Y)|\alpha}(\omega; \theta, \varphi, \chi) = \alpha p_{X(Y)}(\omega; \theta, \varphi, \chi) + \frac{1-\alpha}{4}
	\end{equation}
	where $\alpha \in [0, 1]$ is referred to as the circuit fidelity. Then, the sampled reconstructed function is also shifted and scaled accordingly $\mf{h}_\alpha(\omega; \theta, \varphi, \chi) = \alpha \mf{h}(\omega; \theta, \varphi, \chi) - \frac{1-\alpha}{4}(1+\I)$. Consequentially, the Fourier coefficients are expected to be scaled by $\alpha$ simultaneously and the constant shift only contributes to the zero-indexed Fourier coefficient, namely
	\begin{equation}
	\abs{c^\expl_{0|\alpha}} = \abs{\alpha c^\expl_0 - \frac{1-\alpha}{4}(1+\I)} \approx \alpha\theta + \frac{1-\alpha}{2\sqrt{2}},\quad \abs{c^\expl_{k|\alpha}} \approx \alpha \theta, \ \forall k = 1, \cdots, d-1.
	\end{equation}
	The approximation of $\abs{c^\expl_{0|\alpha}}$ holds when the circuit fidelity is not close to one, namely, $\theta \ll 1-\alpha$. Yet when the circuit fidelity is close to one, the depolarizing error can be neglected as a higher-order effect. Using this feature, the circuit fidelity can be estimated from the difference between the Fourier coefficient of zero index and those of nonzero indices. Then, the estimators of the circuit fidelity and the swap angle are given by
	\begin{equation}\label{eqn:estimate-alpha}
	\begin{split}
	& \hat{\alpha} = 1 - 2\sqrt{2}\left(\abs{c^\expl_{0|\alpha}} - \frac{1}{d-1} \sum_{k=1}^{d-1} \abs{c^\expl_{k|\alpha}}\right),\\
	& \hat{\theta} = \frac{1}{\hat{\alpha}} \times \frac{1}{d-1} \sum_{k=1}^{d-1} \abs{c^\expl_{k|\alpha}}.
	\end{split}
	\end{equation}
	We numerically test the accuracy of these estimators in \cref{sec:additional-numerical}. In \cref{tab:qubit-pair-error-rate}, we list the effective depolarizing error rate on the single-excitation subspace inferred from the exponential decay of circuit fidelities derived from QSPE methods on our experimental measurements. These values agree with our   estimation results using a conventional cross-entropy benchmark that is much slower to run and requires a $10X$ deeper circuit for randomization.

\begin{table}[htbp]
\centering
\begin{tabular}{@{} *{5}{c} @{}}\midrule
(3,6) and (3,7) & (3,6) and (4,6) & (3,7) and (4,7) & (4,5) and (4,6) & (4,7) and (5,7) \\
$4.52\times 10^{-3}$ & $4.73\times 10^{-3}$ & $5.39\times 10^{-3}$ & $4.69\times 10^{-3}$ & $5.15\times 10^{-3}$ \\\midrule
(5,7) and (6,7) & (5,7) and (5,8) & (5,6) and (6,6) & (5,6) and (5,7) & (4,8) and (5,8) \\
$8.25\times 10^{-3}$ & $5.89\times 10^{-3}$ & $2.81\times 10^{-3}$ & $3.59\times 10^{-3}$ & $4.96\times 10^{-3}$ \\\midrule
(5,8) and (5,9) & (5,8) and (6,8) & (6,6) and (7,6) & (6,8) and (7,8) & (7,5) and (7,6) \\
$5.84\times 10^{-3}$ & $5.70\times 10^{-3}$ & $3.36\times 10^{-3}$ & $5.13\times 10^{-3}$ & $3.32\times 10^{-3}$ \\\midrule
(7,6) and (7,7) & (7,7) and (7,8) &  &  &  \\
$2.36\times 10^{-3}$ & $2.89\times 10^{-3}$ &  &  &  \\\midrule
\end{tabular}
\caption{Qubit pairs and the inferred effective error rate on the single-excitation subspace. The error rate is estimated by the regression with respect to the exponential decay. The regression data are the circuit fidelity estimated from QSPE in Figure 4 in the main text (top-right panel).}
\label{tab:qubit-pair-error-rate}
\end{table}

\subsection{Time-dependent error}\label{coherentdriftSec}
The dominant time-dependent noise in superconducting qubits two-qubit control is in the frequency of the qubits. It can be modeled by time-dependent Z phase error in \fsim.  Observed from experimental data, the magnitude of the time-dependent drift error increases when more gates are applied to the circuit. To emulate the realistic time-dependent noise, we model the noise by introducing a random deviation in angle parameters, which is referred to as the coherent angle uncertainty. Given a perfect \fsim parametrized as $U_\fsim(\theta,\varphi,\chi,*)$, the erroneous quantum gate due to the coherent angle uncertainty is another \fsim parametrized as $U_\fsim(\theta_\mathrm{unc},\varphi_\mathrm{unc},\chi_\mathrm{unc},*)$. Here, angle parameters subjected to the uncertainty are distributed uniformly at random around the perfect value
	\begin{equation}\label{eqn:coherent-noise-1}
	\theta_\mathrm{unc} \in [\theta - D_\theta, \theta + D_\theta],\ \varphi_\mathrm{unc} \in [\varphi - D_\varphi, \varphi + D_\varphi],\ \chi_\mathrm{unc} \in [\chi - D_\chi, \chi + D_\chi]
	\end{equation}
	where $D_\theta, D_\varphi, D_\chi$ stand for the maximal deviations of uncertain parameters. Inspired by experimental results, maximum deviations of phase angles are increasing when more \fsim's are applied. Moreover, there is a Gaussian noise~\cite{niu2019universal} in the analog pulse realizations, causing small fluctuations on all gate parameters. To capture this feature and the rough estimate from the experimental data, we set the uncertainty model when the $j$-th \fsim is applied as
	\begin{equation}\label{eqn:coherent-noise-2}
	D_\theta^\scp{j} = 0.1 \times \theta,\ D_\varphi^\scp{j} = D_\chi^\scp{j} = 0.3 \times \frac{j}{d}.
	\end{equation}
	
Noticeably, the proposed model has already taken the phase drift in $Z$-rotation gates into account, which is effectively factored in the random phase drift in the single-qubit phase $\varphi$ and $\chi$ in the \fsim.
	
	\subsection{Numerical performance of the estimation against depolarizing error and time-dependent drift error}\label{sec:additional-numerical}
	In the numerical simulation, we add a depolarizing error channel after each individual gate. In terms of the quantum channel, it is quantified as
	\begin{equation}
	\begin{split}
	\mc{E}_{A_0}\left(\varrho\right) =& \left(1-\frac{3}{4}r\right) \varrho + \frac{r}{4} \left( \left(X_{A_0}\otimes I_{A_1}\right)\varrho\left(X_{A_0}\otimes I_{A_1}\right)\right.\\
	&\left.+ \left(Y_{A_0}\otimes I_{A_1}\right)\varrho\left(Y_{A_0}\otimes I_{A_1}\right) + \left(Z_{A_0}\otimes I_{A_1}\right)\varrho\left(Z_{A_0}\otimes I_{A_1}\right)\right),\\
	\mc{E}_{A_0,A_1}\left(\varrho\right) =& (1-r)\varrho + r \frac{I_{A_0,A_1}}{4}
	\end{split}
	\end{equation}
	where $r$ is the error rate. At the same time, the quantum circuit is subject to  drift error according to \cref{eqn:coherent-noise-1,eqn:coherent-noise-2}. 
	
	In \cref{fig:alpha_degree}, we numerically test the accuracy of estimating the circuit fidelity using the Fourier space data according to the estimator in \cref{eqn:estimate-alpha}. The reference value of the circuit fidelity is computed from the digital error model (DEM) \cite{BoixoIsakovSmelyanskiyEtAl2018} with 
	\begin{equation}
	    \alpha_\text{DEM} := (1-r)^{n_\text{gates}} \approx (1-r)^{2d+5} + \Or(r).
	\end{equation}
	Here, $n_\text{gates}$ stands for the number of total gates in the quantum circuit. Because of the additional phase gate used in the Bell-state preparation, the quantum circuit for computing $p_Y$ uses $n_\text{gates} = 2d+6$ gates while that for $p_X$ uses $n_\text{gates} = 2d+5$ gates. This ambiguity in a gate makes the left-hand side approximate the circuit fidelity up to $\Or(r)$. In \cref{fig:alpha_degree}, the performance of the circuit fidelity estimation is quantified by the deviation $\abs{\hat{\alpha} - \alpha_\text{DEM}}$. As the circuit depth of QSPE increases, it turns out that the deviation decreases to $\sim 0.001$, which is equal to the error rate $r$. The decreasing deviation is due to the improvement of the SNR when increasing the circuit depth. Furthermore, the plateau near $0.001$ is due to the ambiguity discussed in the reference $\alpha_\text{DEM}$. In the left panel, we turn off the time-dependent drift error and the quantum circuit is only subject to Monte Carlo sampling error and depolarizing error. However, the performance of the circuit fidelity estimation does not differ significantly after turning on the time-dependent drift error. The numerical results suggest that the depolarizing error can be inferred with considerable accuracy even in the presence of more complex time-dependent errors. 
	
	\begin{figure}[htbp]
		\centering
		\includegraphics[width=.75\columnwidth]{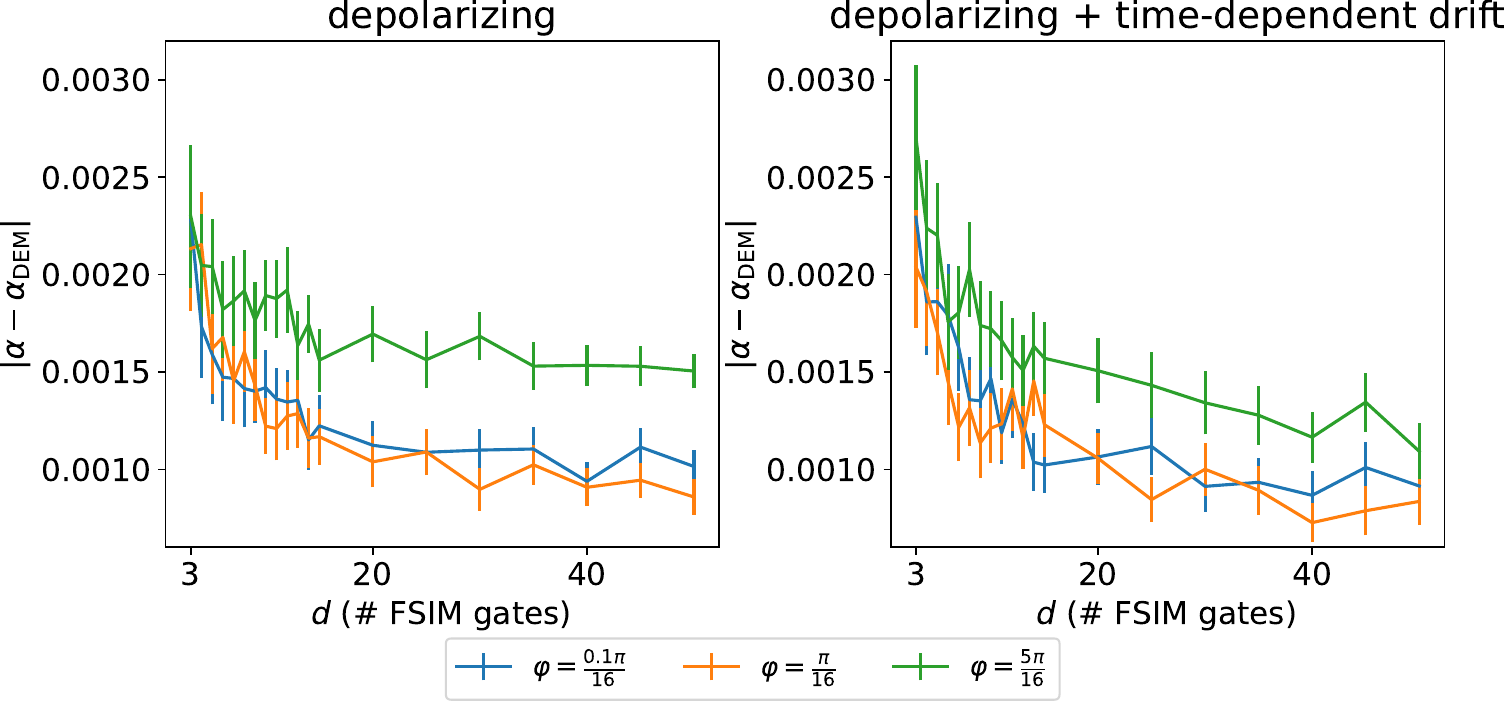}
		\caption{Estimating circuit fidelity using QSPE. The reference value $\alpha_\mathrm{DEM}$ is the circuit fidelity estimated from the digital error model. The sources of noise in the numerical experiments are Monte Carlo sampling error, depolarizing error and drift error. The depolarizing error rate is set to $r = 1 \times 10^{-3}$ and the number of measurement samples is set to $M = 1 \times 10^{5}$. The parameters of \fsim are set to $\theta = 1\times 10^{-3}$ and $\chi = 5\pi/32$. The error bar of each point stands for the confidence interval derived from $96$ independent repetitions.}
		\label{fig:alpha_degree}
	\end{figure}
	
	In \cref{fig:degree_cu,fig:meas_cu}, we test our proposed metrology scheme in the presence of Monte Carlo sampling error, depolarizing error, and time-dependent error. Although the system is subjected to realistic errors, the numerical results suggest that the QSPE estimators show some robustness against errors, and they can give reasonable estimation results with one or two correct digits. Furthermore, the accuracy of $\varphi$-estimation is also not fully contaminated by the time-dependent error on it. The improvement due to the peak fitting becomes less significant under realistic errors because the structure of the highest peak is heavily distorted in the presence of realistic errors. More interestingly, the numerical results show the accuracy of $\theta$-estimation does not decay and even increases after some $d^*$. This transition is due to a tradeoff. When $d$ becomes larger, the inference is expected to be more accurate because the gate parameters are more amplified. However, in the presence of realistic error, the \fsim is subjected to both time-independent errors and time-dependent drift errors. A quantum circuit with more \fsim s violates the model derived from the noiseless setting more. The competition between these two opposite effects makes the estimation error attain some minimum at $d^*$. This observation also suggests that in the experimental deployment, one can consider using a moderate $d$ with respect to the tradeoff.
	
	In \cref{fig:meas_cu}, we perform the numerical simulation with variable swap angle and number of measurement samples. Similar to the case of Monte Carlo sampling error, the estimation results are less accurate when $\theta$ is small because of the insufficient SNR. The numerical results indicate that the estimation accuracy cannot be further improved after the number of measurement samples is greater than some $M^*$. That is because increasing $M$ can only mitigate Monte Carlo sampling error. When $M$ is large enough, the sources of errors are dominated by depolarizing error and time-dependent drift error, which cannot be sufficiently mitigated by large $M$. Combing with the discussion on $d^*$, the numerical results suggest that the experimental deployment does not require large $d$ and $M$, and using a moderate choice of $d^*$ and $M^*$ suffices to get some accurate estimation.
	
	\begin{figure}[htbp]
		\centering
		\includegraphics[width=\columnwidth]{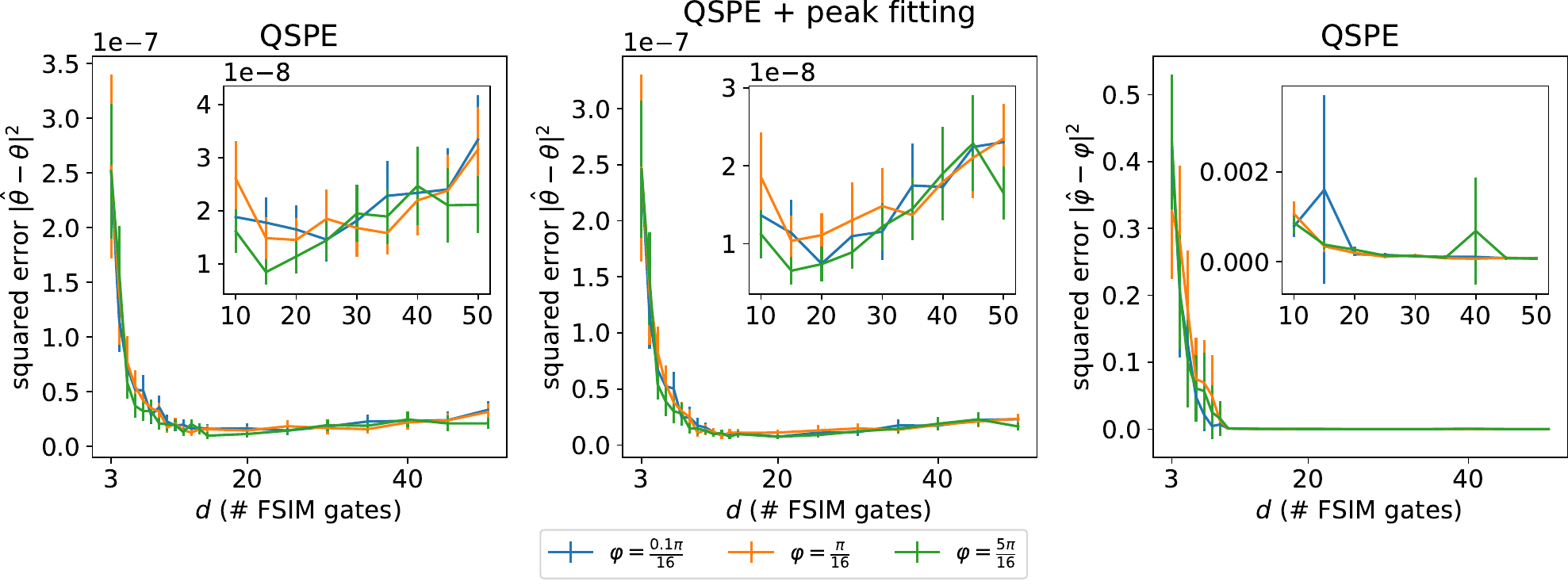}
		\caption{Accuracy of estimators as a function of the number of \fsim s. The sources of noise in the numerical experiments are Monte Carlo sampling error, depolarizing error and time-dependent drift error. The depolarizing error rate is set to $r = 1 \times 10^{-3}$ and the number of measurement samples is set to $M = 1 \times 10^{5}$. The swap angle is set to $\theta = 1 \times 10^{-3}$ and the phase parameter is set to $\chi = 5\pi/32$. The error bar of each point stands for the confidence interval derived from $96$ independent repetitions.}
		\label{fig:degree_cu}
	\end{figure}
	
	\begin{figure}[htbp]
		\centering
		\includegraphics[width=\columnwidth]{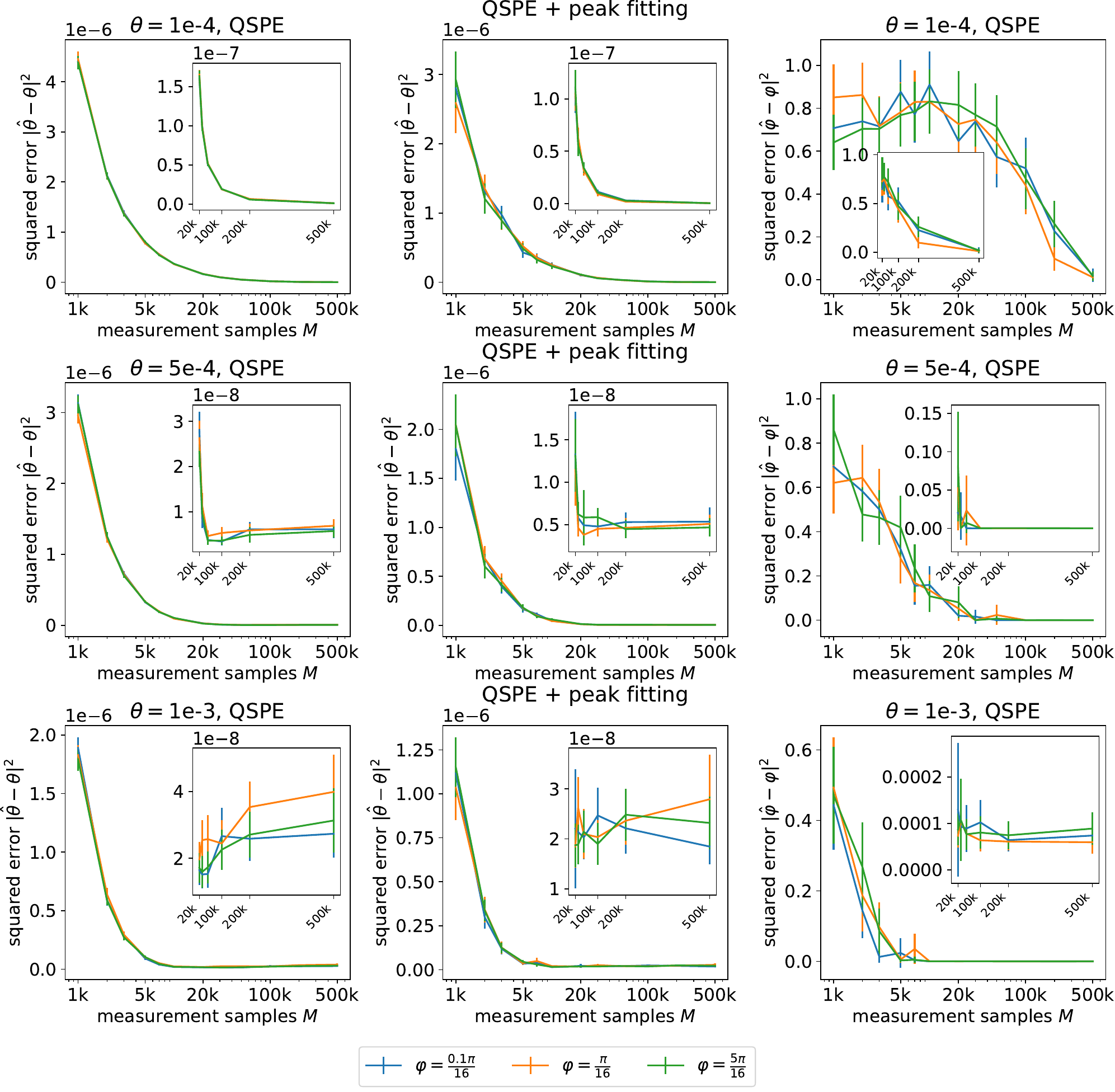}
		\caption{Accuracy of estimators as a function of the number of measurement samples. The sources of noise in the numerical experiments are Monte Carlo sampling error, depolarizing error and time-dependent drift error. The depolarizing error rate is set to $r = 1 \times 10^{-3}$. The circuit degree is set to $d = 50$ and the \fsim phase parameter is set to $\chi = 5\pi/32$.The error bar of each point stands for the confidence interval derived from $96$ independent repetitions.}
		\label{fig:meas_cu}
	\end{figure}
	
	\subsection{Readout error}\label{sec:readout}
	The readout error is modeled by a stochastic matrix whose entry is interpreted as a conditional probability. This matrix is referred to as the confusion matrix in the readout. For a two-qubit system, it takes the form
	\begin{equation}
	R := [\bPP(\mathrm{binary}(j) | \mathrm{binary}(i))]_{i, j = 0}^3 = \left( \begin{array}{*4{c}}
	\bPP(00 | 00) & \bPP(01 | 00) & \bPP(10 | 00) & \bPP(11 | 00)\\
	\bPP(00 | 01) & \bPP(01 | 01) & \bPP(10 | 01) & \bPP(11 | 01)\\
	\bPP(00 | 10) & \bPP(01 | 10) & \bPP(10 | 10) & \bPP(11 | 10)\\
	\bPP(00 | 11) & \bPP(01 | 11) & \bPP(10 | 11) & \bPP(11 | 11)
	\end{array} \right)
	\end{equation}
	where $\bPP(\mathrm{binary}(j) | \mathrm{binary}(j))$ is the conditional probability of measuring the qubits with the bit-string $\mathrm{binary}(j)$ given that the quantum state is $\ket{\mathrm{binary}(i)}$. The sum of each row of the confusion matrix is equal to one due to the normalization of probability. The confusion matrix can be determined by performing additional quantum experiments in which $I \otimes I$, $I \otimes X$, $X \otimes I$, and $X \otimes X$ are measured to determine each row, respectively. If the probability vector from the measurement with readout error is $\bvec{q}^\expl = (q^\expl(00), q^\expl(01), q^\expl(10), q^\expl(11))^\top$, the probability vector after correcting the readout error is given by inverting the confusion matrix
	\begin{equation}
	\bvec{p}^\expl = (p^\expl(00), p^\expl(01), p^\expl(10), p^\expl(11))^\top = \left( R^\top \right)^{-1} \bvec{q}^\expl.
	\end{equation}
	
	In practice, the confusion matrix is determined by finite measurement samples which could introduce error to the confusion matrix due to the statistical fluctuation. We analyze the error and its effect in \cref{thm:confmat}. As a consequence, the theorem indicates a minimal requirement on the measurement sample size so that the readout error can be accurately corrected.
	\begin{theorem}\label{thm:confmat}
		Let $\bvec{p}^\expl_\mathrm{fs}$ be the probability vector computed by inverting the confusion matrix estimated by finite samples. To achieve the bounded error $\norm{\bvec{p}^\expl - \bvec{p}^\expl_\mathrm{fs}}_2 \le \epsilon$ with confidence level $1 - \alpha$, it suffices to set the number of measurement samples in each experiment determining the confusion matrix as
		\begin{equation}
		M_\mathrm{cmt} = \ceil{\frac{2\kappa^2(\kappa+\epsilon)^2 \ln\left(32/\alpha\right)}{\epsilon^2}}
		\end{equation}
		where
		\begin{equation}
		\kappa = \max_{i = 0, \cdots, 3} \frac{1}{2R_{ii} - 1}.
		\end{equation}
	\end{theorem}
	\begin{proof}
		In each experiment given the exact outcome $u \in \{0,1\}^2$ without readout error and exact measurement probability vector $\vp^{(u)} := \left(p(00|u), p(01|u), p(10|u), p(11|u)\right)$ taking readout error into account, the number of measurement outcomes corresponding to each bit-string is multinomial distributed
		\begin{equation}
		\vk^{(u)} := \left(k(00|u), k(01|u), k(10|u), k(11|u)\right) \sim \mathrm{Multinomial}(M_\mathrm{cmt}, \vp^{(u)})
		\end{equation}
		where $k(s|u) := \#(\text{outcome is } s \text{ in }M_\mathrm{cmt} \text{ samples})$. The bit-string frequency
		\begin{equation}
		\vq^{(u)} = \left(q(00|u), q(01|u), q(10|u), q(11|u)\right) := \left(\frac{k(00|u)}{M_\mathrm{cmt}}, \frac{k(01|u)}{M_\mathrm{cmt}}, \frac{k(10|u)}{M_\mathrm{cmt}}, \frac{k(11|u)}{M_\mathrm{cmt}}\right)
		\end{equation}
		is therefore an estimate to the measurement probability since $\expt{\vq^{(u)}} = \vp^{(u)}$. However, the statistical fluctuation makes the estimate deviate the exact probability. Applying Hoeffding's inequality, we have
		\begin{equation}
		\bP\left(\abs{q(s|u) - p(s|u)} > \frac{\wt{\epsilon}}{4} \right) = \bP\left(\abs{k(s|u) - M_\mathrm{cmt} p(s|u)} > \frac{\wt{\epsilon} M_\mathrm{cmt}}{4}\right) \le 2 e^{- \frac{\wt{\epsilon}^2 M_\mathrm{cmt}}{8}}.
		\end{equation}
		Let the confusion matrix determined by finite samples be $R_\mathrm{fs}$ where $\left(R_\mathrm{fs}\right)_{ij} = q\left(\mathrm{binary}(j) | \mathrm{binary}(i)\right)$ and the subscript ``fs'' abbreviates ``finite sample''. Then, the deviation can be bounded as
		\begin{equation}
		\begin{split}
		&\bP\left(\norm{R_\mathrm{fs} - R}_2 > \wt{\epsilon}\right) \le \bP\left(\norm{R_\mathrm{fs} - R}_F > \wt{\epsilon}\right) = \bP\left(\sum_{s, u \in \{0,1\}^2} \abs{q(s|u) - p(s|u)}^2 > \wt{\epsilon}^2\right)\\
		&\le \bP\left(\bigcup_{s, u \in \{0,1\}^2} \left\{ \abs{q(s|u) - p(s|u)} > \frac{\wt{\epsilon}}{4} \right\}\right) \le \sum_{s, u \in \{0,1\}^2} \bP\left(\abs{q(s|u) - p(s|u)} > \frac{\wt{\epsilon}}{4} \right)\\
		&\le 32 e^{- \frac{\wt{\epsilon}^2 M_\mathrm{cmt}}{8}}.
		\end{split}
		\end{equation}
		Therefore, to achieve $\norm{R_\mathrm{fs} - R}_2 \le \wt{\epsilon}$ with confidence level $1 - \alpha$, it suffices to set the number of measurement samples in each experiment as
		\begin{equation}
		M_\mathrm{cmt} = \ceil{\frac{8 \ln\left(32/\alpha\right)}{\wt{\epsilon}^2}}.
		\end{equation}
		Expanding the matrix inverse in terms of power series and denoting $\Delta_\mathrm{fs} := R_\mathrm{fs} - R$ for convenience, we have
		\begin{equation}
		R_\mathrm{fs}^{-1} = \left(R + \Delta_\mathrm{fs}\right)^{-1} = R^{-1} \left(I + \Delta_\mathrm{fs} R^{-1}\right)^{-1} = R^{-1} + \sum_{j=1}^\infty R^{-1} \left(\Delta_\mathrm{fs} R^{-1}\right)^j.
		\end{equation}
		Furthermore, we get
		\begin{equation}
		\norm{R_\mathrm{fs}^{-1} - R^{-1}}_2 \le \norm{R^{-1}}_2 \sum_{j=1}^\infty \norm{\Delta_\mathrm{fs} R^{-1}}_2^j \le \frac{\norm{\Delta_\mathrm{fs}}_2 \norm{R^{-1}}_2^2}{1 - \norm{\Delta_\mathrm{fs}}_2 \norm{R^{-1}}_2}.
		\end{equation}
		Note that $\norm{R^{-1}}_2 = \lambda_\mathrm{min}^{-1}(R)$. To proceed, we have to lower bound the smallest eigenvalue of the confusion matrix. Note that all eigenvalues of the confusion matrix are real as a property of the stochastic matrix. Applying the Gershgorin circle theorem, all eigenvalues of the confusion matrix are contained in the union of intervals
		\begin{equation}
		\bigcup_{i = 0}^3 \left[R_{ii} - \sum_{j \ne i} R_{ij}, R_{ii} + \sum_{j \ne i} R_{ij}\right].
		\end{equation}
		Consequentially, the smallest eigenvalue of the confusion matrix is lower bounded
		\begin{equation}
		\lambda_\mathrm{min}(R) \ge \min_{i=0, \cdots, 3} \left(R_{ii} - \sum_{j \ne i} R_{ij}\right) = \min_{i = 0, \cdots, 3} \left(2 R_{ii} - 1\right) =: \kappa^{-1}.
		\end{equation}
		Thus, by properly choosing the number of measurement samples, with confidence level $1-\alpha$, we can bound the inverse confusion matrix as
		\begin{equation}
		\norm{R_\mathrm{fs}^{-1} - R^{-1}}_2 \le \frac{\wt{\epsilon} \kappa^2}{1 - \wt{\epsilon} \kappa}.
		\end{equation}
		When computing the probability vector by inverting the confusion matrix determined by finite measurement samples, the error is bounded as
		\begin{equation}
		\norm{\bvec{p}^\expl - \bvec{p}^\expl_\mathrm{fs}}_2 \le \norm{R_\mathrm{fs}^{-1} - R^{-1}}_2 \norm{\bvec{q}^\expl}_2 \le \norm{R_\mathrm{fs}^{-1} - R^{-1}}_2 \norm{\bvec{q}^\expl}_1 \le \frac{\wt{\epsilon} \kappa^2}{1 - \wt{\epsilon} \kappa}.
		\end{equation}
		Let 
		\begin{equation}
		\frac{\wt{\epsilon} \kappa^2}{1 - \wt{\epsilon} \kappa} = \epsilon \Rightarrow \wt{\epsilon} = \frac{\epsilon}{\kappa(\kappa+\epsilon)}
		\end{equation}
		Thus, to achieve the bounded error $\norm{\bvec{p}^\expl - \bvec{p}^\expl_\mathrm{fs}}_2 \le \epsilon$ with confidence level $1 - \alpha$, it suffices to set the number of measurement samples in each experiment determining the confusion matrix as
		\begin{equation}
		M_\mathrm{cmt} = \ceil{\frac{8\kappa^2(\kappa+\epsilon)^2 \ln\left(32/\alpha\right)}{\epsilon^2}}.
		\end{equation}
		The proof is completed.
	\end{proof}

    \subsection{\REV{Initial state preparation error}}\label{sec:initial_state_error}

    In this subsection, we analyze the effect of the error in the initialization of Bell states in the quantum circuits. One case is that the error brings the initial state out of the subspace $\spans\{\ket{0_\ell}, \ket{1_\ell}\}$. Because the rest of the circuit preserves that two-dimensional subspace, such error can be mitigated by post-selecting the measurement outcomes in the two-dimensional subspace. Hence, in the rest of this subsection, we focus on the initial state error within the two-dimensional subspace. 
    
    The initial state error is modeled by a unitary rotation $E_\eta := e^{- \I \eta K}$ where $\eta \in \mathbb{R}$ is the error magnitude and $K$ is the normalized Hermitian generator of the error with $\norm{K}_2 = 1$. Suppose the ideal Bell state is $\ket{\beta}$ with $\beta \in \{+, \I\}$, the experimentally prepared initial state is $\ket{\beta^\prime} := E_\eta \ket{\beta} \approx \ket{\beta} - \I \eta K \ket{\beta} + \Or(\eta^2)$. The experimental measurement probability is then 
    \begin{equation}
        \begin{split}
            p_\beta^\prime &= \abs{\bra{0_\ell} \mc{U}^{(d)}(\omega; \theta, \varphi, \chi) \ket{\beta^\prime}}^2 \approx \abs{\bra{0_\ell} \mc{U}^{(d)}(\omega; \theta, \varphi, \chi) \ket{\beta} - \I \eta \bra{0_\ell} \mc{U}^{(d)}(\omega; \theta, \varphi, \chi) K \ket{\beta}}^2 + \Or(\eta^2)\\
            &\approx p_\beta\left(1 + \Or(p_\beta^{-1/2} \eta \bra{0_\ell} \mc{U}^{(d)}(\omega; \theta, \varphi, \chi) K \ket{\beta})\right) + \Or(\eta^2)\\
            &\approx p_\beta + \Or(\eta \bra{0_\ell} \mc{U}^{(d)}(\omega; \theta, \varphi, \chi) K \ket{\beta}) + \Or(\eta^2 + d \theta \eta)
        \end{split}
    \end{equation}
    Here, the second and the third lines use the conclusion that $\bra{0_\ell} \mc{U}^{(d)}(\omega; \theta, \varphi, \chi) \ket{\beta} = \Or(\sqrt{p_\beta}) = \Or(1) + \Or(d \theta)$ when $d \theta \ll 1$. Note that \cref{eqn:approximate-U-d-expansion,eqn:approximate-U-d-H} indicate that
    \begin{equation}
        \Or(\mc{U}^{(d)}(\omega; \theta, \varphi, \chi)) = \Or(I + \I \theta X U_{d - 1}(\cos(\omega - \varphi))) + \Or((d\theta)^2). 
    \end{equation}
    Note that the identity matrix in the expansion only contributes to the constant term, which we explicitly annotate in the following result:
    \begin{equation}
        p_\beta^\prime \approx p_\beta + \text{const} \times \eta + \text{const} \times \eta \theta U_{d-1}(\omega - \varphi) + \Or(\eta^2 + d \theta \eta). 
    \end{equation}
    We remark that the constant-shift term $\text{const} \times \eta$ only affects the zeroth Fourier mode, which can be dropped by post-selecting Fourier modes of positive indices. This perturbative analysis implies that when subjected to initial state error with magnitude $\eta$, each Fourier component has an additive error magnitude $\Or(\eta \theta)$. Thus, the initial state error contributes to a relative error in the estimation, $|\hat{\theta} - \theta| / \theta = \Or(\eta)$. 

    The perturbative analysis above gives intuition in understanding the initial state error. However, the derivation relies on several approximations that may not cover all cases. In the following text, we perform numerical analysis to have a more comprehensive understanding of the estimation error induced by initial state error. We turn off all other noises including sampling error to focus on the effect of initial state error. We perform three different strategies to add initial state error to the system.
    \begin{enumerate}
        \item ``identical'': The initial state error $E_\eta = \exp(-\I \eta (X + Y) / \sqrt{2})$ is identical in all experiments for any $\beta \in \{+, \I\}$ and any $\omega$.
        \item ``fixed'': The initial state error $E_\eta$ only depends on the initial state to be prepared. It is fixed in all experiments for any $\omega$ but takes two possible values $E_{\eta, +} = \exp(-\I \eta (X + Y) / \sqrt{2})$ and $E_{\eta, \I} = \exp(-\I \eta (2X + Y) / \sqrt{5})$.
        \item ``random'': The initial state error is drawn randomly in each experiment for each $\beta \in \{+, \I\}$ and $\omega$, namely, $E_\eta = \exp(-\I \eta K)$ and $K$ is a random normalized Hermitian error generator.
    \end{enumerate}
    The numerical results are given in \cref{fig:spam_error_analysis}. These results suggest that the error scaling depends weakly on the strategy of adding initial state error. The study in \cref{fig:spam_error} suggests that the estimation error is linearly additive in the initial state error $\eta$, which agrees with our perturbative analysis. Moreover, we see that when $\theta = 10^{-3}$ and $\eta = 10^{-3}$, the induced estimation error of $\theta$ is around $10^{-6}$ and that of $\varphi$ is around $10^{-3}$ which are negligible compared to their own magnitudes. It is also worth noting that the plateau around $10^{-7}$ of $\theta$ estimation error (left panel in \cref{fig:spam_error}) is because the approximation error derived in \cref{sec:analytical-result-qspc} dominates the estimation error. In \cref{fig:spam_error_vary_theta}, we study the error scaling in a wide range of $\theta$ values. We see that the function form is nontrivial across different $\theta$ values. When $\theta$ is around $10^{-3}$, the $\theta$ estimation error scales quadratically in $\theta$, which might be due to some complex cancellation in Fourier transformation. The $\varphi$ estimation error scales reciprocally in $\theta$, because the signal-to-noise ratio of the complex phase increases as $\theta$ gets large. It agrees with the analysis in \cref{sec:analytical-result-qspc}. Moreover, \cref{fig:spam_error_vary_deg} shows that the spam error can be mitigated by increasing the circuit depth as the estimations error decays as $d^{-1}$ and $d^{-2}$ for $\theta$ and $\varphi$ respectively. Theoretically understanding these error scalings will be part of our future work.

\begin{figure}[htbp]
    \centering
    \subfigure[Estimation error as a function of initial state error $\eta$ ($\theta = 10^{-3}, d = 20$)\label{fig:spam_error}]{
        \includegraphics[width=.475\columnwidth]{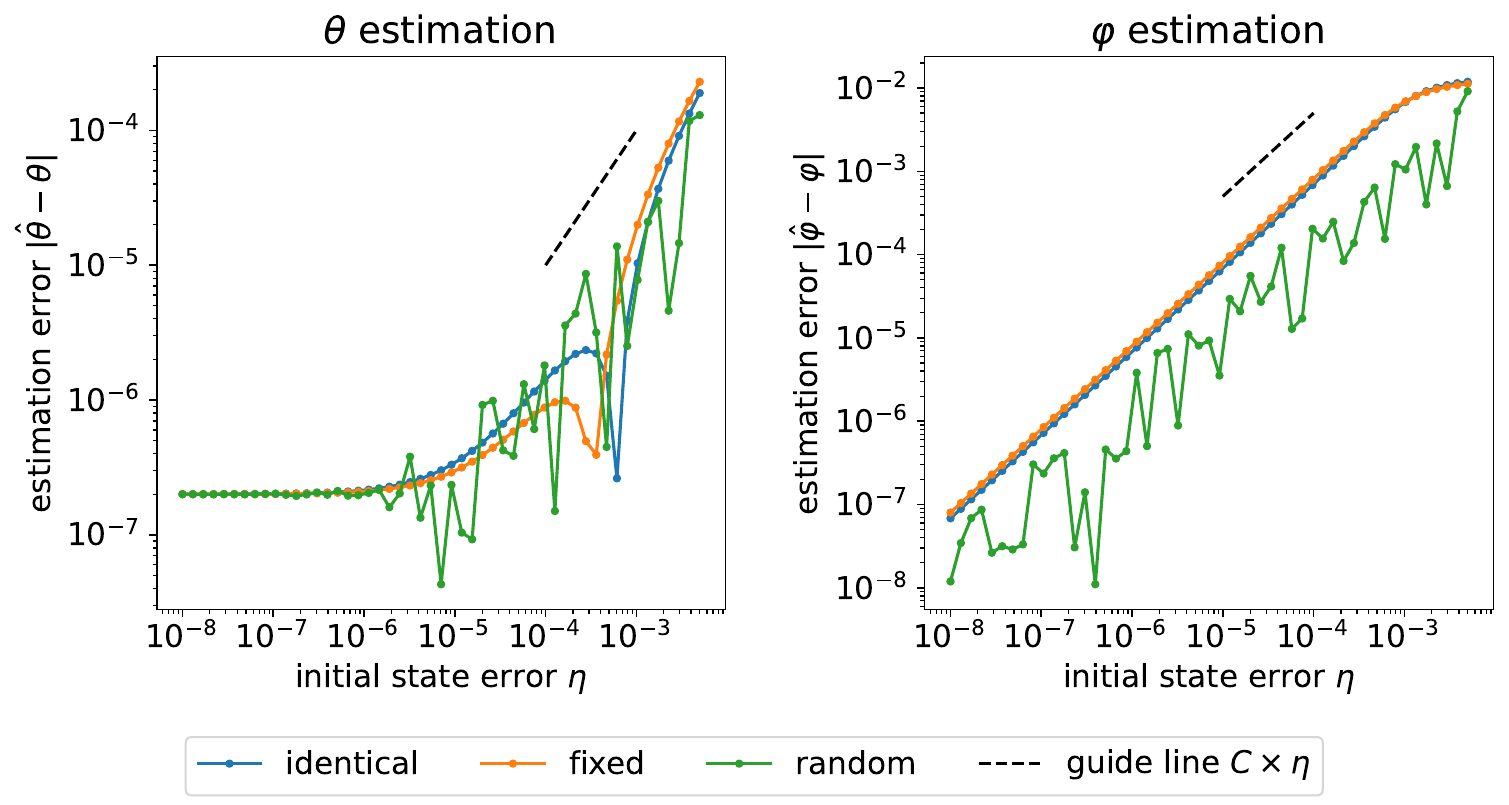}
    }
    \subfigure[Estimation error as a function of swap angle $\theta$  ($\eta = 10^{-6}, d = 20$) \label{fig:spam_error_vary_theta}]{
        \includegraphics[width=.49\columnwidth]{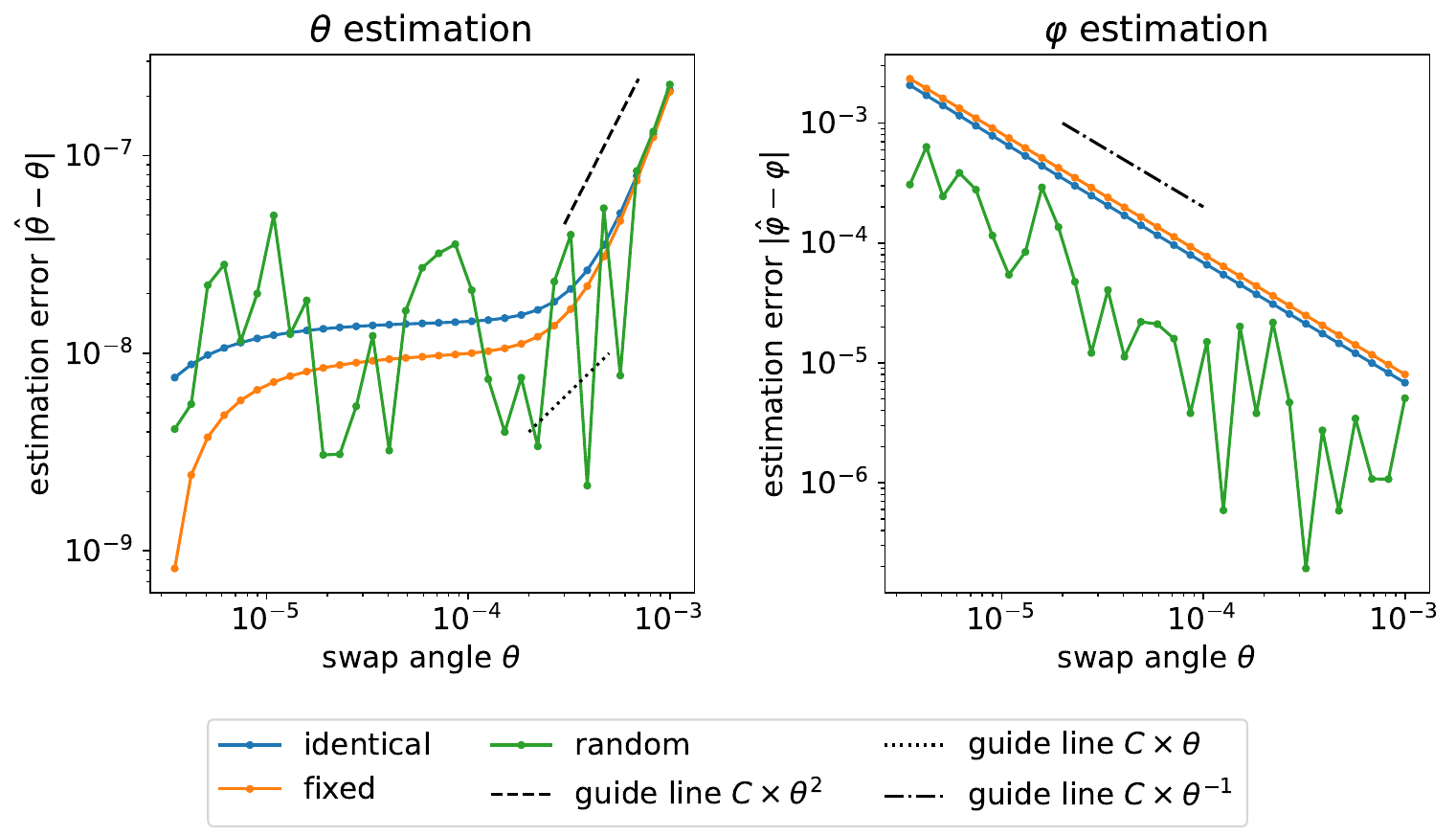}
    }
    \subfigure[Estimation error as a function of depth parameter $d$ ($\theta = 10^{-5}, \eta = 10^{-6}$) \label{fig:spam_error_vary_deg}]{
        \includegraphics[width=.5\columnwidth]{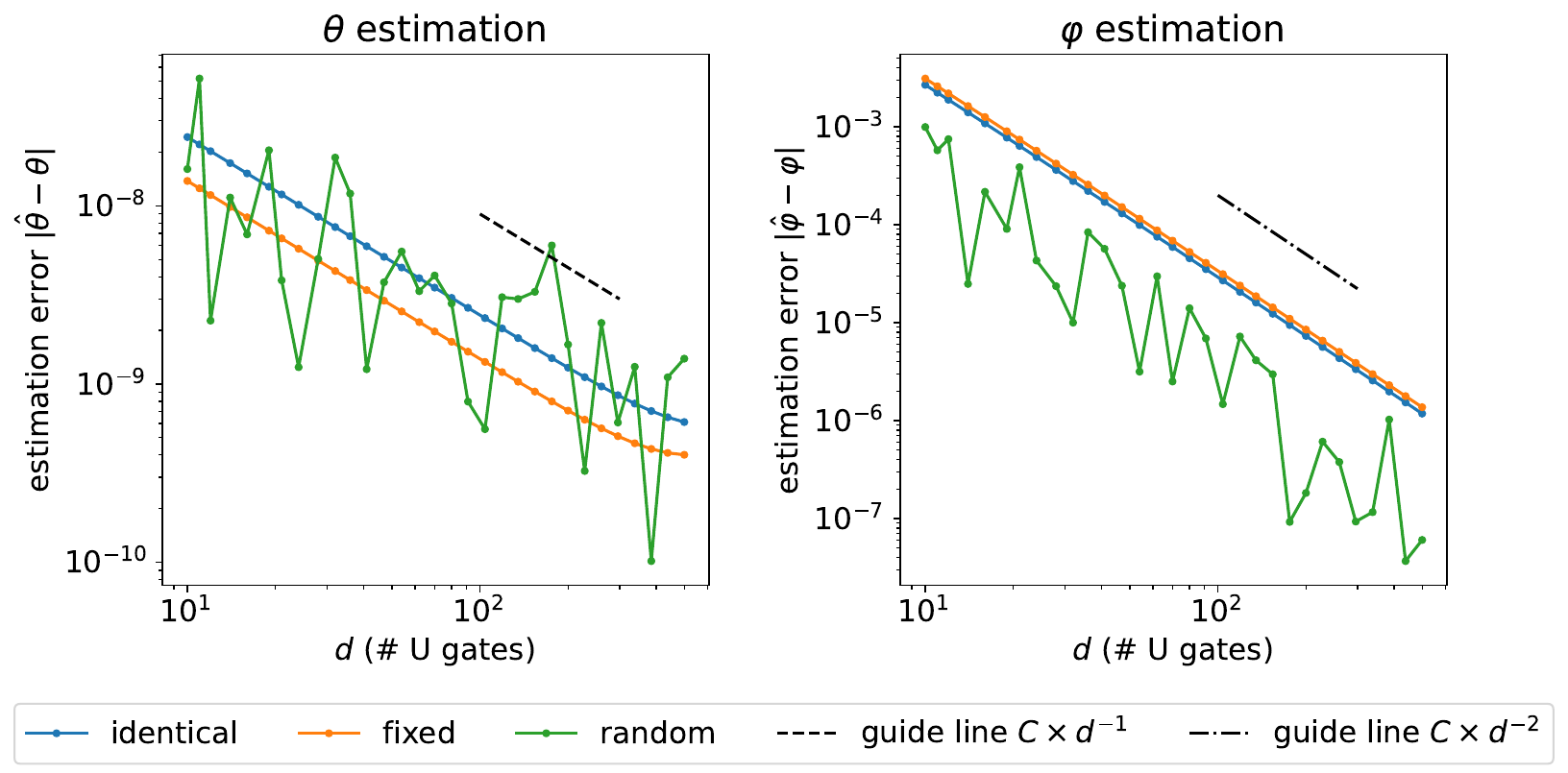}
    }
    \caption{Estimation error induced by initial state error. We set $\varphi = 3 \pi  / 16$ and $\chi = 5 \pi / 32$ for all experiments. \label{fig:spam_error_analysis}}
\end{figure}

	\section{Computing the polynomial representation on a special set of points}
\begin{lemma}\label{lma:poly-rep-special-pts}
	Let $d = 2^j$ for some $j = 0, 1, 2, \cdots$. Then 
	\begin{equation}
	P_\omega^\scp{d}(x) = e^{\I\omega} \left(\cos\left(d \sigma\right) + \I \frac{\sin\left(d \sigma\right)}{\sin\sigma} \left(\sin \omega\right) x\right) \text{ and } Q_\omega^\scp{d}(x) = \frac{\sin\left(d \sigma\right)}{\sin\sigma}
	\end{equation}
	where $\sigma = \arccos\left(\left(\cos \omega\right) x \right)$.
\end{lemma}
\begin{proof}
	A system of recurrence relations can be established by inserting the resolution of identity in the matrix multiplication:
	\begin{equation}\label{eqn:recurrence-Q}
	\begin{split}
	\I \sqrt{1-x^2} Q^\scp{d}(x) &= \bra{0}U^\scp{d}(\omega,\theta)\ket{1} = \bra{0}U^\scp{d/2}(\omega,\theta) e^{-\I\omega Z} \left(\ket{0}\bra{0} + \ket{1}\bra{1}\right) U^\scp{d/2}(\omega,\theta) \ket{1}\\
	&= e^{-\I\omega} P_\omega^\scp{d/2}(x) \I \sqrt{1-x^2} Q_\omega^\scp{d/2}(x) + e^{\I\omega} \I\sqrt{1-x^2} Q_\omega^\scp{d/2}(x) P_\omega^{\scp{d/2} *}(x)\\
	&= \I \sqrt{1-x^2} Q_\omega^\scp{d/2}(x) 2 \Re\left(e^{-\I\omega} P_\omega^\scp{d/2}(x)\right)\\
	\Rightarrow Q_\omega^\scp{d}(x) &= 2 Q_\omega^\scp{d/2}(x) \Re\left(e^{-\I\omega} P_\omega^\scp{d/2}(x)\right),
	\end{split}
	\end{equation}
	and
	\begin{equation}\label{eqn:recurrence-P}
	\begin{split}
	P_\omega^\scp{d}(x) &= \bra{0}U^\scp{d}(\omega,\theta)\ket{0} = \bra{0}U^\scp{d/2}(\omega,\theta) e^{-\I\omega Z} \left(\ket{0}\bra{0} + \ket{1}\bra{1}\right) U^\scp{d/2}(\omega,\theta) \ket{0}\\
	&= e^{-\I\omega} \left(P_\omega^\scp{d/2}(x)\right)^2 - e^{\I\omega} (1-x^2) \left(Q_\omega^\scp{d/2}(x)\right)^2\\
	&\stackrel{(\star)}{=} -e^{\I\omega} + 2 P_\omega^\scp{d/2}(x) \Re\left(e^{-\I\omega} P_\omega^\scp{d/2}(x) \right)\\
	\Rightarrow &\Re\left(e^{-\I\omega} P_\omega^\scp{d}(x) \right) = -1 + 2 \Re^2\left(e^{-\I\omega} P_\omega^\scp{d/2}(x) \right),\\
	\text{and } & \Im\left(e^{-\I\omega} P_\omega^\scp{d}(x) \right) = 2 \Im\left(e^{-\I\omega} P_\omega^\scp{d/2}(x) \right) \Re\left(e^{-\I\omega} P_\omega^\scp{d/2}(x) \right).
	\end{split}
	\end{equation}
	Here, equation $(\star)$ uses the special unitarity of $U^\scp{d/2}(\omega,\theta)$ which yields $P_\omega^\scp{d/2}(x) P_\omega^{\scp{d/2} *}(x) + (1-x^2) \left(Q_\omega^\scp{d/2}(x)\right)^2 = 1$ by taking determinant. We will first solve the nonlinear recurrence relation for $\Re\left(e^{-\I\omega} P_\omega^\scp{d}\right)$ in \cref{eqn:recurrence-P}. Note that the second-order Chebyshev polynomial of the first kind is $T_2(x) = 2x^2 - 1$. Then,
	\begin{equation}
	\Re\left(e^{-\I\omega} P_\omega^\scp{d}(x) \right) = T_2\left(\Re\left(e^{-\I\omega} P_\omega^\scp{d/2}(x) \right)\right) = \cdots = \underbrace{T_2\circ \cdots \circ T_2}_{\log_2(d)}\left(\Re\left(e^{-\I\omega} P_\omega^\scp{1}(x) \right)\right)
	\end{equation}
	Using the composition identity of the Chebyshev polynomials $T_n \circ T_m = T_{nm}$, we have $\underbrace{T_2\circ \cdots \circ T_2}_{\log_2(d)} = T_{d}$. On the other hand, when $d=1$, we have
	\begin{equation}
	\begin{split}
	U^\scp{1}(\omega, \arccos(x)) &= e^{\I\omega Z} e^{\I \arccos(x) X} e^{\I\omega Z} = \left(
	\begin{array}{cc}
	e^{2\I\omega} x & \I\sqrt{1-x^2} \\
	\I \sqrt{1-x^2} & e^{-2\I\omega} x
	\end{array}
	\right)\\
	\Rightarrow e^{-\I\omega} P_\omega^\scp{1}(x) &= e^{\I\omega} x,\ Q_\omega^\scp{1}(x) = 1.
	\end{split}
	\end{equation}
	Therefore
	\begin{equation}
	\Re\left(e^{-\I\omega} P_\omega^\scp{d}(x) \right) = T_{d}\left(\left(\cos\omega\right) x\right).
	\end{equation}
	Furthermore, $Q_\omega^\scp{d}$ and $\Im\left(e^{-\I\omega} P_\omega^\scp{d} \right)$ can be determined from the recurrence relation in \cref{eqn:recurrence-Q,eqn:recurrence-P}
	\begin{equation}
	Q_\omega^\scp{d}(x) = d \prod_{j=0}^{\log_2(d)-1} T_{2^j}\left(\left(\cos\omega\right) x\right),\ \Im\left(e^{-\I\omega} P_\omega^\scp{d}(x) \right) = Q_\omega^\scp{d}(x) \left(\sin\omega\right) x.
	\end{equation}
	For convenience, let $\cos \sigma := \left(\cos \omega\right) x = \cos\omega \cos\theta$. Then
	\begin{equation}
	\begin{split}
	Q_\omega^\scp{d}(x) \sin\sigma &= \left(\frac{d}{2} \prod_{j=1}^{\log_2(d)-1}\right) 2 \cos \sigma \sin\sigma = \left(\frac{d}{4} \prod_{j=2}^{\log_2(d)-1}\right) 2 \cos(2\sigma) \sin(2\sigma)\\
	&= \cdots = 2 \cos\left(\frac{d}{2} \sigma\right) \sin\left(\frac{d}{2} \sigma\right) = \sin\left(d \sigma\right).
	\end{split}
	\end{equation}
	Therefore
	\begin{equation}
	P_\omega^\scp{d}(x) = e^{\I\omega} \left(\cos\left(d \sigma\right) + \I \frac{\sin\left(d \sigma\right)}{\sin\sigma} \left(\sin \omega\right) x\right), \text{ and } Q_\omega^\scp{d}(x) = \frac{\sin\left(d \sigma\right)}{\sin\sigma}.
	\end{equation}
\end{proof}

\section{Analysis of periodic calibration: variance lower bound and shortcomings}\label{app:periodic-calibration}
\subsection{Overview of the methodology of periodic calibration}
In this subsection, we provide an overview of periodic calibration, also known as Floquet calibration, which is proposed in \cite[Appendix C]{arute_observation_2020} and \cite[Appendix A]{neill_accurately_2021}. Periodic calibration is a generalization of the robust single-qubit gate calibration \cite{kimmel2015robust} to entangling gates. The main component of the quantum circuit for performing periodic calibration is equivalent to the periodic part used in our method (see the shaded part in Figure 1 in the main text). However, instead of the initialization in terms of Bell states, periodic calibration measures the transition probability between tensor product states $\ket{01}$ and $\ket{10}$, namely, between logical quantum state $\ket{0_\ell}$ and $\ket{1_\ell}$. The exact parametric expression of this probability can be derived from the results presented in \cref{sec:analytical-result-qspc}, which is consistent with that in literature (upon convention difference):
\begin{equation}
    P_\mathrm{pc}(\theta, \varphi, \omega, d) = \sin^2(\theta) \frac{\sin^2(d \sigma)}{\sin^2(\sigma)} = \sin^2(\theta) U_{d - 1}^2(\cos(\sigma)) \text{ where } \sigma = \arccos(\cos(\theta) \cos(\omega - \varphi)).
\end{equation}
Given the parametric expression, periodic calibration is performed by minimizing the distance metric between the parametric ansatz and the experimentally measured probabilities. To improve the efficiency, the depth parameter $d$ is chosen to be logarithmically spaced, namely, $d = 1, 2, 4, 8, \cdots$. In practice, multiple modulation angles $\omega$ can be chosen. Yet, for simplicity, and due to the linear additivity of the Fisher information gained from multiple $\omega$ values, we focus on a single variable value of $\omega$ in our analysis. Remarkably, we prove that for each circuit depth, the Fisher information is maximized when $\omega = \varphi$, which is referred to as phase-matching condition. Consequently, the simplified choice of fixed $\omega$ value in the analysis will indeed capture the optimality of the estimation.

\subsection{Optimality analysis using Fisher information and Cram\'{e}r-Rao bound}
According to the procedure outlined in the previous section, we can compute the Fisher information as follows. Suppose $d = 2^L$ is the maximum depth of the experiment, and $M$ is the number of measurement samples in each experiment. For notational simplicity, let the parameter vector be $\Xi = (\xi_1 := \theta, \xi_2 := \varphi)$. Then, the Fisher information matrix is expressed as:
\begin{equation}
    I_{k, k^\prime}(\Xi; \omega, d) = \sum_{j = 0}^{\log_2(d)} \frac{M}{P_\mathrm{pc}(\theta, \varphi, \omega, 2^j)(1 - P_\mathrm{pc}(\theta, \varphi, \omega, 2^j))} \frac{\partial P_\mathrm{pc}(\theta, \varphi, \omega, 2^j)}{\partial \xi_k} \frac{\partial P_\mathrm{pc}(\theta, \varphi, \omega, 2^j)}{\partial \xi_{k^\prime}}.
\end{equation}
Here, the relevant derivatives are
\begin{equation}
    \begin{split}
        & \frac{\partial P_\mathrm{pc}(\theta, \varphi, \omega, d)}{\partial \theta} = 2 \sin(\theta)\cos(\theta) U_{d-1}^2(\cos(\sigma)) - 2 \sin^3(\theta) \cos(\omega - \varphi) U_{d-1}(\cos(\sigma)) U_{d-1}^\prime(\cos(\sigma)),\\
        & \frac{\partial P_\mathrm{pc}(\theta, \varphi, \omega, d)}{\partial \varphi} = 2 \sin^2(\theta) \cos(\theta) \sin(\omega - \varphi) U_{d-1}(\cos(\sigma)) U_{d-1}^\prime(\cos(\sigma)),
    \end{split}
\end{equation}
and 
\begin{equation*}
    U_{d - 1}^\prime(\cos(\sigma)) = - \frac{d \cos(d \sigma) - \cos(\sigma) U_{d - 1}(\cos(\sigma))}{\sin^2(\sigma)}.
\end{equation*}
Using the explicit expressions above, the Fisher information matrix can be numerically evaluated. The estimation variances are lower bounded by the diagonal elements of the inverse Fisher information matrix according to Cram\'{e}r-Rao bound. For simplicity, we absorb the modulation angle $\omega$ into the definition of $\varphi$ angle, namely, setting $\omega = 0$ and $\varphi$ variable. Hence, the phase-matching condition is equivalent to $\varphi = 0$ ($\omega = \varphi$ in the original setting).

When the phase-matching condition is satisfied, the CRLB can be exactly derived. Note that $P_\mathrm{pc}|_{\varphi = 0, d = \ell} = \sin^2(\ell \theta)$. Then, for a fixed circuit depth $\ell$ with $M$ measurement shots, the Fisher information is
\begin{equation*}
    J_\theta(\ell)|_{\varphi = 0} = \frac{M}{\sin^2(\ell \theta) \cos^2(\ell \theta)} (2 \ell \sin(\ell \theta) \cos(\ell \theta))^2 = 4 M \ell^2.
\end{equation*}
As we derived in \cref{eqn:QFI-omega} in \cref{sec:quantum-CRLB}, the quantum Fisher information that a quantum circuit using $\ell$ $U$-gates can maximally contribute is $\mf{F}_{\max} = 4 M \ell^2$. Hence, the coincidence between the classical Fisher information with phase-matching condition and the maximal quantum Fisher information indicates that the optimality is attained at phase-matching because the Fisher information attains its maximum.

Under phase-matching conditions, the total Fisher information is
\begin{equation*}
    I_{1, 1} = \sum_{j = 0}^{\log_2(d)} \frac{M}{\sin^2(2^j\theta) \cos^2(2^j\theta)} 4 \times 2^{2j} \sin^2(2^j\theta) \cos^2(2^j\theta) = 4 M \sum_{j = 0} 2^{2 j} = \frac{4 M (4d^2 - 1)}{3}.
\end{equation*}
Hence, the CRLB on the estimation variance of $\theta$ is
\begin{equation}\label{eqn:phase-matching-variance}
    \mathrm{Var}(\theta)|_\mathrm{phase-matching} \ge I_{1, 1}^{-1} = \frac{3}{4 M (4d^2 - 1)} \approx \frac{3}{16 M d^2}.
\end{equation}

The numerical results are depicted in \cref{fig:crlb-periodic-calibration}. In \cref{fig:exact_crlb_periodic_calibration}, it can be seen that the improvement in the estimation variance with increasingly large depth $d$ is very limited when $\varphi$ angle deviates from zero, namely, the phase-matching condition is violated. To understand the optimal variance scaling of periodic calibration as a function of the depth parameter $d$, we numerically depict the results in \cref{fig:exact_crlb_depth_periodic_calibration}. The numerical result indicates that the variance scales as $1 / d^2$ when the phase-matching condition is satisfied. This is consistent with the exactly derived result in \cref{eqn:phase-matching-variance}. Yet when the phase-matching condition is violated, the variance decays as $1 / d^2$ when the depth parameter is not too large ($d \lesssim 1 / \varphi$). However, when $d$ gets increasingly large, the variance is almost plateaued with very limited decay. 

\begin{figure}[htbp]
    \centering
    \subfigure[CRLB as a function of $\varphi$ angle ($\theta = 10^{-3}$)\label{fig:exact_crlb_periodic_calibration}]{
        \includegraphics[width=.33\columnwidth]{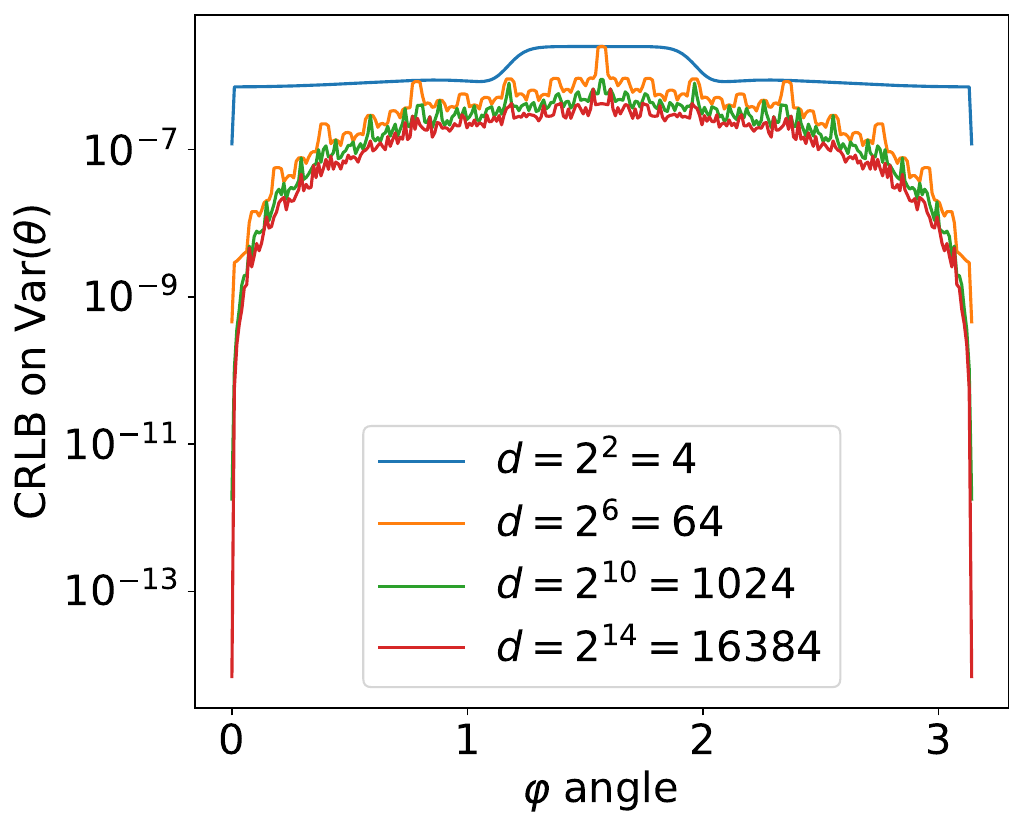}
    }
    \subfigure[CRLB as a function of depth parameter $d$\label{fig:exact_crlb_depth_periodic_calibration}]{
        \includegraphics[width=.64\columnwidth]{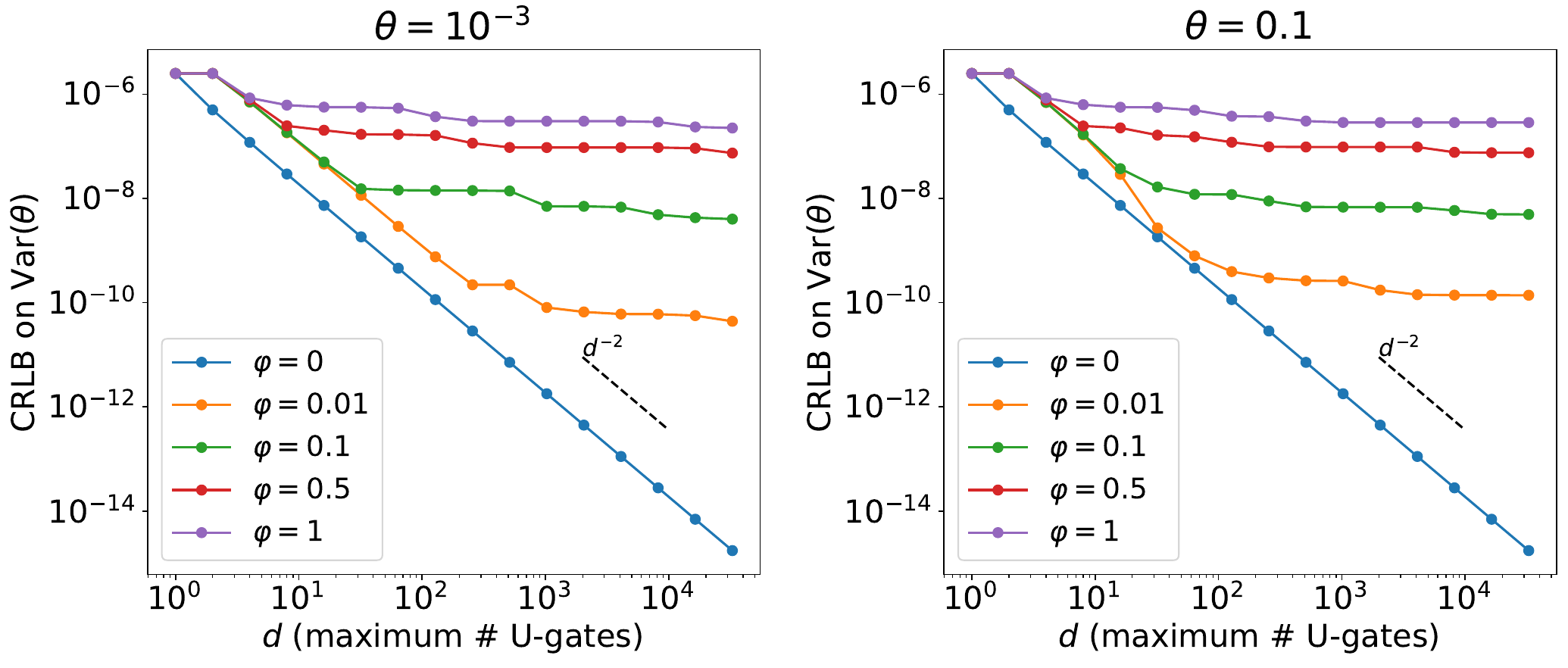}
    }
    \caption{Cram\'{e}r-Rao lower bound (CRLB) on the estimation variance of $\theta$ using periodic calibration. The number of measurement samples is set to $M = 1 \times 10^5$.\label{fig:crlb-periodic-calibration}}
\end{figure}

\subsection{Violation of phase-matching condition implies exponentially worse estimation variance}\label{app:violate-phase-matching-periodic-calibration}
    To understand the importance of the phase-matching condition and the plateau of estimation variance, we provide an approximation to the Fisher information when the swap angle $\theta$ is small. For simplicity, we focus on the quantum circuit with a fixed depth parameter $d$ rather than a series of circuits with variable depths. It is referred to as $J_\theta$ below. We remark that the focus on the fixed-depth Fisher information does not lose generality due to the additivity of Fisher information.

    Note that when $\theta$ is small, to the first order in $\theta$, the Fisher information of $\theta$ is approximately
    \begin{equation}
        J_\theta \approx 4 M \frac{\sin^2(d \varphi)}{\sin^2(\varphi)} + \Or(d^4 \theta^2).
    \end{equation}
    Now, we provide an approximation to the expression above. Note that the extrema of $g(\varphi) := \sin(d \varphi) / \sin(\varphi)$ is attained when $d \tan(\varphi) = \tan(d \varphi)$ is satisfied. Let $\varphi^\star_k$ be the $k$-th solution to the equation involving tangent functions. They belong to the maxima of $g^2$, which are
    \begin{equation*}
        g^2_{\max, k} = g^2(\varphi^\star_k) = \frac{\sin^2(d \varphi^\star_k)}{\sin^2(\varphi^\star_k)} = \frac{1 - \frac{1}{1 + \tan^2(d \varphi^\star_k)}}{1 - \frac{1}{1 + \tan^2(\varphi^\star_k)}} = \frac{1 + \tan^2(\varphi^\star_k)}{d^{-2} + \tan^2(\varphi^\star_k)} \le \frac{1 + (\varphi^\star_k)^2}{d^{-2} + (\varphi^\star_k)^2}.
    \end{equation*}
    By extrapolating this function to other points rather than maxima, we get a reasonably good approximation to the Fisher information:
    \begin{equation}\label{eqn:approximate-Fisher-information}
        J_\theta \lesssim \frac{1 + \varphi^2}{d^{-2} + \varphi^2}. 
    \end{equation}
    Furthermore, when the violation of the phase-matching condition satisfies $\abs{\varphi} \ge \pi / d$, the loss in the Fisher information compared to the maximum is
    \begin{equation}\label{eqn:loss-FI-first-peak}
        \frac{\max_{\abs{\varphi} \ge \pi / d} J_\theta}{ J_\theta|_\mathrm{phase-matching} } \lesssim \frac{1}{d^2} \frac{1 + (\pi / d)^2}{d^{-2} + (\pi / d)^2} = \frac{1 + (\pi / d)^2}{1 + \pi^2} \to \frac{1}{1 + \pi^2} \approx 9.2\%.
    \end{equation}
    Hence, though the phase-matching condition is slightly violated $\abs{\varphi} \ge \pi / d$, which is common when $d$ is large, more than 90\% of the statistical power of $\theta$-estimation is eliminated. 
    
    Furthermore, the plateaued estimation variance is justified. According to \cref{eqn:approximate-Fisher-information}, the total Fisher information is approximated as
    \begin{equation}\label{eqn:loss-FI-O1-varphi}
        I_{1, 1} \lesssim 4 M \sum_{j = 0}^{\log_2(d)} \frac{1 + \varphi^2}{4^{-j} + \varphi^2} \approx 4 M \int_0^{\log_2(d)} \frac{1 + \varphi^2}{4^{-x} + \varphi^2} \rd x = 2 M \frac{1 + \varphi^2}{\varphi^2} \log_2\left( \frac{1 + (d \varphi)^2}{1 + \varphi^2} \right).
    \end{equation}
    This result indicates that when the phase-matching condition is not satisfied, the Fisher information is exponentially diminished as $\Theta(\log(d))$ compared to the optimal $\Theta(d^2)$ Fisher information with the phase-matching condition.

    To conclude, in \cref{eqn:loss-FI-first-peak}, we find that more than 90\% of Fisher information is eliminated even though the phase-matching condition is slightly violated such that $\abs{\varphi} \ge \pi / d$ is outside the principal peak. Moreover, in \cref{eqn:loss-FI-O1-varphi}, we derive that the majority of the violation of phase-matching condition, e.g. when $\varphi$ is constantly large, leads to an exponentially worse estimation variance rather than saturating Heisenberg limit. Due to the sensitivity of $\theta$ estimation to phase-matching conditions, our proposed QSPE method is superior because of the isolation of $\theta$ and $\varphi$ estimations in Fourier space and its consequent robustness against realistic errors.

    In \cref{fig:approx-crlb-periodic-calibration}, we visualize the approximation analysis of Fisher information and CRLB. It can be seen that our approximation well fits the exact values and explains the exponentially worse scaling of $\theta$-estimation variance when the phase-matching condition is violated.

\begin{figure}[htbp]
    \centering
    \subfigure[\label{fig:approx_fisher_info_periodic_calibration}]{
        \includegraphics[width=.5\columnwidth]{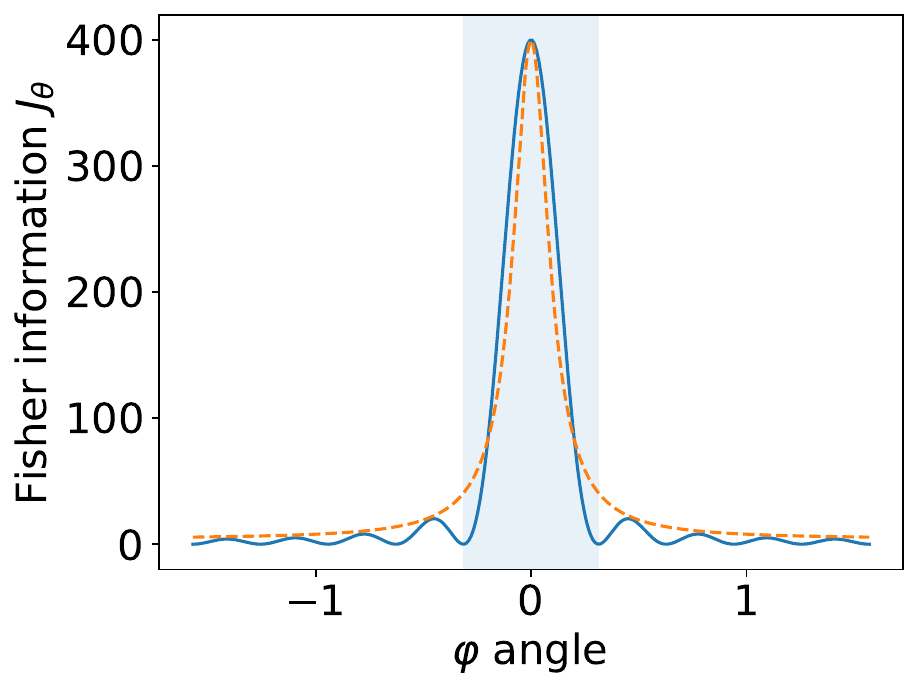}
    }
    \subfigure[\label{fig:approx_crlb_periodic_calibration}]{
        \includegraphics[width=.45\columnwidth]{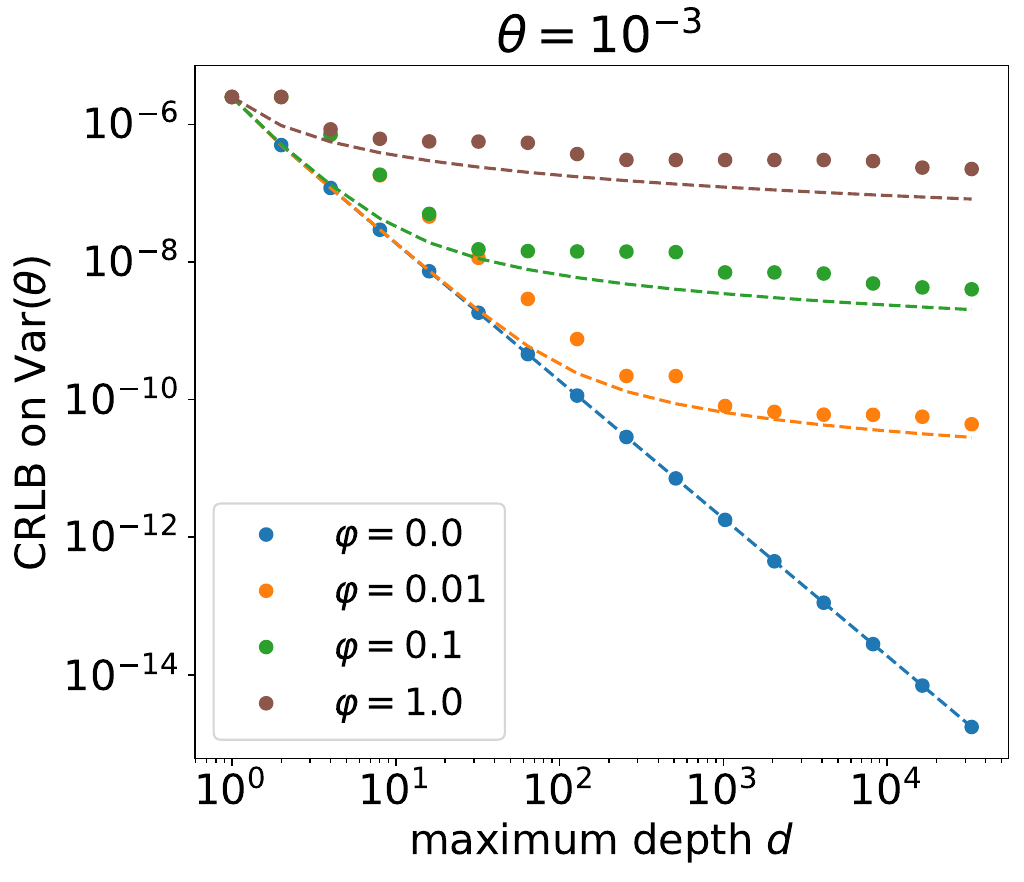}
    }
    \caption{Approximation analysis of Cram\'{e}r-Rao lower bound (CRLB). (a) Fisher information of a fixed depth $J_\theta$ (blue solid line) and its approximation in \cref{eqn:approximate-Fisher-information} (orange dashed line). The shaded area is the principal peak with $\abs{\varphi} \le \pi / d$. We set $d = 10$ and $M = 1$. (b) CRLB (scatters) and its approximation (dashed lines) based on \cref{eqn:approximate-Fisher-information}. We set $M = 1 \times 10^5$. \label{fig:approx-crlb-periodic-calibration}}
\end{figure}

 We summarize our findings in the following theorem in which for completeness, we explicitly include the phase modulation angle $\omega$ as per our original formalism.
    \begin{theorem}\label{thm:phase-matching-lower_bound}
        Let $d$ be the depth parameter of periodic calibration and $M$ be the number of measurement samples in each experiment. When the swap angle $\theta$ is small, the optimal estimation variance that periodic calibration can achieve is lower bounded as follows.
        \begin{enumerate}
            \item When phase-matching condition is satisfied $\omega = \varphi$, the optimal variance is lower bounded by $\Omega(1 / (M d^2))$.
            \item When phase-matching condition is constantly violated $\abs{\omega - \varphi} \ge \text{constant}$, the optimal variance is lower bounded by $\Omega(1 / (M \log(d)))$.
        \end{enumerate}
    \end{theorem}

    \subsection{Practical challenges due to the complex optimization landscape}\label{sec:periodic-calibration-landscape}
    In the previous subsections, we analyze the lower bound on the optimal estimation variance that periodic calibration can achieve in an idealized scenario. However, due to the actual need to minimize a loss function, the estimation performance that periodic calibration can achieve highly depends on the optimization landscape. In this subsection, we analyze the practical challenges that prevent periodic calibration from achieving high estimation accuracy by visualizing the optimization landscape. We consider the following mean squared error (MSE) loss function:
    \begin{equation}
        \mc{L}(\theta, \varphi, d, M) = \sum_{j = 0}^{\log_2(d)} \abs{ \sin^2(\theta) U_{2^j - 1}^2(\cos(\theta) \cos(\varphi)) - \hat{P}^\expl_{2^j, M} }^2
    \end{equation}
    where $\hat{P}^\expl_{d, M}$ stands for the experimentally sampled probability using a depth-$d$ circuit with $M$ measurement samples. We numerically study the optimization landscape in \cref{fig:landscape-periodic-calibration}.

    \begin{figure}[htbp]
    \centering
    \subfigure[\label{fig:landscape_varphi_periodic_calibration}]{
        \includegraphics[width=.4\columnwidth]{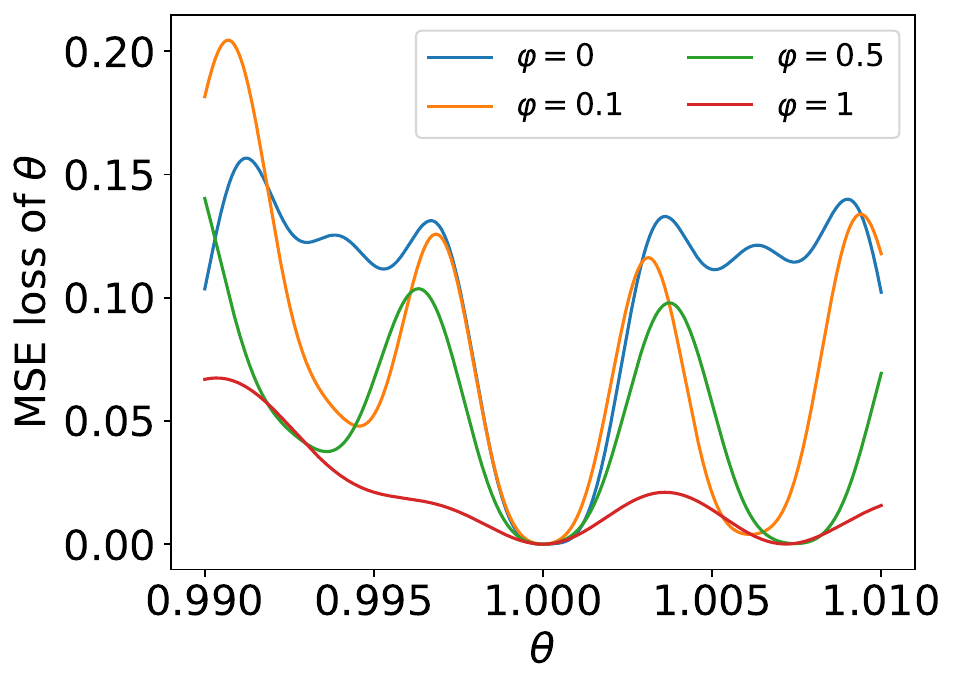}
    }
    \subfigure[\label{fig:landscape_depth_periodic_calibration}]{
        \includegraphics[width=\columnwidth]{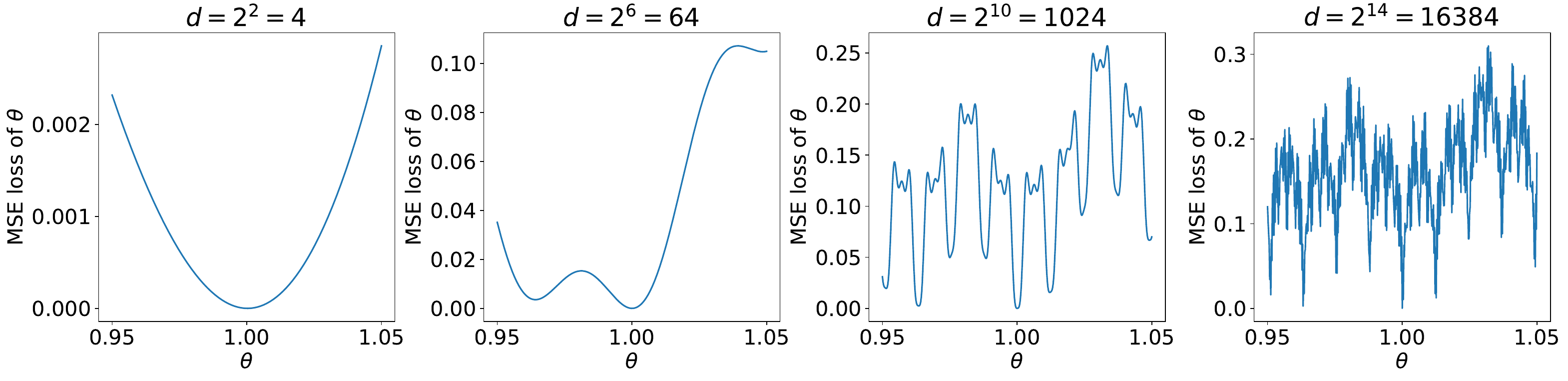}
    }
    \subfigure[\label{fig:landscape_meas_periodic_calibration}]{
        \includegraphics[width=\columnwidth]{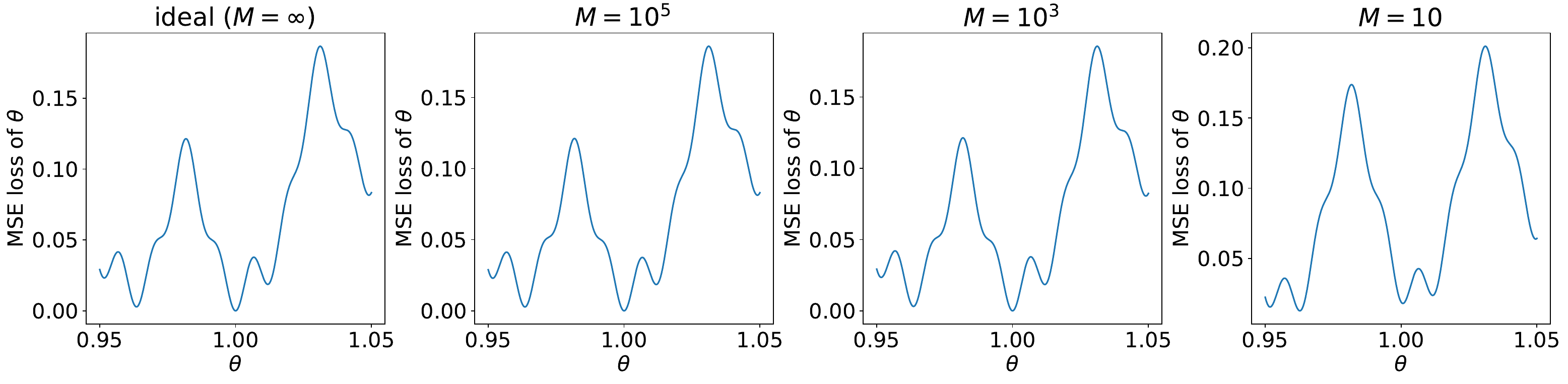}
    }
    \caption{Landscape analysis of periodic calibration. We set the exact value $\theta = 1$ for all figures. (a) Optimization landscape of different $\varphi$ angles. We set the maximum depth $d = 2^10 = 1024$ and consider an ideal case where infinite many measurement samples are used $M = \infty$. (b) Optimization landscape of different maximum depth parameters $d$. We consider an ideal case where infinite many measurement samples are used $M = \infty$, and the phase-matching condition is satisfied $\varphi = 0$. (c) Optimization landscape of different numbers of measurement samples $M$. We set $d = 2^8 = 256$. We consider an ideal case where the phase-matching condition is satisfied $\varphi = 0$.  \label{fig:landscape-periodic-calibration}}
\end{figure}

When the phase-matching condition is violated, we visualize the optimization landscape in \cref{fig:landscape_varphi_periodic_calibration}. It can be seen that as $\varphi$ becomes large, the local minima become increasingly influential. Moreover, some local minima attain an almost vanishing loss value, which makes the optimization challenging. 

    In \cref{fig:landscape_depth_periodic_calibration}, we visualize the optimization landscape of variable depth parameters. It indicates that the optimization landscape becomes increasingly complex when $d$ is larger. The landscape becomes highly non-convex and has increasingly many local minima. As the local landscape becomes sharper and the local gradient is increasingly large, it is hard for the optimizer to proceed with a reasonable step size. This renders the practical optimization-based solution hard to obtain.

    In \cref{fig:landscape_meas_periodic_calibration}, we see that the optimization landscape is not very sensitive to the number of measurement samples. Although the landscape deviates from the true one when $M = 10$ is too small (rightmost panel), the optimization landscape almost faithfully reproduces the one derived from the ideal case (leftmost panel) when $M \ge 1000$. Hence, the hardness of optimization is less affected by the number of measurement samples but mainly bottlenecked by the exponential disadvantage as circuit depth increases when the phase matching condition is not met.

	\bibliographystyle{abbrvurl}

\begin{thebibliography}{10}

\bibitem{acharya2024quantum}
R.~Acharya, L.~Aghababaie-Beni, I.~Aleiner, T.~I. Andersen, M.~Ansmann, F.~Arute, K.~Arya, A.~Asfaw, N.~Astrakhantsev, J.~Atalaya, et~al.
\newblock Quantum error correction below the surface code threshold.
\newblock {\em arXiv preprint arXiv:2408.13687}, 2024.

\bibitem{acharya2022suppressing}
R.~Acharya, I.~Aleiner, R.~Allen, T.~I. Andersen, M.~Ansmann, F.~Arute, K.~Arya, A.~Asfaw, J.~Atalaya, R.~Babbush, D.~Bacon, J.~C. Bardin, J.~Basso, A.~Bengtsson, S.~Boixo, G.~Bortoli, A.~Bourassa, J.~Bovaird, L.~Brill, M.~Broughton, B.~B. Buckley, D.~A. Buell, T.~Burger, B.~Burkett, N.~Bushnell, Y.~Chen, Z.~Chen, B.~Chiaro, J.~Cogan, R.~Collins, P.~Conner, W.~Courtney, A.~L. Crook, B.~Curtin, D.~M. Debroy, A.~Del Toro~Barba, S.~Demura, A.~Dunsworth, D.~Eppens, C.~Erickson, L.~Faoro, E.~Farhi, R.~Fatemi, L.~Flores~Burgos, E.~Forati, A.~G. Fowler, B.~Foxen, W.~Giang, C.~Gidney, D.~Gilboa, M.~Giustina, A.~Grajales~Dau, J.~A. Gross, S.~Habegger, M.~C. Hamilton, M.~P. Harrigan, S.~D. Harrington, O.~Higgott, J.~Hilton, M.~Hoffmann, S.~Hong, T.~Huang, A.~Huff, W.~J. Huggins, L.~B. Ioffe, S.~V. Isakov, J.~Iveland, E.~Jeffrey, Z.~Jiang, C.~Jones, P.~Juhas, D.~Kafri, K.~Kechedzhi, J.~Kelly, T.~Khattar, M.~Khezri, M.~Kieferov{\'a}, S.~Kim, A.~Kitaev, P.~V. Klimov, A.~R. Klots, A.~N. Korotkov, F.~Kostritsa, J.~M.
  Kreikebaum, D.~Landhuis, P.~Laptev, K.-M. Lau, L.~Laws, J.~Lee, K.~Lee, B.~J. Lester, A.~Lill, W.~Liu, A.~Locharla, E.~Lucero, F.~D. Malone, J.~Marshall, O.~Martin, J.~R. McClean, T.~McCourt, M.~McEwen, A.~Megrant, B.~Meurer~Costa, X.~Mi, K.~C. Miao, M.~Mohseni, S.~Montazeri, A.~Morvan, E.~Mount, W.~Mruczkiewicz, O.~Naaman, M.~Neeley, C.~Neill, A.~Nersisyan, H.~Neven, M.~Newman, J.~H. Ng, A.~Nguyen, M.~Nguyen, M.~Y. Niu, T.~E. O'Brien, A.~Opremcak, J.~Platt, A.~Petukhov, R.~Potter, L.~P. Pryadko, C.~Quintana, P.~Roushan, N.~C. Rubin, N.~Saei, D.~Sank, K.~Sankaragomathi, K.~J. Satzinger, H.~F. Schurkus, C.~Schuster, M.~J. Shearn, A.~Shorter, V.~Shvarts, J.~Skruzny, V.~Smelyanskiy, W.~C. Smith, G.~Sterling, D.~Strain, M.~Szalay, A.~Torres, G.~Vidal, B.~Villalonga, C.~Vollgraff~Heidweiller, T.~White, C.~Xing, Z.~J. Yao, P.~Yeh, J.~Yoo, G.~Young, A.~Zalcman, Y.~Zhang, N.~Zhu, and G.~Q. AI.
\newblock Suppressing quantum errors by scaling a surface code logical qubit.
\newblock {\em Nature}, 614(7949):676--681, 2023.
\newblock \href {https://doi.org/10.1038/s41586-022-05434-1} {\path{doi:10.1038/s41586-022-05434-1}}.

\bibitem{arute_observation_2020}
F.~Arute, K.~Arya, R.~Babbush, D.~Bacon, J.~C. Bardin, R.~Barends, A.~Bengtsson, S.~Boixo, M.~Broughton, B.~B. Buckley, D.~A. Buell, B.~Burkett, N.~Bushnell, Y.~Chen, Z.~Chen, Y.-A. Chen, B.~Chiaro, R.~Collins, S.~J. Cotton, W.~Courtney, S.~Demura, A.~Derk, A.~Dunsworth, D.~Eppens, T.~Eckl, C.~Erickson, E.~Farhi, A.~Fowler, B.~Foxen, C.~Gidney, M.~Giustina, R.~Graff, J.~A. Gross, S.~Habegger, M.~P. Harrigan, A.~Ho, S.~Hong, T.~Huang, W.~Huggins, L.~B. Ioffe, S.~V. Isakov, E.~Jeffrey, Z.~Jiang, C.~Jones, D.~Kafri, K.~Kechedzhi, J.~Kelly, S.~Kim, P.~V. Klimov, A.~N. Korotkov, F.~Kostritsa, D.~Landhuis, P.~Laptev, M.~Lindmark, E.~Lucero, M.~Marthaler, O.~Martin, J.~M. Martinis, A.~Marusczyk, S.~McArdle, J.~R. McClean, T.~McCourt, M.~McEwen, A.~Megrant, C.~Mejuto-Zaera, X.~Mi, M.~Mohseni, W.~Mruczkiewicz, J.~Mutus, O.~Naaman, M.~Neeley, C.~Neill, H.~Neven, M.~Newman, M.~Y. Niu, T.~E. O'Brien, E.~Ostby, B.~Pató, A.~Petukhov, H.~Putterman, C.~Quintana, J.-M. Reiner, P.~Roushan, N.~C. Rubin, D.~Sank, K.~J.
  Satzinger, V.~Smelyanskiy, D.~Strain, K.~J. Sung, P.~Schmitteckert, M.~Szalay, N.~M. Tubman, A.~Vainsencher, T.~White, N.~Vogt, Z.~J. Yao, P.~Yeh, A.~Zalcman, and S.~Zanker.
\newblock Observation of separated dynamics of charge and spin in the {Fermi}-{Hubbard} model.
\newblock {\em arXiv:2010.07965 [quant-ph]}, Oct. 2020.
\newblock arXiv: 2010.07965.
\newblock URL: \url{http://arxiv.org/abs/2010.07965}.

\bibitem{GoogleQuantumSupremacy2019}
F.~Arute, K.~Arya, R.~Babbush, D.~Bacon, J.~C. Bardin, R.~Barends, R.~Biswas, S.~Boixo, F.~G. Brandao, D.~A. Buell, et~al.
\newblock Quantum supremacy using a programmable superconducting processor.
\newblock {\em Nature}, 574(7779):505--510, 2019.

\bibitem{BoixoIsakovSmelyanskiyEtAl2018}
S.~Boixo, S.~V. Isakov, V.~N. Smelyanskiy, R.~Babbush, N.~Ding, Z.~Jiang, M.~J. Bremner, J.~M. Martinis, and H.~Neven.
\newblock Characterizing quantum supremacy in near-term devices.
\newblock {\em Nature Physics}, 14(6):595--600, 2018.

\bibitem{BraunsteinCaves1994}
S.~L. Braunstein and C.~M. Caves.
\newblock Statistical distance and the geometry of quantum states.
\newblock {\em Physical Review Letters}, 72(22):3439, 1994.

\bibitem{DingLin2023}
Z.~Ding and L.~Lin.
\newblock Even shorter quantum circuit for phase estimation on early fault-tolerant quantum computers with applications to ground-state energy estimation.
\newblock {\em PRX Quantum}, 4(2):020331, 2023.

\bibitem{Dong2023Dissertation}
Y.~Dong.
\newblock {\em Quantum Signal Processing Algorithm and Its Applications}.
\newblock PhD thesis, UC Berkeley, 2023.

\bibitem{DongMengWhaleyEtAl2020}
Y.~Dong, X.~Meng, K.~B. Whaley, and L.~Lin.
\newblock Efficient phase-factor evaluation in quantum signal processing.
\newblock {\em Physical Review A}, 103(4):042419, 2021.

\bibitem{foxen2020}
B.~Foxen, C.~Neill, A.~Dunsworth, P.~Roushan, B.~Chiaro, A.~Megrant, J.~Kelly, Z.~Chen, K.~Satzinger, R.~Barends, F.~Arute, K.~Arya, R.~Babbush, D.~Bacon, J.~C. Bardin, S.~Boixo, D.~Buell, B.~Burkett, Y.~Chen, R.~Collins, E.~Farhi, A.~Fowler, C.~Gidney, M.~Giustina, R.~Graff, M.~Harrigan, T.~Huang, S.~V. Isakov, E.~Jeffrey, Z.~Jiang, D.~Kafri, K.~Kechedzhi, P.~Klimov, A.~Korotkov, F.~Kostritsa, D.~Landhuis, E.~Lucero, J.~McClean, M.~McEwen, X.~Mi, M.~Mohseni, J.~Y. Mutus, O.~Naaman, M.~Neeley, M.~Niu, A.~Petukhov, C.~Quintana, N.~Rubin, D.~Sank, V.~Smelyanskiy, A.~Vainsencher, T.~C. White, Z.~Yao, P.~Yeh, A.~Zalcman, H.~Neven, and J.~M. Martinis.
\newblock Demonstrating a continuous set of two-qubit gates for near-term quantum algorithms.
\newblock {\em Phys. Rev. Lett.}, 125:120504, Sep 2020.
\newblock URL: \url{https://link.aps.org/doi/10.1103/PhysRevLett.125.120504}, \href {https://doi.org/10.1103/PhysRevLett.125.120504} {\path{doi:10.1103/PhysRevLett.125.120504}}.

\bibitem{GilyenSuLowEtAl2019}
A.~Gily{\'e}n, Y.~Su, G.~H. Low, and N.~Wiebe.
\newblock Quantum singular value transformation and beyond: exponential improvements for quantum matrix arithmetics.
\newblock In {\em Proceedings of the 51st Annual ACM SIGACT Symposium on Theory of Computing}, pages 193--204. ACM, 2019.

\bibitem{Lloyd2006}
V.~Giovannetti, S.~Lloyd, and L.~Maccone.
\newblock Quantum metrology.
\newblock {\em Phys. Rev. Lett.}, 96:010401, Jan 2006.
\newblock URL: \url{https://link.aps.org/doi/10.1103/PhysRevLett.96.010401}, \href {https://doi.org/10.1103/PhysRevLett.96.010401} {\path{doi:10.1103/PhysRevLett.96.010401}}.

\bibitem{gross_characterizing_2024}
J.~A. Gross, E.~Genois, D.~M. Debroy, Y.~Zhang, W.~Mruczkiewicz, Z.-P. Cian, and Z.~Jiang.
\newblock Characterizing coherent errors using matrix-element amplification.
\newblock {\em npj Quantum Information}, 10(1):123, Nov. 2024.
\newblock arXiv:2404.12550 [quant-ph].
\newblock URL: \url{http://arxiv.org/abs/2404.12550}, \href {https://doi.org/10.1038/s41534-024-00917-7} {\path{doi:10.1038/s41534-024-00917-7}}.

\bibitem{Kay1989}
S.~Kay.
\newblock A fast and accurate single frequency estimator.
\newblock {\em IEEE Transactions on Acoustics, Speech, and Signal Processing}, 37(12):1987--1990, 1989.

\bibitem{KeenerTheoreticalStatistics2010}
R.~W. Keener.
\newblock {\em Theoretical statistics: Topics for a core course}.
\newblock Springer, 2010.

\bibitem{kimmel2015robust}
S.~Kimmel, G.~H. Low, and T.~J. Yoder.
\newblock Robust calibration of a universal single-qubit gate set via robust phase estimation.
\newblock {\em Physical Review A}, 92(6):062315, 2015.

\bibitem{PhysRevA.77.012307}
E.~Knill, D.~Leibfried, R.~Reichle, J.~Britton, R.~B. Blakestad, J.~D. Jost, C.~Langer, R.~Ozeri, S.~Seidelin, and D.~J. Wineland.
\newblock Randomized benchmarking of quantum gates.
\newblock {\em Phys. Rev. A}, 77:012307, Jan 2008.
\newblock URL: \url{https://link.aps.org/doi/10.1103/PhysRevA.77.012307}, \href {https://doi.org/10.1103/PhysRevA.77.012307} {\path{doi:10.1103/PhysRevA.77.012307}}.

\bibitem{KullGuerinVerstraete2020}
I.~Kull, P.~A. Gu{\'e}rin, and F.~Verstraete.
\newblock Uncertainty and trade-offs in quantum multiparameter estimation.
\newblock {\em Journal of Physics A: Mathematical and Theoretical}, 53(24):244001, 2020.

\bibitem{LowChuang2017}
G.~H. Low and I.~L. Chuang.
\newblock Optimal hamiltonian simulation by quantum signal processing.
\newblock {\em Physical review letters}, 118(1):010501, 2017.

\bibitem{LowChuang2019}
G.~H. Low and I.~L. Chuang.
\newblock Hamiltonian simulation by qubitization.
\newblock {\em Quantum}, 3:163, 2019.

\bibitem{PhysRevLett.106.180504}
E.~Magesan, J.~M. Gambetta, and J.~Emerson.
\newblock Scalable and robust randomized benchmarking of quantum processes.
\newblock {\em Phys. Rev. Lett.}, 106:180504, May 2011.
\newblock URL: \url{https://link.aps.org/doi/10.1103/PhysRevLett.106.180504}, \href {https://doi.org/10.1103/PhysRevLett.106.180504} {\path{doi:10.1103/PhysRevLett.106.180504}}.

\bibitem{PhysRevA.85.042311}
E.~Magesan, J.~M. Gambetta, and J.~Emerson.
\newblock Characterizing quantum gates via randomized benchmarking.
\newblock {\em Phys. Rev. A}, 85:042311, Apr 2012.
\newblock URL: \url{https://link.aps.org/doi/10.1103/PhysRevA.85.042311}, \href {https://doi.org/10.1103/PhysRevA.85.042311} {\path{doi:10.1103/PhysRevA.85.042311}}.

\bibitem{Markov1890}
A.~A. Markov.
\newblock On a question by di mendeleev.
\newblock {\em Zapiski Imperatorskoi Akademii Nauk}, 62(1-24):12, 1890.

\bibitem{martinis2003decoherence}
J.~M. Martinis, S.~Nam, J.~Aumentado, K.~Lang, and C.~Urbina.
\newblock Decoherence of a superconducting qubit due to bias noise.
\newblock {\em Physical Review B}, 67(9):094510, 2003.

\bibitem{martyn2021grand}
J.~M. Martyn, Z.~M. Rossi, A.~K. Tan, and I.~L. Chuang.
\newblock Grand unification of quantum algorithms.
\newblock {\em PRX Quantum}, 2(4):040203, 2021.

\bibitem{neill_accurately_2021}
C.~Neill, T.~McCourt, X.~Mi, Z.~Jiang, M.~Y. Niu, W.~Mruczkiewicz, I.~Aleiner, F.~Arute, K.~Arya, J.~Atalaya, R.~Babbush, J.~C. Bardin, R.~Barends, A.~Bengtsson, A.~Bourassa, M.~Broughton, B.~B. Buckley, D.~A. Buell, B.~Burkett, N.~Bushnell, J.~Campero, Z.~Chen, B.~Chiaro, R.~Collins, W.~Courtney, S.~Demura, A.~R. Derk, A.~Dunsworth, D.~Eppens, C.~Erickson, E.~Farhi, A.~G. Fowler, B.~Foxen, C.~Gidney, M.~Giustina, J.~A. Gross, M.~P. Harrigan, S.~D. Harrington, J.~Hilton, A.~Ho, S.~Hong, T.~Huang, W.~J. Huggins, S.~V. Isakov, M.~Jacob-Mitos, E.~Jeffrey, C.~Jones, D.~Kafri, K.~Kechedzhi, J.~Kelly, S.~Kim, P.~V. Klimov, A.~N. Korotkov, F.~Kostritsa, D.~Landhuis, P.~Laptev, E.~Lucero, O.~Martin, J.~R. McClean, M.~McEwen, A.~Megrant, K.~C. Miao, M.~Mohseni, J.~Mutus, O.~Naaman, M.~Neeley, M.~Newman, T.~E. O’Brien, A.~Opremcak, E.~Ostby, B.~Pató, A.~Petukhov, C.~Quintana, N.~Redd, N.~C. Rubin, D.~Sank, K.~J. Satzinger, V.~Shvarts, D.~Strain, M.~Szalay, M.~D. Trevithick, B.~Villalonga, T.~C. White, Z.~Yao,
  P.~Yeh, A.~Zalcman, H.~Neven, S.~Boixo, L.~B. Ioffe, P.~Roushan, Y.~Chen, and V.~Smelyanskiy.
\newblock Accurately computing the electronic properties of a quantum ring.
\newblock {\em Nature}, 594(7864):508--512, June 2021.
\newblock URL: \url{https://www.nature.com/articles/s41586-021-03576-2}, \href {https://doi.org/10.1038/s41586-021-03576-2} {\path{doi:10.1038/s41586-021-03576-2}}.

\bibitem{neill2017path}
C.~J. Neill.
\newblock {\em A path towards quantum supremacy with superconducting qubits}.
\newblock University of California, Santa Barbara, 2017.

\bibitem{NiLiYing2023}
H.~Ni, H.~Li, and L.~Ying.
\newblock On low-depth algorithms for quantum phase estimation.
\newblock {\em Quantum}, 7:1165, 2023.

\bibitem{niu2019universal}
M.~Y. Niu, S.~Boixo, V.~N. Smelyanskiy, and H.~Neven.
\newblock Universal quantum control through deep reinforcement learning.
\newblock {\em npj Quantum Information}, 5(1):1--8, 2019.

\bibitem{google2020hartree}
G.~A. Quantum, Collaborators*†, F.~Arute, K.~Arya, R.~Babbush, D.~Bacon, J.~C. Bardin, R.~Barends, S.~Boixo, M.~Broughton, B.~B. Buckley, et~al.
\newblock Hartree-fock on a superconducting qubit quantum computer.
\newblock {\em Science}, 369(6507):1084--1089, 2020.

\bibitem{rao2008cramer}
C.~R. Rao.
\newblock Cram{\'e}r-rao bound.
\newblock {\em Scholarpedia}, 3(8):6533, 2008.

\bibitem{Niu_Sinha_2023}
G.~Research.
\newblock Google-research/qsp\_quantum\_metrology at master · google-research/google-research, Jun 2023.
\newblock URL: \url{{https://github.com/google-research/google-research/tree/master/qsp_quantum_metrology}}.

\bibitem{ShenLiu2019}
Z.~Shen and R.~Liu.
\newblock Efficient and accurate frequency estimator under low {SNR} by phase unwrapping.
\newblock {\em Mathematical Problems in Engineering}, 2019, 2019.

\bibitem{TangTian2024}
E.~Tang and K.~Tian.
\newblock A cs guide to the quantum singular value transformation.
\newblock In {\em 2024 Symposium on Simplicity in Algorithms (SOSA)}, pages 121--143. SIAM, 2024.

\bibitem{trefethen2019approximation}
L.~N. Trefethen.
\newblock {\em Approximation theory and approximation practice, extended edition}.
\newblock SIAM, 2019.

\bibitem{Tretter1985}
S.~Tretter.
\newblock Estimating the frequency of a noisy sinusoid by linear regression (corresp.).
\newblock {\em IEEE Transactions on Information theory}, 31(6):832--835, 1985.

\bibitem{WangDongLin2021}
J.~Wang, Y.~Dong, and L.~Lin.
\newblock On the energy landscape of symmetric quantum signal processing.
\newblock {\em Quantum}, 6:850, 2022.

\bibitem{wudarski2023characterizing}
F.~Wudarski, Y.~Zhang, A.~N. Korotkov, A.~Petukhov, and M.~Dykman.
\newblock Characterizing low-frequency qubit noise.
\newblock {\em Physical Review Applied}, 19(6):064066, 2023.

\end{thebibliography}

\end{document}